\documentclass[a4paper,12pt]{article}

\usepackage[english]{babel}
\usepackage{authblk}
\usepackage{dsfont}
\usepackage{amsfonts}
\usepackage{mathrsfs}
\usepackage{amssymb}        
\usepackage{amsmath}
\usepackage{graphicx,caption}
\usepackage{float}
\usepackage[a4paper]{geometry} 
\usepackage{verbatim}
\usepackage{pstricks}
\usepackage{amsthm}
\usepackage{mathrsfs}
\usepackage{enumitem}

\numberwithin{equation}{section}

\newtheorem{Theorem}{Theorem}
\newtheorem{Proposition}{Proposition}
\newtheorem{Corollary}{Corollary}
\newtheorem{Lemma}{Lemma}
\newtheorem{Definition}{Definition}

\newtheorem{Assumption}{Assumption}

\newcommand{\R}{\mathbb{R}}

\newcommand{\Z}{\mathbb{Z}}
\newcommand{\N}{\mathbb{N}}

\newcommand{\dd}{\mathrm{d}}

\newcommand{\ed}{\mathrm{e}}

\newcommand{\Const}{\mathrm{C}}
\newcommand{\const}{c}
\newcommand{\E}{\mathsf E}
\newcommand{\Proba}{\mathsf P}
\newcommand{\imag}{\mathrm i}
\newcommand{\dist}{\mathrm{dist}}

\newcommand{\har}{{\mathrm{har}}}
\newcommand{\an}{{\mathrm{an}}}

\newcommand{\loccen}{{\mathbf x}}
\newcommand{\interval}{{\mathcal I}}

\newcommand{\lscal}{[}
\newcommand{\rscal}{]}

\begin{document}

\title{Long Persistence of Localization in a Disordered Anharmonic Chain Beyond the Atomic Limit}

\author[1]{Wojciech De Roeck}
\author[2]{François Huveneers}
\author[3]{Oskar A. Prośniak}
\affil[1]{K.U.Leuven University, Leuven 3000, Belgium}
\affil[2]{Department of Mathematics, King’s College London, Strand, London WC2R 2LS, United Kingdom}
\affil[3]{Department of Physics and Materials Science, University of Luxembourg, L-1511 Luxembourg, G.\ D.\ Luxembourg}
\date{\today}
\maketitle 

\begin{abstract}
    We establish rigorous bounds on the decorrelation time and thermal transport in the disordered Klein-Gordon chain with a quartic on-site potential, governed by a parameter $\lambda$. At $\lambda = 0$, the chain is harmonic, and any form of transport is fully suppressed by Anderson localization. For the anharmonic system, at $\lambda > 0$, our results show that decorrelation and transport can occur only on time scales that grow faster than any polynomial in $1/\lambda$ as $\lambda \to 0$.
    From a technical perspective, the main novelty of our work is that we don't restrict ourselves to the atomic limit. Instead, we develop perturbation theory around the harmonic system with a fixed harmonic interaction between nearby oscillators. This allows us to compare our mathematical results with previous numerical work and contribute to resolving an ongoing debate, as detailed in a companion paper \cite{ourphysicspaper}. 
\end{abstract}

\section{Introduction}

The fate of Anderson localization in the presence of genuine many-body interactions, or anharmonicity, is a matter of considerable interest and debate, 
both from a mathematical and physical point of view. 
For extensive classical Hamiltonian systems at positive temperature, in which the total energy scales extensively, it is generally believed that a small, generic interaction destroys localization and integrability, leading to hydrodynamic behavior with either diffusive or superdiffusive transport. 

Before proceeding, we note that the above restriction to positive temperature is important. Indeed, if we consider a system close to the ground state, i.e.\@ close to the configuration of minimal energy, then the phenomenology is different and KAM tori can persist in the presence of small perturbations, see \cite{frohlich1986localization,johansson2010kam}.
It is also important to consider sufficiently generic interactions, see e.g.\  \cite{WDR_Huveneers_Olla} for an example of a many-body system that, despite probably being chaotic, exhibits subdiffusive transport. 

Our paper is not concerned with the fundamental yet uncontested question of whether chaoticity and normal hydrodynamic behaviour are restored at small interaction. 
Instead, we do investigate how much the dynamics at positive temperature and at weak interaction is slowed down by its proximity to a localized system. We will explain later how this connects to an interesting ongoing debate, but first, we introduce the model and set the stage. 

\subsection{Model}
The disordered Klein-Gordon chain with a quartic interaction is a prototypical example of an interacting classical many-body system, 
intensively studied in the physical literature. 
Newton's equations of motion read
\begin{equation}\label{eq: newton}
	\ddot q_x \; = \; - \omega_x^2 q_x + \eta (\Delta q)_x - \gamma q_x^3, \qquad x \in \Lambda_L
\end{equation}
with $\Lambda_L = [-L,L]\cap \Z$ for $L>0$ and 
\begin{enumerate}
    \item $\omega_x^2 > 0$ i.i.d.\@ random variables with a smooth compactly supported distribution bounded away from $0$,
    \item $\Delta$ the lattice Laplacian with free boundary conditions,  
    \item $\eta > 0$ the harmonic coupling, 
and $\gamma \ge 0$ the anharmonic coupling.
\end{enumerate}
We are interested in the large volume limit $L\to\infty$, but we do not attempt to define the dynamics itself in infinite volume. 
See Section~\ref{sec: model and results} for a more thorough description, 
and notice that $\gamma$ will later on be replaced by $\lambda$, the \emph{effective} strength of the anharmonic coupling at equilibrium, cf.~\eqref{eq: def of lambda} below. 

At $\gamma = 0$, the equations of motion are linear and the system is an Anderson insulator, i.e.\@ all vibration eigenmodes are spatially localized, as they are related to the eigenfunctions of the linear operator $\mathcal H_{\Lambda_L}$ defined in Section \ref{sec: the localized harmonic chain}.
Instead, taking $\gamma > 0$ induces an interaction among the modes and presumably restores ergodicity.  The model can also be cast as a Hamiltonian system, as we explain in Section~\ref{sec: model and results}, and hence there is a locally conserved quantity, namely the energy. 

\subsection{Atomic Limit and Asymptotic Localization}

The mathematical understanding of the long time behavior of this system remains so far out of reach. Some progress has however been made by considering the so-called \emph{atomic limit}, where both the harmonic hopping strength $\eta$ and the interaction strength $\gamma$
are regarded as perturbative parameters.
The appeal of such a limit is clear.
On the one hand, the description of the dynamics is fully explicit at $\eta=\gamma=0$, making it an ideal starting point for perturbative expansions, with the eigenfrequencies of the harmonic modes being simply the i.i.d.\@ random variables $\omega_x$ featuring in \eqref{eq: newton}. 
On the other hand, the size of $\eta$ mainly governs the localization length of the harmonic system 
and presumably does not affect the fundamental characteristics of the long-time behavior of the dynamics, so that the restriction to the atomic limit should not fundamentally alter the long-time behaviour of the system.

In the atomic limit $\eta\propto\gamma \to 0$, 
several mathematical works support the following global picture for systems of the type of \eqref{eq: newton}, cf.~\cite{bourgain_wang_2007,wang2009long,Huveneers_2013,cong_long-time_2021,DeRoeck2013,de2015asymptotic}:
\emph{Delocalization is non-perturbatively slow as a function of $\gamma$, 
i.e.\@ ergodicity is restored on time scales that diverge faster than any power of $1/\gamma$.} 
This phenomenon was called \emph{asymptotic localization} in \cite{DeRoeck2013,de2015asymptotic}.
This picture is also supported by a few numerical studies \cite{oganesyan_pal_huse_2009,kumar_transport_2020}, as well as theoretical works \cite{basko2011weak,fishman2009perturbation,fishman2012nonlinear}, and this support is \emph{not} restricted to the atomic limit.

\subsection{Numerics on Spreading Wavepackets}
Many numerical papers deal with an infinite system with a finite amount of energy: 
They analyze the spreading of an initially localized wave packet, and keep both $\eta$ and $\gamma$ as positive, fixed, parameters. 
The question then is how fast the wavepacket spreads. 
One can make this quantitative in many ways, but one could for example just define the width $w$ of a wavepacket as\footnote{A wiser choice is to replace $q_x^2$ by the local energy, but in practice the given expression will yield the same behaviour} 
$$
w^2 \; = \; \frac{\sum_x x^2 q_x^2}{\sum_x q_x^2}
$$
and ask how it evolves with time. 
The numerical answer seems to be unambiguous: one finds 
\begin{equation}\label{eq: numerical scaling}
    w(t) \;\propto\; t^{1/6},
\end{equation}
see e.g.\ \cite{pikovsky_destruction_2008,garcia-mata_delocalization_2009,flach_universal_2009} as well as references given in \cite{ourphysicspaper}.

Let us explain how, at least at a heuristic level,  
the long-time dynamics of a wave packet can be mapped to the limit of vanishing anharmonicity $\gamma$ or vanishing temperature $\beta^{-1}$, in an extensive system with positive energy density as studied here. 
Before proceeding, we note that these two last limits are equivalent for what matters: 
As we will see in the next section, in appropirate units, the statistics of the time-dependent observables considered here depend only on the product $\lambda = \beta^{-1}\gamma$, which we refer to as the effective anharmonicity, cf.~\eqref{eq: def of lambda} below.
Now, the main point is that, in a spreading experiment, the non-linear term $-\gamma q_x^3$ becomes negligible compared to the linear one in \eqref{eq: newton} as the packet spreads. 
Consequently, the system approaches the linear model, which is known to localize: 
when $\gamma = 0$, the time-dependent width $w$ is almost surely bounded from above. 
It is important to note, however, that the system does not approach the atomic limit, since the harmonic coupling does not scale down as the packet spreads. 
This is a key motivation for this work. 

With this in mind, the scaling law in \eqref{eq: numerical scaling} sharply contrasts with the concept of ``asymptotic localization" described above and with the theorems derived in the present paper: 
We would expect a spreading slower than any polynomial in time, 
see \cite{ourphysicspaper} for further details.

\subsection{Asymptotic Localization Beyond the Atomic Limit}
Our primary goal is to establish rigorous results that bridge the gap between the conjectured framework of ``asymptotic localization'' and the conflicting numerical evidence. To achieve this, we employ a dual strategy.
First, we present mathematical results that extend beyond the atomic limit by keeping $\eta$ fixed and treating $\gamma$ as a perturbative parameter. These results, which form the core of this paper, support the concept of ``asymptotic localization''.
Second, we offer insights into why state-of-the-art numerical simulations may struggle to reach the long-time limit in spreading dynamics.
This paper focuses exclusively on the rigorous results, and we now restrict our attention to this aspect. 
We discuss the numerical analysis and overall conclusions in our companion paper \cite{ourphysicspaper}.

In this paper, we set a controlled expansion around the localized system at $\gamma=0$ and we derive rigorous bounds on the transport of energy and the decorrelation of local observables. For the case of decorrelation rates, our bounds rule out any other behaviour than ``asymptotic localization". For the case of the diffusion constant, our bounds provide a strong hint, but no proof, that ``asymptotic localization" holds. 
Since a precise statement of these results necessarily involves quite some setup, we postpone it to Section \ref{sec: model and results}.  
 
In going beyond the atomic limit, we face new technical difficulties, since the unperturbed system is a full Anderson insulator, not simply a set of uncoupled sites. 
In particular, to control our expansions, we need lower bounds on KAM-like ``denominators'' involving linear combinations of the eigenfrequencies of the linear system: 
\[
    \frac{1}{\sigma_1 \nu_{k_1} + \dots + \sigma_n \nu_{k_n}}
\]
where $\nu_{k_j}^2$ are eigenvalues of the Anderson operator $\mathcal H_{\Lambda_L}$, cf.~\eqref{eq: Anderson hamiltonian} below, 
and where $\sigma_{j} =\pm 1$.

The main difficulty in estimating this denominator is the control of correlations between eigenvalues $\nu_k^2$ that are not necessarily close to each other. Therefore, this problem goes beyond Minami estimates \cite{minami1996local,graf2007remark} and we know of only one paper \cite{klopp2011decorrelation} tackling such correlations, but only for $n=2$. See also \cite{aizenman2008joint}. 
The technique we use in this paper relies crucially on the fact that the above denominator contains $\nu_k$ instead of $\nu_k^2$, see Section~\ref{sec: bound denominators}. 
We note that this approach would thus not have worked if we had considered the nonlinear Schrödinger equation instead of the Klein-Gordon equation, see \cite{fishman2009perturbation,fishman2012nonlinear}.

In addition to addressing this main technical difficulty, our paper develops a systematic and clear perturbative expansion that is well-suited for controlling local observables in extensive systems, see Section~\ref{sec: perturbative expansion}. Furthermore, we introduce a formalism to handle the fact that the eigenfunctions of the Anderson Hamiltonian \(\mathcal H_{\Lambda_L}\) are localized but not strictly local, i.e.\@ are not compactly supported with fixed support, see Sections~\ref{sec: control of local observables}–\ref{sec: properties of the Z random variables}.

\subsection{Organization of the Paper}
In Section~\ref{sec: model and results}, we provide the needed set-up to state our main results: 
Theorems~\ref{th: decorrelation} to \ref{th: Green Kubo} as well as an out-of-equilibrium generalization, Theorem~\ref{th: out of equilibrium}. 
In Section~\ref{sec: proof of the theorems}, we prove Theorems~\ref{th: decorrelation} to \ref{th: Green Kubo} under some assumptions, 
stated as Propositions~\ref{pro: assumptions theorem 1} to \ref{pro: assumptions theorem 3} below, that will be shown afterwards.  
In Section~\ref{sec: perturbative expansion}, we develop a perturbative expansion in $\lambda$ for the observables appearing in our theorems.
This expansion features ``small denominators'' involving energies of the eigenmodes of the chain. 
We provide a probabilistic bound to control them in Section~\ref{sec: bound denominators}. 
Next, in Sections~\ref{sec: control of local observables} and \ref{sec: properties of the Z random variables}, 
we define a special class of random variables and show some of their properties. 
We show in Section~\ref{sec: proof of the three propositions}
that the variables of interest for us are of this type, and we derive probabilistic bounds on the perturbative expansion from the properties established in 
Sections~\ref{sec: control of local observables} and \ref{sec: properties of the Z random variables}. 
Thanks to this, we conclude the proof of Propositions~\ref{pro: assumptions theorem 1} to \ref{pro: assumptions theorem 3}.
We also describe in Section~\ref{sec: proof of the three propositions} the needed adaptations needed to derive Theorem~\ref{th: out of equilibrium}, 
completing thus the proof of this theorem as well. 
Finally, Appendices~\ref{sec: localization} and \ref{sec: gibbs} gather some standard results 
on localization and on the decay of correlation of the Gibbs state respectively.

\subsubsection*{Data Availability Statement}
Data sharing is not applicable to this article as no new data were created or analyzed in this study.

\subsubsection*{Statements and Declarations}
W.D.R. and O.A.P. were supported in part by the FWO (Flemish Research Fund) under grant G098919N.  
Apart from this support, the authors have no relevant financial or non-financial interests to disclose.

\subsubsection*{Acknowledgements}
We sincerely thank the anonymous referee for their careful reading of a first version of this paper and for pointing out numerous inaccuracies; their comments and suggestions greatly improved the quality of this work.
W.D.R.\ is grateful to Alex Elgart for discussions on the problem of Section \ref{sec: bound denominators}.

\section{Model and Results}\label{sec: model and results}

\subsection{Model}\label{subsec: model}
Let $L\in\N^*$, wiht $\N^*$ the set of natural numbers, i.e.\@ positive integers, and $\Lambda_L = [-L,L] \cap \Z$. 
Let the Hamiltonian $\mathsf H:\R^{2|\Lambda_L|}\to \R$ be given by 
\begin{equation}\label{eq: H original}
	\mathsf H(\mathsf q,\mathsf p) 
	\; = \; 
	\sum_{x\in\Lambda_L}\left(  
	\frac{\mathsf p_x^2}{2} + \frac{\mathsf{\omega}_x^2 \mathsf q_x^2}{2} + \frac{\eta}{2} (\mathsf q_x - \mathsf q_{x+1})^2 + \frac{\gamma}{4} \mathsf q_x^4 
	\right)
\end{equation}
with free boundary conditions i.e.\@ $\mathsf q_{L+1}=\mathsf q_L$.
We assume that $\eta, \gamma \ge 0$,  
and that $(\omega_x^2)_{x\in\Z}$ is an i.i.d.\@ sequence of random variables with a smooth density supported on the compact interval $[\omega_-^2,\omega_+^2]$, 
where $0 < \omega_-^2 < \omega_+^2 < + \infty$. 
We denote by $\E(\cdot)$ the average w.r.t.\@ the $\omega_x$. 
The variables $\mathsf q$ and $\mathsf p$ are canonically conjugated and they evolve according to Hamilton's equations: 
$$
	\dot{\mathsf q} \; = \; \nabla_{\mathsf p} \mathsf H, 
	\qquad 
	\dot{\mathsf p} \; = \; -\nabla_{\mathsf q} \mathsf H.
$$

We will principally consider this system in the Gibbs ensemble at inverse temperature $\beta$: 
We assume that $(\mathsf q,\mathsf p)$ are distributed according to the invariant probability measure on $\R^{2|\Lambda_L|}$ with density
\begin{equation}\label{eq: def Gibbs}
    \rho_\beta(\mathsf q,\mathsf p) \; = \; \frac{\ed^{-\beta \mathsf H(\mathsf q,\mathsf p)}}{Z_\beta},
\end{equation}
where $Z_\beta$ is the partition function that ensures normalization: 
$Z_\beta = \int_{\R^{2|\Lambda_L|}} \ed^{-\beta \mathsf{H}(\mathsf q,\mathsf p)}\dd \mathsf q \dd \mathsf p$. 
We denote by $\langle \cdot \rangle_\beta$ the average w.r.t.\@ the equilibrium measure, i.e.\ the Gibbs ensemble.

Let us assume that the harmonic part of the Hamiltonian is fixed, i.e.\ the parameter $\eta$ and the distribution of $\omega_x$ are considered to be fixed throughout the paper.
The statistics of time-evolved observables will then depend on the remaning parameters $\gamma$ and $\beta$.
However, it turns out that in appropriate units they only depend on these parameters through the dimensionless \emph{effective anharmonicity} 
\begin{equation}\label{eq: def of lambda}
	\lambda \; = \; \beta^{-1}\gamma.
\end{equation}
Indeed, let us introduce the dimensionless variables:
$$
	q \; = \; \beta^{1/2} \mathsf q \quad \text{and} \quad p \; = \; \beta^{1/2} \mathsf p.
$$
Writing $\mathsf q(q) = \beta^{-1/2}q$ and $\mathsf p(p) = \beta^{-1/2}p$, 
we introduce also the Hamiltonian $H:\R^{2|\Lambda_L|}\to\R$ defined by 
$$
	H(q,p) \; = \; \beta \mathsf H(\mathsf q(q),\mathsf p(p)).
$$
Then $(q,p)$ satisfy
\begin{equation}\label{eq: Hamilton equations}
	\dot q \; = \; \nabla_p ( H), 
	\qquad 
	\dot p \; = \; - \nabla_q ( H). 
\end{equation}
where
\begin{equation}\label{eq: H}
	 H(q,p) \; = \; \sum_{x\in \Lambda_L} \left(
	\frac{p_x^2}{2} + \frac{\omega_x^2 q_x^2}{2} + \frac{\eta}{2} (q_x - q_{x+1})^2 + \frac{\lambda}{4} q_x^4
	\right).
\end{equation}
From this point forward, we work in these new coordinates, retaining $\lambda$ as the sole variable parameter. The equilibrium measure is denoted simply by $\langle\cdot\rangle$.

\subsection{Localized Harmonic Chain}\label{sec: the localized harmonic chain}
Let us decompose the Hamiltonian in eq.~\eqref{eq: H} as 
$$
	H \; = \; H_\har + \lambda H_\an
$$
where $H_\har$ and $H_\an$ do not depend on $\lambda$, 
and where $\har$ and $\an$ stand for ``harmonic'' and ``anharmonic'' respectively.
We focus here on the dynamics at $\lambda = 0$: The chain is harmonic and the system is decomposed into a set of $|\Lambda_L|$ independent modes.

It will prove useful to consider the harmonic dynamics on more generic subsets than $\Lambda_L$.  
We say that a set $I\subset \Z$ is an \emph{interval} if it is non-empty and if it takes the form $I = I' \cap \Z$ for some interval $I'\subset \R$.
Given a (possibly infinite) interval $I\subset \Z$, 
let us introduce the random operator 
\begin{equation}\label{eq: Anderson hamiltonian}
    \mathcal H_I : L^2(I) \to L^2(I),\quad f \mapsto \mathcal H_I f 
    \; = \;  
    (V - \eta\Delta) f.
\end{equation}
Here $L^2(I)$ is the Hilbert space of square integrable real valued functions on $I$, 
and the ``random potential'' $V$ and the discrete Laplacian $\Delta$ are defined respectively by 
$$
	Vf(x) \; = \; \omega_x^2 f(x), 
	\qquad
	\Delta f(x) \; = \; f(x+1) - 2f(x) + f(x-1) 
$$
for all $f \in L^2(I)$ and $x\in I$, 
with free boundary conditions for $\Delta$: 
$f(-a-1)=f(-a)$ and $f(b+1)=f(b)$, with $a = \min I$ and $b=\max I$
(possibly $a=-\infty$ and $b=+\infty$).
The operator $\mathcal H_I$ is positive and symmetric and its spectrum is almost surely discrete and non-degenerate, see e.g.\@~\cite{aizenman2015random}.  
By our assumptions on the distribution of the disorder, there exist deterministic constants $0 < \nu_-^2 \le \nu_+^2 < + \infty$ such that 
\begin{equation}\label{eq: nu - nu +}
	\nu_-^2  \; \le \; \nu^2 \; \le \;  \nu_+^2 \qquad a.s. 
\end{equation}
for any eigenvalue $\nu^2$ of $\mathcal H_I$. 
Almost surely, the spectrum of $\mathcal{H}_I$ can be written as $(\nu_k^2)_k$, where $k \in \{1, \dots, |I|\}$ if $|I|<\infty$, or $k \in \mathbb{N}^*$ if $|I| = \infty$.
From here on, we will assume that $\{1,\dots,|I|\}$ is interpreated as $\N^*$ whenever $|I|=\infty$.
Given an eigenvalue $\nu_k^2$, we denote by $\psi_k$ the corresponding normalized eigenvector, which is unique up to a phase.

The observable $H_\har$ takes the form 
\begin{equation}\label{eq: H not decomposition}
	H_\har(q,p) 
	\; = \; 
	\frac{1}{2} \big(  \lscal p,p\rscal + \lscal q,\mathcal H_{\Lambda_L} q \rscal \big)
	\; = \; 
	\frac12 \sum_{k=1}^{|\Lambda_L|} \left( \lscal p,\psi_k\rscal^2 + \nu_k^2 \lscal q,\psi_k\rscal^2 \right)
	\; = :\;
	\sum_{k=1}^{|\Lambda_L|} E_k(q,p).
\end{equation}
where $\lscal\cdot,\cdot\rscal$ is the standard scalar product on $L^2(\Lambda_L)$
(this unusual notation has been chosen to avoid possible confusion with the Gibbs state $\langle\cdot\rangle$ 
and the covariance $\langle \cdot ; \cdot \rangle$ introduced later). 
We say that $E_k$ is the energy of the $k$ mode for $1\le k \le |\Lambda_L|$.
It is conserved under the harmonic dynamics:
$$
	\{H_\har,E_k \} \; =\;  0 \quad \text{for} \quad 1 \le k \le |\Lambda_L|
$$ 
where $\{\cdot,\cdot\}$ denotes the Poisson bracket: 
$$
	\{f,g\} \; = \; \lscal \nabla_p f , \nabla_q g \rscal -  \lscal \nabla_q f , \nabla_p g \rscal
$$
for smooth functions $f,g$ on $\R^{2|\Lambda_L|}$.

The operator $\mathcal H_{\Lambda_L}$ is the Anderson Hamiltonian for a single quantum particle in a disordered potential in the tight-binding approximation.
Since the system is one-dimensional, well known results \cite{gol1977pure,Kunz,carmona} guarantee that all modes are \emph{localized}:
There exists constants $\Const <+\infty$ and $\xi > 0$ so that, for all $L\in\N^*$ and for all $x,y\in\Lambda_L$, 
\begin{equation}
	\E \left( \sum_{k=1}^{|\Lambda_L|} |\psi_k(x) \psi_k(y)| \right) \; \le \; \Const \ed^{-|x-y|/\xi}.
\end{equation} 
See \cite{Kunz,aizenman2015random} as well as Appendix~\ref{sec: localization} below.  

\subsection{Constants}
We will denote by \emph{constants} strictly positive real numbers that may depend on the harmonic coupling $\eta$ and on the distribution of the disorder, but not on the realization of the disorder (they are deterministic constants). 
We will specify explicitly when constants are allowed to depend on some other parameter; see in particular the comment at the very end of this section.
In any case, constants never depend on the total length $L$ nor on the anahrmonicity parameter $\lambda$. 

We will often use the letters $\Const$ and $\const$ to denote generic constants, with the understanding that their values may vary from one instance to another.
Usually, we use the letter $\Const$ to stress that the constant needs to be taken large enough, and the letter $\const$ to stress that it needs to be taken small enough.

\subsection{Main Results}
When $\lambda >0$, the energies $E_k$ are no longer conserved by the dynamics, and any sense of localization is presumably destroyed. 
In our first result, we establish however that these conservation laws are broken very slowly in the limit $\lambda \to 0$. 
For an observable $O$, i.e.\@ a smooth function $O:\R^{2|\Lambda_L|}\to\R$, we will write $O(t)$ for $O(q(t),p(t))$ for any $t\ge 0$, 
where $(q(t),p(t))$ evolve according to eq.~\eqref{eq: Hamilton equations}.  
For $1\le k \le |\Lambda_L|$ and $t\ge 0$, let us define 
\begin{equation}\label{eq: def of C k}
	C_k (t) 
	\; = \; 
	\frac12 \left\langle\left(E_k(t) - E_k(0)\right)^2\right\rangle
	\; = \; 
	\langle E_k;E_k\rangle - \langle E_k(t);E_k(0)\rangle
\end{equation}
where we have removed the time dependence in the first term in the r.h.s.\@ thanks to the invariance of the Gibbs state, and where 
$$
	\langle f;g\rangle \; = \; \langle (f - \langle f \rangle)( g - \langle g \rangle) \rangle
$$
for any observables $f,g$. 
Obviously, $C_k(0) = 0$ and, by Cauchy-Schwarz inequality, $0 \le C_k(t) \le 2 \langle E_k;E_k\rangle$ for all $t \ge 0$.
Numerical results in \cite{ourphysicspaper} suggest that $C_k(t)$ grows monotonically from $0$ to a value close to $\langle E_k;E_k\rangle$, and the deviation from $\langle E_k;E_k\rangle$ vanishes in the thermodynamic limit $L\to\infty$.
The following theorem implies however that for most modes $k$, this decorrelation has to proceed very slowly. 

\begin{Theorem}\label{th: decorrelation}
	Let $n\in\N^*$.
	There exist constants $\Const_n$ and $\const_n$ so that 
	$$
		\limsup_{L\to\infty}\frac1{|\Lambda_L|}\sum_{k=1}^{|\Lambda_L|} C_k(t) \; \le \; \Const_n \left( \lambda^{\const_n}+ (\lambda^n t)^2 \right)
		\qquad \text{a.s.}\quad \forall \lambda \ge 0, \quad \forall t \ge 0.
	$$
\end{Theorem}

Next, we aim to express the slowness of energy transfer in the system without relying on the modes of the harmonic system.
This is not totally obvious because, for example, it is simply not true that the flow of energy between adjacent sites vanishes as $\lambda \to 0$, due to the harmonic couplings. 
What is however true is that the onset of \emph{persistent} transport is very slow. 
This can be understood by contemplating the behavior of the time-integrated energy-current between nearby sites. 
Let us decompose the Hamiltonian in \eqref{eq: H} as 
$$
	H(q,p) \; = \; \sum_{x\in\Lambda_L} H_x(q,p) 
	\quad \text{with} \quad 
	H_x (q,p) \; = \; \frac{p_x^2}{2} + \frac{\omega_x^2 q_x}{2} + \frac{\eta}{2} (q_x - q_{x+1})^2 + \frac{\lambda}{4} q_x^4.
$$
For $x\in\Lambda_L$, we define the local currents $j_x$ by 
\begin{equation}\label{eq: explicit expression current}
	j_x \; = \; \{H_{x-1},H_x\}\; = \; \eta (q_{x-1} - q_x) p_x
\end{equation}
with the convention $j_{-L} = 0$, 
leading to the continuity equation
\begin{equation}\label{eq: energy conservation}
    \frac{\dd H_x}{\dd t} \; = \; \{H_{x-1},H_x\} - \{H_{x},H_{x+1}\} \; = \; j_x - j_{x+1}.
\end{equation}

Taking $x=0$ as a reference point, we define the process $(J_0(t))_{t\ge 0}$ as 
\begin{equation}\label{eq: time integrated current}
	J_0 (t) \; = \; \int_0^t j_0 (s) \dd s 
\end{equation}
which equals the integrated heat flow from the region $x<0$ to the region $x\geq 0$.
At $\lambda = 0$, we expect this process to stay \emph{bounded} in the limit $L\to \infty$, $t\to \infty$ as a result of localization: 
Only a finite amount of energy can be transferred from $x<0$ to $x\geq 0$.
In contrast, this process becomes unbounded if energy spreads across the full chain.
E.g., if one assumes that the system is governed by normal diffusive behaviour and it is started in equilibrium, then $J_0(t)$ grows typically as $t^{1/4}$, see e.g.\@ \cite{derrida_gerschenfeld}.
Our second theorem shows that a long time is needed before energy starts getting dissipated {across} the system for small values of $\lambda$: 
\begin{Theorem}\label{th: current}
	Let $n\in\N^*$ and let $J_0$ be defined as in \eqref{eq: time integrated current}.
	There is a positive random variable $V_n$ so that 
	$$
		\limsup_{L\to\infty} \langle \left(J_0(t)\right)^2 \rangle \; \le \; V_n (1+(\lambda^n t)^2) 
		\qquad \text{a.s.}\quad \forall \lambda \ge 0, \quad \forall t \ge 0.
	$$
Moreover, for every $p>0$, there exists a constant $\Const_{n,p}$ such that $\E(V_n^p)\le \Const_{n,p}$. 
\end{Theorem} 

We can reformulate the above result in a way which is more directly relevant for the standard description of thermal transport. 
Let us define a rescaled process $(\mathcal J(t))_{t\ge 0}$ by 
\begin{equation}\label{eq: total current}
	\mathcal J(t) \; = \; \frac{1}{\sqrt{t|\Lambda_L|}}\sum_{x\in|\Lambda_L|} \int_0^t j_x(s)\dd s. 
\end{equation}
If the limit below is well defined, the thermal conductivity of the chain is defined as 
$$
    \kappa (\lambda) \; := \; \lim_{t\to \infty} \lim_{L\to \infty} \langle (\mathcal J(t))^2 \rangle \; = \; \int_{-\infty}^{+\infty}  \sum_{x\in \Z} \langle j_0(0)j_x(t)\rangle \, \dd t,
$$
where the last expression is obtained by taking the limits \emph{formally}.

To our knowledge, there is no Hamiltonian system for which one can rigorously control the thermal conductivity, 
except in some exceptional cases where $\kappa$ turns out to be $0$ or $+\infty$.
In our case, one can nevertheless gain some understanding in the limit $\lambda\to 0$: 
If one expects that $\kappa (\lambda) \sim \lambda^n$ in the limit $\lambda \to 0$ for some $n>0$,
one should then also expect to detect this behavior by following the dynamics on time scales of the order of $\lambda^{-n'}$ with $n'\ge n$. 
Our next theorem shows that we should not expect $\kappa$ to scale polynomially with $\lambda$, 
and hints to the fact that $\frac{\partial^n \kappa}{\partial \lambda^n} (0) = 0$ for all $n\in\N^*$: 
\begin{Theorem}\label{th: Green Kubo}
	Let $\tau>0$, let $n'\ge n\in\N^*$ and let $\mathcal J$ be defined as in \eqref{eq: total current}.
	There exists a constant $\Const_{n'} = \Const_{n'}(\tau)$ so that 
	$$
		\limsup_{L\to\infty} \langle (\mathcal J(\lambda^{-n'}\tau))^2 \rangle \; \le \; \Const_{n'} \lambda^n 
		\qquad \text{a.s.} \quad \forall \lambda \in[0,1].
	$$
\end{Theorem}
Without loss of generality, we can restrict ourselves to the case $n'=n$ in the above Theorem~\ref{th: Green Kubo}.

\subsection{Out-of-Equilibrium Dynamics}

It would be of definite interest to know to what extent our results still hold if the initial condition of the chain is an out-of-equilibrium state, 
such as for example a pre-thermal state or a locally approximate Gibbs state where the temperature varies as a function of space. 

The main advantage of starting the dynamics in equilibrium is that it yields a control on the expectation of local observables at all later times. 
If we add an extra assumption on the evolution of these observables, it becomes then possible to formulate a meaningful version of  Theorem~\ref{th: current} 
without requiring that the initial state is time-invariant. 

For simplicity of notations, 
let us keep working with the parameter $\lambda$ defined in \eqref{eq: def of lambda}, instead of $\gamma$, 
as well as with the dynamics defined by \eqref{eq: H} and \eqref{eq: Hamilton equations},
even though the parameter $\beta$ has no meaning in our new context. 
Given the interval $\Lambda_L$, let $\langle \cdot\rangle_{\mathrm{ne}}$ denote some initial measure on $\R^{2|\Lambda_L|}$, with the subscript ${\mathrm{ne}}$ standing for ``non-equilibrium'', reminding us that $\langle \cdot\rangle_{\mathrm{ne}}$ replaces the family of Gibbs states used before.
We make the following assumption on the family of measures $\{\langle \cdot\rangle_{\mathrm{ne}}\}_{L}$: 
\begin{Assumption}\label{as: time evolved ensemble}
Given $m\ge 1$, there exists a constant $\Const_m$ such that 
for any interval $\Lambda_L$ with $L \in \N^*$,
for any $x\in \Lambda_L$ 
and for any $t\ge 0$, 
$$
	\langle |q_x(t)|^m \rangle_{\mathrm{ne}} \; \le\; \Const_m, 
	\qquad 
	\langle |p_x(t)|^m \rangle_{\mathrm{ne}} \; \le\; \Const_m.
$$
\end{Assumption}
We deem that this assumption is very reasonable if it holds at the initial time. 
Moreover let us stress that, if we would like to show that it holds in a specific example, 
the difficulty would not be to show it at $t=0$ given an explicit expression for $\langle \cdot \rangle_{\mathrm{ne}}$, 
but to show that the property is preserved over time. 
Let us now state our generalization of Theorem~\ref{th: current} to out-of-equilibrium initial states:
\begin{Theorem}\label{th: out of equilibrium}
	Assume that $\{\langle \cdot \rangle_{\mathrm{ne}}\}_L$ is a family of initial measures that satisfies Assumption~\ref{as: time evolved ensemble}. 
	Given $n\in\N^*$, there is a positive random variable $V_n$ so that
	$$
		\limsup_{L\to\infty} \langle \left(J_0(t)\right)^2 \rangle_{\mathrm{ne}}\; \le \; V_n (1+(\lambda^n t)^2) 
		\qquad \text{a.s.}\quad \forall \lambda \ge 0, \quad \forall t \ge 0.
	$$
Moreover, for every $p>0$, there exists a constant $\Const_{n,p}$, that depends on the familly of initial measures only through the constants introduced in Assumption~\ref{as: time evolved ensemble}, such that $\E(V_n^p)\le \Const_{n,p}$.
\end{Theorem}

In what follows, we will assume that the number $n$ featuring in Theorems~\ref{th: decorrelation} to \ref{th: out of equilibrium} is fixed (recall that it suffices to prove Theorem~\ref{th: Green Kubo} for $n'=n$), 
we will not try to keep track of the dependence of all our expressions on this parameter, 
and we will thus allow constants to depend on $n$ as well in the remainder of this paper.

\section{Proof of Theorems~\ref{th: decorrelation} to \ref{th: Green Kubo}}\label{sec: proof of the theorems}

In this section, 
we prove Theorems~\ref{th: decorrelation} to \ref{th: Green Kubo} assuming that Propositions~\ref{pro: assumptions theorem 1} to \ref{pro: assumptions theorem 3} below hold.   

\subsection{Proof of Theorem~\ref{th: decorrelation}}

\begin{Proposition}\label{pro: assumptions theorem 1}
Given $1 \le k \le |\Lambda_L|$, let us define the observable $f_k$ on $\Lambda_L$ by 
$$	
	f_k \; = \; \{H_\an , E_k\} \; = \; \frac1\lambda \{H,E_k\}.
$$
There exist observables $u_k$ and $g_{k}$ on $\Lambda_L$
(that may depend on $\lambda$) satisfying two properties. 
First, 
\begin{equation}\label{eq: hyp 1 theo 1}
    f_k \; = \; -\{H,u_k\} + \lambda^n g_{k} \qquad \text{a.s.} \quad \forall \lambda\ge 0.
\end{equation}
To state the second property, let us define the ``bad'' set of modes 
\begin{equation}\label{eq: 1st bad set}
	B (M) \; = \; \{k : 1\le k \le |\Lambda_L|, \langle u_k^2\rangle \ge M \text{ or }\langle g_{k}^2 \rangle \ge M\}
\end{equation}
for $M>0$ as well as 
\begin{equation}\label{eq: nu M in proposition 1}
	\nu (M) \; = \; \limsup_{L\to\infty} \frac{|B(M)|}{|\Lambda_L|}.
\end{equation}
The second property is that there exist constants $\Const$ and $0<q<1$ such that 
\begin{equation}\label{eq: result of proposition 1}
	\nu (M) \;\le\; \frac{\Const}{M^q} \qquad \text{a.s.}\quad \forall M>0, \quad \forall \lambda \ge 0.
\end{equation}
\end{Proposition}

\begin{proof}[Proof of Theorem~\ref{th: decorrelation}]
We can obtain two different bounds on $C_k(t)$. 
First, by Corollary~\ref{cor: decay of correlations} in Appendix~\ref{sec: gibbs} on the decay of correlations of the Gibbs state, 
there exits a constant $\Const$ such that
\begin{equation}\label{eq: bound 1 correlator}
	C_k (t) \; \le \; 2 \langle E_k ; E_k \rangle \; \le \; \Const.
\end{equation}
Second, by \eqref{eq: hyp 1 theo 1}, we may write 
\begin{align*}
	E_k(t) - E_k(0) 
	\; &= \; 
	\int_0^t \{H,E_k\}(s) \dd s 
	\; = \; 
	\lambda \int_0^t f_k(s) \dd s 
	\; = \; 
	\lambda \int_0^t (- \{H,u_{k}\} + \lambda^n g_{k})(s) \dd s\\
	\; &= \; 
	-\lambda (u_{k}(t) - u_{k}(0)) + \lambda^{n+1} \int_0^t g_{k}(s) \dd s .
\end{align*}
Hence
\begin{equation}\label{eq: bound 2 correlator}
	C_k(t) \; = \; 
	\langle (E_k(t) - E_k(0))^2 \rangle
	\; \le \;
	6\lambda^2 \langle u_k^2 \rangle +  3\lambda^{2(n+1)}t^2 \langle g_{k}^2 \rangle.
\end{equation}

With the use of the set $B$ introduced in \eqref{eq: 1st bad set}, we decompose
$$
	\frac1{|\Lambda_L|}\sum_{k=1}^{|\Lambda_L|}C_k(t) 
	\; = \;
	\frac1{|\Lambda_L|}\sum_{k\in B(1/\lambda)}C_k(t) + \frac1{|\Lambda_L|}\sum_{k\notin B(1/\lambda)}C_k(t) .
$$
By Proposition~\ref{pro: assumptions theorem 1} and the inequality~\eqref{eq: bound 1 correlator}, we bound the first sum as
$$
	\limsup_{L\to\infty} \frac1{|\Lambda_L|}\sum_{k\in B(1/\lambda)}C_k(t) \; \le \; \Const \nu(1/\lambda) \; \le \; \Const \lambda^q
	\qquad \text{a.s.} \quad \forall \lambda >0. 
$$
We use the bound \eqref{eq: bound 2 correlator} for the second sum: 
$$
	\limsup_{L\to\infty}\frac1{|\Lambda_L|}\sum_{k\notin B(1/\lambda)}C_k(t)  \; \le \; 6 \lambda + 3 \lambda (\lambda^n t)^2
$$
Summing these two bounds yields the claim. 
\end{proof}

\subsection{Proof of Theorem~\ref{th: current}}\label{subsec: proof of Theorem 2}

\begin{Proposition}\label{pro: assumptions theorem 2}
Recall the definition of the current $j_x$ in \eqref{eq: explicit expression current}.
Given $x\in\Lambda_L$, there exist observables $u_x$ and $g_{x}$ on $\Lambda_L$
(which may depend on $\lambda$) that satisfy two properties. 
First, 
\begin{equation}\label{eq: hyp 1 theo 2}
    j_x \; = \; -\{H,u_x\} + \lambda^n g_{x}
    \qquad \text{a.s.} \quad \forall \lambda\ge 0.
\end{equation}
To state the second property, let $M>0$ and let now $B'(M)$ be the set of ``bad'' points defined as
\begin{equation}\label{eq: 2d bad set}
	B'(M) \; = \; \{ x \in\Lambda_L : \langle u_x^2\rangle \ge M \text{ or }\langle g_x^2 \rangle \ge M \} ,
\end{equation}
and define the random variable 
\begin{equation}\label{eq: def ell 0 proposition 2}
    \ell_0 \; = \;\min\{ x\in\Lambda_L : x\ge 0 \text{ and }x\notin B'(M) \} \wedge (L+1)
\end{equation}
with the convention $\min\varnothing = + \infty$.
The second property is that there exists constants $M_0$, $\Const$ and $c$ such that for all $M\ge M_0$, 
\begin{equation}\label{eq: tail of x not distribution}
    \Proba(\ell_0\ge \ell) \; \le \; \Const \ed^{-c \ell^{1/6}} \qquad \forall \ell \ge 0.
\end{equation}
Moreover, the almost sure limit $\limsup_{L\to\infty}\ell_0$ still satisfies this bound.
\end{Proposition}

\begin{proof}[Proof of Theorem~\ref{th: current}]
Let $M$ be such that \eqref{eq: tail of x not distribution} holds.
From the conservation fo energy in \eqref{eq: energy conservation}, and noting that the left-hand side is equal to $\{H,H_x\}$, we find 
$$
	j_0 \; = \; \left\{H,\sum_{x=0}^{\ell_0 - 1}H_x\right\}  + j_{\ell_0},
$$
with the convention $\sum_{k=0}^{-1}(\dots)=0$ and with the definition $j_{L+1}=0$.
Hence, using \eqref{eq: hyp 1 theo 2} if $\ell_0 < L+1$ and $j_{L+1}=0$ otherwise, yields 
\begin{align}
	J_0(t) 
	\; &= \; 
	\int_0^t j_0(s) \dd s \nonumber\\
	\; &= \;
	\int_0^t \left\{H,\sum_{x=0}^{\ell_0 - 1}H_x\right\}(s) \dd s - 
	\chi_{\ell_0 < L+1}\left( \int_0^t \{H,u_{\ell_0}\}(s) \dd s - \lambda^n \int_0^t g_{\ell_0}(s) \dd s \right)\nonumber\\
	&=\;
	\sum_{x=0}^{\ell_0 - 1}H_x(t) - \sum_{x=0}^{\ell_0 - 1}H_x(0) - 
	\chi_{\ell_0 < L+1} \left( u_{\ell_0}(t) - u_{\ell_0}(0) - \lambda^n \int_0^t g_{\ell_0}(s) \dd s \right).\label{eq: some expansion for J0t}
\end{align}
Therefore 
\begin{equation}\label{eq: j 0 t to be modified for last theorem}
	\langle (J_0(t))^2 \rangle 
	\; \le \; (2 \ell_0 + 3) \left( 2 \sum_{x=0}^{\ell_0 - 1}\langle H_x^2\rangle 
	+ \chi_{\ell_0 < L+1}\left(2 \langle u_{\ell_0}^2 \rangle + (\lambda^nt)^2\langle g_{\ell_0}^2\rangle \right) \right).
\end{equation}
By Corollary~\ref{cor: decay of correlations} in Appendix~\ref{sec: gibbs}, there exists a constant $\Const$ such that $\langle H_x^2 \rangle \le \Const$ and hence, 
by Proposition~\ref{pro: assumptions theorem 2}, we find a constant $\Const'$ such that 
$$
	\langle (J_0(t))^2 \rangle 
	\; \le \; (2 \ell_0 + 3) \left( 2\ell_0 \Const + 2 M + (\lambda^nt)^2M \right)
	\; \le \; \Const' (1+\ell_0^2)(1 + (\lambda^nt)^2). 
$$
The claim then follows from \eqref{eq: tail of x not distribution} in the limit $L\to\infty$.
\end{proof}

\subsection{Proof of Theorem~\ref{th: Green Kubo}}

\begin{Proposition}\label{pro: assumptions theorem 3}
For all $x\in\Lambda_L$, there exist observables $u_x$ and $g_{x}$ on $\Lambda_L$
(that may depend on $\lambda$) that satisfy two properties. 
The first one is the decomposition~\eqref{eq: hyp 1 theo 2} again. 
To state the second one, let $M>0$ and let now $B''(M)$ be the set of ``bad'' points defined as 
$$
	B''(M) \; = \; \left\{x\in\Lambda_L :  
	\sum_{y\in \Lambda_L}|\langle u_xu_y\rangle| \ge M 
	\quad\text{or}\quad
	\sum_{y\in \Lambda_L}|\langle g_xg_y\rangle| \ge M
	\right\}.
$$
as well as the weight $w:\Lambda_L \to \N$ such that
\begin{equation}\label{eq: weight green kubo}
    w(x)=\min\{|x-y| : y<x, y\notin B''(M)\} 
\end{equation}
(notice that the condition $y\notin B''(M)$ includes the possibility $y=-L-1$).
The second property is that there exist constants $M_0$ and $\Const$ such that, 
for every $M\ge M_0$,  
\begin{equation}\label{eq: ergodic average weight}
	\limsup_{L\to\infty}\frac1{|\Lambda_L|}\sum_{x\in\Lambda_L} (w(x))^2 \; \le \; \Const.
\end{equation}
\end{Proposition}

\begin{proof}[Proof of Theorem~\ref{th: Green Kubo}]
We assume that $M$ is taken large enough so that \eqref{eq: ergodic average weight} holds.
For $x\in\Lambda_L$, let 
$$
	y_0 (x) \; = \; \min\{y\in\Lambda_L:y\ge x \text{ and }y\notin B''(M) \} \wedge (L+1)
$$
and, using \eqref{eq: hyp 1 theo 2}, let us decompose 
\begin{align*}
	j_x 
	\; &= \; \left\{H,\sum_{y=x}^{y_0(x)-1}H_y\right\} + j_{y_0(x)}\\
	\; &= \; \left\{H,\sum_{y=x}^{y_0(x)-1}H_y\right\}  - \chi_{y_0(x)<L+1}\left( \left\{H,u_{y_0(x)}\right\} - \lambda^n g_{x+\ell_0(x)} \right) .
\end{align*}
Integrating this from $0$ to $t=\lambda^{-n}\tau$ yields
\begin{align*}
	\int_0^t j_x (s) \dd s 
	\; = &\; 
	\sum_{y=x}^{y_0(x)-1}\tilde H_y (t) - \sum_{y=x}^{y_0(x)-1}\tilde H_y (0)\\
	&\; -\chi_{y_0(x)<L+1} \left( u_{y_0(x)}(t) - u_{y_0(x)}(0) - \lambda^n \int_0^t g_{y_0(x)} (s) \dd s \right)
\end{align*}
with $\tilde H_y = H_y - \langle H_y\rangle$.
Hence
\begin{multline}\label{eq: GK sum of 3}
	\left\langle\left(\sum_{x\in\Lambda_L} \int_0^t j_x(s) \dd s \right)^2\right\rangle
	\; \le \; 
	10 \left\langle\left(\sum_{x\in\Lambda_L}\sum_{y=x}^{y_0(x)-1}\tilde H_y\right)^2\right\rangle\\
	\; + 10  \left\langle\left(\sum_{x\in\Lambda_L} u_{y_0(x)} \chi_{y_0(x)<L+1} \right)^2\right\rangle
	+ 5 (\lambda^{n}t)^2 \left\langle\left(\sum_{x\in\Lambda_L} g_{y_0(x)} \chi_{y_0(x)<L+1} \right)^2\right\rangle .
\end{multline}
To conclude, it is now enough to prove that there exists a constant $\Const$ such that 
$$
	\limsup_{L\to\infty} \frac{A_L}{|\Lambda_L|} \; \le \; \Const
$$
where $A_L$ denotes any of the three expectations with respect to the Gibbs state featuring in \eqref{eq: GK sum of 3}.
Indeed, this will imply that 
$$
	\limsup_{L\to\infty} \langle (J(t))^2 \rangle 
	\; = \;
	\limsup_{L\to\infty}\frac{1}{|\Lambda_L|t}\left\langle\left(\sum_x \int_0^t j_x(s) \dd s \right)^2\right\rangle
	\; \le \;
	\Const \left( \frac1t + \lambda^{2n}t\right) \; = \; \Const \left(\frac{\lambda^n}{\tau} + \lambda^n \tau \right). 
$$

Let us first deal with the first term in the right hand side of \eqref{eq: GK sum of 3}.
We observe that
$$
	\sum_{x\in\Lambda_L}\sum_{y=x}^{y_0(x)-1}\tilde H_y
	\; = \; 
	\sum_{x\in B''(M)} w(x) \tilde H_x.
$$ 
Hence 
$$
	\left\langle\left(\sum_{x\in\Lambda_L}\sum_{y=x}^{y_0(x)-1}\tilde H_y\right)^2\right\rangle
	\; \le \; 
	\sum_{x,y\in\Lambda_L} w(x)w(y) |\langle \tilde H_x \tilde H_y\rangle|
	\; \le \; 
	\sum_{x\in\Lambda_L} w^2(x) \sum_{y\in\Lambda_L} |\langle \tilde H_x \tilde H_y\rangle| .
$$
By Proposition \ref{pro: decay of correlations} and Corollary~\ref{cor: decay of correlations} in Appendix~\ref{sec: gibbs}, 
the second factor in the right hand side of this expression is bounded by a constant. 
Hence, using \eqref{eq: ergodic average weight}, we find that 
$$
	\limsup_{L\to\infty}\frac1{|\Lambda_L|}\left\langle\left(\sum_{x\in\Lambda_L}\sum_{y=x}^{y_0(x)-1}\tilde H_y\right)^2\right\rangle \; \le \; \Const. 
$$
Let us next deal with the second term in the right hand side of \eqref{eq: GK sum of 3}. 
This time, it holds that 
$$
	\sum_{x\in\Lambda_L} u_{y_0(x)} \chi_{y_0(x)<L+1} \; = \; \sum_{x\in\Lambda_L:x\notin B''(M)} w(x)u_x.
$$
Hence, by definition of $B''(M)$,
$$
	 \left\langle\left(\sum_{x\in\Lambda_L} u_{y_0(x)} \chi_{y_0(x)<L+1} \right)^2\right\rangle
	 \; \le \; 
	 \sum_{x\in\Lambda_L:x\notin B''(M)}w^2(x) \sum_{y\in \Lambda_L} |\langle u_xu_u\rangle|
	 \; \le \;
	 M  \sum_{x\in\Lambda_L:x\notin B''(M)}w^2(x) .
$$
Using again \eqref{eq: ergodic average weight}, we conclude that 
$$
	\frac{1}{|\Lambda_L|}\left\langle\left(\sum_{x\in\Lambda_L} u_{y_0(x)} \chi_{y_0(x)<L+1} \right)^2\right\rangle \; \le \; \Const. 
$$
One deals with the third term in the right hand side of \eqref{eq: GK sum of 3} as the second term, and this concludes the proof. 
\end{proof}

\section{Approximate Solution of the Commutator Equation}\label{sec: perturbative expansion}

In this section, we will write down a perturbative expansion that yields the decomposition
\begin{equation}\label{eq: commutator equation}
	f \; = \; - \{H,u\} + \lambda^n g \qquad \text{a.s.} 
 \quad \forall \lambda > 0
\end{equation}
for some specific observables $f$. 
This decomposition appears in \eqref{eq: hyp 1 theo 1} for Proposition~\ref{pro: assumptions theorem 1}, 
with $f$ given by 
\begin{equation}\label{eq: f for prop 1}
	f_{k_0} \; := \; \{H_\an,E_{k_0}\} \; = \; \frac{1}{\lambda}\{H,E_{k_0}\}
\end{equation}
for some $1 \le k_0 \le |\Lambda_L|$, 
and in \eqref{eq: hyp 1 theo 2} for Propositions~\ref{pro: assumptions theorem 2} and \ref{pro: assumptions theorem 3}, 
with $f$ given by 
\begin{equation}\label{eq: f for prop 2 and 3}
	j_{x_0} \; = \; \eta (q_{x_0-1} - q_{x_0}) p_{x_0}
\end{equation}
for some $x_0\in \Lambda_L$. 

Formally, this expansion can be set up like this: 
Let $f^{(1)}=f$ and let us assume that we can define recursively observables $u^{(i)}$ and $f^{(i+1)}$ for $i \ge 1$ such that 
\begin{equation}\label{eq: perturbative scheme}
	-\{H_\har,u^{(i)}\}=f^{(i)}, \qquad f^{(i+1)} \; = \; \{H_\an,u^{(i)}\}.
\end{equation}
Then, defining $u$ and $g$ as 
\begin{equation}\label{eq: first expression for u and g}
	u \; = \; u^{(1)} + \dots + \lambda^{n-1}u^{(n)}, \qquad g \; = \; f^{(n+1)},
\end{equation}
we find that \eqref{eq: commutator equation} is indeed satisfied.
We will now make sure that this perturbative scheme is well defined and we will write down explicit expressions for $u$ and $g$.

\subsection{Solving the Commutator Equation at $\lambda = 0$}\label{subsec: commutator equation}
Here we provide an almost sure solution to the commutator equation 
\begin{equation}\label{eq: commutator equation lambda = 0}
    h \; = \; - \{H_\har,v\}
\end{equation}
where $h$ is given and satisfies some properties, and where $v$ is the unknown. 
To do so, we begin by introducing a change of variables that simplify our expressions as much as possible.

For $1 \le k \le |\Lambda_L|$, let 
\begin{equation}\label{eq: pqa transform}
	a^\pm_k \; = \; \frac{1}{\sqrt 2} \left( \nu_k^{1/2} \lscal q,\psi_k\rscal \mp \frac{\imag}{\nu_k^{1/2}} \lscal p,\psi_k \rscal \right)
\end{equation}
where we remind that $\lscal\cdot , \cdot\rscal$ is our notation for the scalar product. 
The inverse transform
\begin{equation}\label{eq: pqa reverse transform}
	\lscal q,\psi_k\rscal \; = \; \frac{a_k^+ + a_k^-}{\sqrt 2 \nu_k^{1/2}}, 
	\qquad
	\lscal p,\psi_k\rscal \; = \; \frac{\imag \nu_k^{1/2} (a_k^+ - a_k^-)}{\sqrt 2}.
\end{equation}
leads to an expression for $H_\har$ and $H_\an$ in these new coordinates: 
\begin{align}
	H_\har \; &= \; \sum_{k=1}^{|\Lambda_L|} \nu_k\, a^+_k a^-_k,\\
	H_\an \; &= \; \sum_{1\le k_1,\dots,k_4\le |\Lambda_L|} 
	\sum_{\sigma_1,\dots,\sigma_4\in \{\pm 1 \}} \hat H_\an(k_1,k_2,k_3,k_4) a_{k_1}^{\sigma_1} a_{k_2}^{\sigma_2}  a_{k_3}^{\sigma_3}  a_{k_4}^{\sigma_4} 
\end{align}
with 
\begin{equation}\label{eq: H an k}
	\hat H_\an(k_1,\dots ,k_4) \; = \; \frac{\sum_{x\in \Lambda_L} \psi_{k_1}(x)\psi_{k_2}(x)\psi_{k_3}(x)\psi_{k_4}(x)}{16 (\nu_{k_1} \nu_{k_2} \nu_{k_3} \nu_{k_4})^{1/2}}.
\end{equation}

Let us now proceed to solving the equation \eqref{eq: commutator equation}. 
For this, let us first specify the form of the function $h$. 
Let $m\ge 1$ and let $h$ be a homogeneous polynomial in $a^{\pm}$ of degree $m$:
$$
    h \; = \; \sum_{k_1,\dots,k_m}\sum_{\sigma_1,\dots,\sigma_m} \hat h (k_1,\dots,k_m,\sigma_1,\dots,\sigma_m) a_{k_1}^{\sigma_1} \dots a_{k_m}^{\sigma_m}.
$$
Notice that the coefficients $\hat h(k_1,\dots,k_m,\sigma_1,\dots,\sigma_m)$ are not uniquely defined since any permutation of $(k_1,\sigma_1),\dots,(k_m,\sigma_m)$ leaves the monomial $a_{k_1}^{\sigma_1}\dots a_{k_m}^{\sigma_m}$ invariant.
By \eqref{eq: pqa transform}, the function $h$ can also be considered as a homogeneous polynomial in the variables $q$ and $p$.
We assume here that $h$ is $p$-\emph{antisymmetric}, i.e.\@ $h(q,-p)=-h(q,p)$.

By \eqref{eq: pqa transform}, 
the $p$-antisymmetry of $h$ implies the following key property: 
It is possible to take the coefficients $\hat{h}$ so that  $\hat h(k,\sigma)=0$ whenever $(k,\sigma)\in\mathcal S$, where $\mathcal S$ is the set of indices $(k,\sigma)$ so that the monomial $a_{k_1}^{\sigma_1}\dots a_{k_m}^{\sigma_m}$ is $\sigma$-symmetric. {That is, $(k,\sigma) \in \mathcal S$ whenever $\{1,\ldots, m\}$ can be partitioned in pairs $(i,j)$ so that $k_i=k_j$ and $ \sigma_i=-\sigma_j$.}
In the rest of this paragraph, we assume that the coefficients $\hat{h}$ have this property.

Moving forward, we observe that the rules
\begin{equation}\label{eq: poisson bracket rule}
    \{a^\sigma_k,a^{\sigma'}_{k'}\} 
    \; = \; 
    \imag \sigma \delta(\sigma+\sigma') \delta(k-k'), 
    \qquad 
    \{\varphi_1,\varphi_2\varphi_3\} 
    \; = \; 
    \{\varphi_1,\varphi_2\} \varphi_3 + \varphi_2 \{\varphi_1,\varphi_3\}
\end{equation}
hold for any $1 \le k,k'\le |\Lambda_L|$ and any smooth observables $\varphi_1,\varphi_2,\varphi_3$.
Here and in the remainder of the paper, to avoid the appearance of excessive subscripts, we use $\delta(\cdot)$ to denote the Kronecker delta function, i.e.\@ $\delta(u) = 1$ if $u=0$ and $\delta(u)=0$ otherwise.
Therefore, given some $k$ and $\sigma$, we compute 
$$
	\{H_\har,a^\sigma_k\} \; = \; -\imag \sigma \nu_k a^\sigma_k
$$
and more generally 
$$
	\{H_\har,a^{\sigma_1}_{k_1} \dots a^{\sigma_m}_{k_m}\} \; = \; -\imag (\sigma_{1} \nu_{k_1} + \dots +\sigma_{m} \nu_{k_m}) a^{\sigma_1}_{k_1} \dots a^{\sigma_m}_{k_m}.
$$
Therefore, we may set 
\begin{equation}\label{eq: expression for v}
	v \; = \; -\imag \sum_{k_1,\dots,k_m}\sum_{\sigma_1,\dots,\sigma_m} 
	\frac{\hat h(k_1,\dots,k_m,\sigma_1,\dots,\sigma_m)}{\sigma_1\nu_{k_1} + \dots +\sigma_m \nu_{k_m}} a_{k_1}^{\sigma_1} \dots a_{k_m}^{\sigma_m}. 
\end{equation}
For any monomial in the r.h.s., either the denominator $\sigma_1\nu_{k_1} + \dots +\sigma_m \nu_{k_m}$ is non-zero a.s., or it vanishes a.s.
The latter occurs if and only if $(k,\sigma)\in\mathcal S$. 
Since we assume that $\hat h(k,\sigma)=0$ for $(k,\sigma)\in\mathcal S$,
the expression \eqref{eq: expression for v} defines almost surely a function $v$ that solves \eqref{eq: commutator equation lambda = 0}, provided we use the convention $0/0=0$.

\subsection{Well-Defined Nature of the Perturbative Expansion}\label{sec: well defined perturbation}
We now show that the perturbative scheme defined by \eqref{eq: perturbative scheme} is well defined, leading to well-defined functions $u$ and $g$ in \eqref{eq: first expression for u and g}. 

Let us first show that $f$, as given by \eqref{eq: f for prop 1} or \eqref{eq: f for prop 2 and 3}, is a $p$-antisymmetric homogeneous polynomial, of degree $d_1$ that we will determine. 
If $f = f_{k_0} = \{H_\an,E_{k_0}\}$, since $\lbrace H_\an,\cdot\rbrace  =  - \lscal\nabla_q H_\an,\nabla_p \cdot\rscal$, 
the claim follows from the fact that $E_{k_0}$ and $H_\an$ are symmetric in $p$, and $d_1=4$. 
If $f=j_{x_0}$, this follows directly from the expression \eqref{eq: explicit expression current}, and $d_1=2$.

Assuming now recursively that $f^{(i)}$ is a $p$-antisymmetric homogeneous polynomial of degree $d_i$ for $i\ge 1$, 
we deduce from the computations in Section~\ref{subsec: commutator equation} that $u^{(i)}$ is a well defined homogeneous polynomial of the same degree.
Moreover, it is $p$-symmetric since $\{H_\har,\cdots\}$ exchanges symmetric and antisymmetric functions.
Again, since $\lbrace H_\an,\cdot\rbrace  =  - [\nabla_q H_\an,\nabla_p \cdot]$, 
we conclude that $f^{(i+1)}$ is a $p$-antisymmetric homogeneous polynomial of degree $d_i+2$, and in particular that $\langle f^{(i+1)}\rangle=0$.
We conclude thus also that the degree of $u^{(i)}$ and $f^{(i)}$ is given by
\begin{equation}\label{eq: di}
	d_i \; = \; d_1 + 2(i-1).
\end{equation}

\subsection{Explicit Expressions}
We can now provide explicit representations of the functions $u$ and $g$ in \eqref{eq: first expression for u and g}. 
From now on, to lighten some expressions, we will use the shorthand notations
$$
	k=(k_1,\dots,k_m), \qquad \sigma=(\sigma_1,\dots,\sigma_m), \qquad a_k^\sigma=a_{k_1}^{\sigma_1}\dots a_{k_{m}}^{\sigma_{m}}.
$$
for given $m\ge 1$, and we will write $d(k,\sigma) = m$ to specify the value of $m$. 

\paragraph{Expressions for $f^{(1)}$:}
Let us expand $f^{(1)} = f$ as
\begin{equation}\label{eq: coef f1}
	f \; = \; \sum_{k,\sigma} \hat f(k,\sigma) a_k^\sigma 
\end{equation}
where $d(k,\sigma)=d_1$, with $d_1$ as in Section~\ref{sec: well defined perturbation}. 
First, if $f = f_{k_0} = - \{E_{k_0},H_\an\}$, we compute 
\begin{equation}\label{eq: coef f1 k0}
	\hat f(k,\sigma) \; = \; \imag \hat H_\an(k) \sum_{j=1}^4 \sigma_j \delta (k_0 - k_j)
\end{equation}
with $\hat H_\an(k)$ given by \eqref{eq: H an k}.
Next, if $f = j_{x_0} = \eta (q_{x_0-1}-q_{x_0})p_{x_0}$, we obtain 
\begin{align*}
 	j_{x_0}
	\; &= \; 
	\eta \sum_{k_1,k_2} \lscal q,\psi_{k_1}\rscal\lscal p,\psi_{k_2}\rscal (\psi_{k_1}(x_0-1)-\psi_{k_1}(x_0))\psi_{k_2}(x_0)\\
	\; &= \; 
	\frac{\eta\imag}{2}\sum_{k_1,k_2} (\psi_{k_1}(x_0-1)-\psi_{k_1}(x_0))\psi_{k_2}(x_0) \left(\frac{\nu_{k_2}}{\nu_{k_1}}\right)^{1/2} 
	\sum_{\sigma_1,\sigma_2} \sigma_2 a_{k_1}^{\sigma_1}a_{k_2}^{\sigma_2}.
\end{align*}
Hence we may set 
\begin{equation}\label{eq: coef f1 j}
    \hat f(k_1,k_2,\sigma_1,\sigma_2) 
    \; = \;
    \frac{\eta\imag}{2}(\psi_{k_1}\big(x_0-1)-\psi_{k_1}(x_0)\big)\psi_{k_2}(x_0) \left(\frac{\nu_{k_2}}{\nu_{k_1}}\right)^{1/2} \sigma_2
    \; \big(1 - \delta(k_1-k_2)\delta(\sigma_1+\sigma_2) \big)
\end{equation}
where $\delta(\cdot)$ is the Kronecker delta function. 
The expressions for $\hat f(k,\sigma)$ given by \eqref{eq: coef f1 k0} and \eqref{eq: coef f1 j} vanish for $\sigma$-symmetric monomials, i.e.\@ are such that $\hat f(k,\sigma) = 0$ for $(k,\sigma)\in\mathcal S$. 

\paragraph{Expressions for $f^{(i)}$ and $u^{(i)}$ for $i\ge 1$:}
By the analysis in Section~\ref{sec: well defined perturbation}, 
the functions $f^{(i)}$ and $u^{(i)}$ take the form
\begin{equation}\label{eq: expansion f i and u i}
    f^{(i)} \; = \; \sum_{k,\sigma} \hat f^{(i)}(k,\sigma) a_k^\sigma, 
    \qquad
    u^{(i)} \; = \; \sum_{k,\sigma} \hat u^{(i)}(k,\sigma) a_k^\sigma, 
\end{equation}
with $d(k,\sigma)=d_i$ as defined in \eqref{eq: di}. 

In order to provide explicit expressions, let us introduce some notations.
First, it is convenient to adopt a more compact notation for denominators: 
Given $m\ge1$, $k=(k_1,\dots,k_m)$ and $\sigma=(\sigma_1,\dots,\sigma_m)$, let 
\begin{equation}\label{eq: definition of denominator}
	\Delta(k,\sigma) \; = \; \sigma_1\nu_{k_1} + \dots + \sigma_m\nu_{k_m}.
\end{equation}
Second, we need to define a set of indices on which \emph{contractions} will be performed. 
Here, a contraction refers to the use of the first Poisson bracket in \eqref{eq: poisson bracket rule}, where two $a$-observables are ``contracted'' into a number. 
For $i\ge2$, let 
\begin{equation}\label{eq: I i set}
    \mathcal U_i \subset \N^{i-1}\times \N^{i-1}
\end{equation}
be such that $(s,t)\in \mathcal U_i$ if and only if 
\begin{align*}
    & d_1 + 4(j -1) + 1 \;\le\; t_j \;\le\; d_1 + 4j, \quad 1 \le j \le i-1, \\
    & 1 \;\le\; s_j \;\le\; d_1 + 4(j-1), \quad 1 \le j \le i-1, \qquad
    s_j \;\ne\; s_1,\dots s_{j-1},t_1,\dots,t_{j-1}, \quad 2 \le j \le i-1. 
\end{align*}
Finally, we want to make explicit that all denominators involved in our expressions are almost surely non-zero. 
Let $(k,\sigma)$ with $d(k,\sigma)=d_1 + 4(i-1)$ for some $i\ge 1$. 
We say that $(k,\sigma)\in\mathcal R$ if and only if
\begin{equation}\label{eq: definition of the R set}
    (k_1,\dots,k_{d_1+4j},\sigma_1,\dots,\sigma_{d_1+4j}) \notin \mathcal S
    \qquad 
    \forall j = 0,\dots, i-1.
\end{equation}

With these notations, we can state
\begin{Proposition}\label{pro: main result perturbation theory}
Let first $i=1$.
The expression for $f^{(1)}=f$ takes the form \eqref{eq: coef f1} with coefficients given by \eqref{eq: coef f1 k0} for $f=f_{k_0}$, or by \eqref{eq: coef f1 j} for $f=j_{x_0}$.
The expression for $u^{(1)}$ takes the form \eqref{eq: expansion f i and u i} with $\hat u^{(1)}(k,\sigma) = - \imag \hat f(k,\sigma)/\Delta(k,\sigma)$.

Let then $i\ge 2$. 
The function $f^{(i)}$ is given by 
\begin{align}\label{eq: expression for fi}
    f^{(i)} 
    \; = &\; 
    \sum_{(k,\sigma)\in \mathcal R} \sum_{(s,t)\in \mathcal U_{i}}
    \delta(k_{s_1}-k_{t_1})\dots \delta(k_{s_{i-1}} - k_{t_{i-1}}) \delta(\sigma_{s_1}+\sigma_{t_1})\dots\delta(\sigma_{s_{i-1}}+\sigma_{t_{i-1}})\nonumber\\
    &\frac{\sigma_{t_1}\dots \sigma_{t_{i-1}}\hat f(k'_1,\sigma'_1)
    \hat H_{\mathrm{an}}(k'_2)\dots \hat H_{\mathrm{an}}(k'_{i})}
    {\Delta(k'_1,\sigma'_1)\Delta((k'_1,k'_2),(\sigma'_1,\sigma'_2)) 
    \dots \Delta((k'_1,\dots,k'_{i-1}),(\sigma'_1,\dots,\sigma'_{i-1}))}
    \; \tilde a_k^\sigma
\end{align}
and the function $u^{(i)}$ is given by 
\begin{align}\label{eq: expression for ui}
    u^{(i)} 
    \; = &\; -\imag
    \sum_{(k,\sigma)\in \mathcal R} \sum_{(s,t)\in \mathcal U_{i}}
    \delta(k_{s_1}-k_{t_1})\dots \delta(k_{s_{i-1}} - k_{t_{i-1}}) \delta(\sigma_{s_1}+\sigma_{t_1})\dots\delta(\sigma_{s_{i-1}}+\sigma_{t_{i-1}})\nonumber\\
    &\frac{\sigma_{t_1}\dots \sigma_{t_{i-1}}\hat f(k'_1,\sigma'_1)
    \hat H_{\mathrm{an}}(k'_2)\dots \hat H_{\mathrm{an}}(k'_{i})}
    {\Delta(k'_1,\sigma'_1)\Delta((k'_1,k'_2),(\sigma'_1,\sigma'_2)) 
    \dots \Delta((k'_1,\dots,k'_{i}),(\sigma'_1,\dots,\sigma'_{i}))}
    \; \tilde a_k^\sigma. 
\end{align}
In the above expressions, we have used the following conventions: 
\begin{enumerate}
	\item 
	$d(k,\sigma)=d_1+4(i-1)$.
	\item
	The coordinates of $k$ and $\sigma$ are expressed in two different ways: First $k = (k_1,\dots,k_{d(k,\sigma)})$
	and second $k=(k'_1,\dots,k'_{i})$ with 
	$$
		k'_1=(k_1,\dots,k_{d_1}), \;
		k'_2=(k_{d_1+1},\dots,k_{d_1+4}),\; 
		\dots,\; 
		k'_i = (k_{d_1 + 4(i-2)+1},\dots,k_{d_1+4(i-1)}),
	$$ 
        and similarly for $\sigma$.
	\item
	The factors $a_{k_{s_1}}^{\sigma_{s_1}} a_{k_{t_1}}^{\sigma_{t_1}}\dots a_{k_{s_{i-1}}}^{\sigma_{s_{i-1}}} a_{k_{t_{i-1}}}^{\sigma_{t_{i-1}}}$ are omitted from $\tilde a_k^\sigma$ (they have been ``contracted'').
\end{enumerate}
\end{Proposition}

Before proceeding to the proof of that proposition, let us make a few comments. 
First, the expressions (\ref{eq: expression for fi},\ref{eq: expression for ui}) are compatible with the general form \eqref{eq: expansion f i and u i}, 
though it may not be obvious to deduce the value of the coefficients $f^{(i)}(k,\sigma)$ and $u^{(i)}(k,\sigma)$ featuring in \eqref{eq: expansion f i and u i} (but we will not need them). 
Indeed, since $2(i-1)$ factors have been contracted in each term, the polynomial is homogeneous of degree $d_i$. 
Second, and similarly, the contractions induce cancellations also in the denominators, so that $\Delta(k'_1,\dots,k'_m,\sigma'_1,\dots,\sigma'_m)$ are linear combinations involving only $d_m$ eigenfrequencies.
Third, the condition $(k,\sigma)\in\mathcal R$ ensures that all denominators are almost surely non-zero. 

It may also be worth to stress that the above proposition solves our problem~\eqref{eq: commutator equation}. 
Indeed, by \eqref{eq: first expression for u and g}, we define $g=f^{(n+1)}$ and $u = u^{(1)}+ \dots + \lambda^{n-1} u^{(n)}$. 
We notice also that $\langle g \rangle = 0$ and that $u$ is defined up to an additive constant, that can in particular be taken such that $\langle u \rangle = 0$.

\begin{proof}[Proof of Proposition~\ref{pro: main result perturbation theory}.]
    First, the expressions (\ref{eq: coef f1}-\ref{eq: coef f1 j}) for $f^{(1)}$ have been established at the beginning of this section, and the announced expression for $u^{(1)}$ follows directly from the discussion in Section~\ref{subsec: commutator equation}. 
    
    Second, for all $i\ge 2$, let us show that the representation \eqref{eq: expression for fi} for $f^{(i)}$ implies the representation \eqref{eq: expression for ui} for $u^{(i)}$. 
    Since $u^{(i)}$ solves the equation $-\{H_\har,u^{(i)}\}=f^{(i)}$, we may use
    the general form \eqref{eq: expression for v} for the solution of this commutator equation established in Section~\ref{subsec: commutator equation}. 
    To show that this formula can be meaningfully applied, we notice that the constraint $(k,\sigma)\in\mathcal R$ in \eqref{eq: expression for fi} implies that $(k,\sigma)\notin \mathcal S$, and thus also $(\tilde k,\tilde \sigma)\notin \mathcal S$, 
    where $(\tilde k, \tilde \sigma)$ are the sub/superscripts of the $a$-operators that have not been omitted in $\tilde a_k^\sigma$. 
    Similarly, we check that the new denominator $\Delta((k_1',\dots,k_i'),(\sigma_1',\dots,\sigma_i'))$ featuring in \eqref{eq: expression for ui} is a re-expression of the denominator stemming from the use of \eqref{eq: expression for v}.

    Finally, let us show that our expressions for $u^{(i)}$, with $i\ge 1$, imply the expression \eqref{eq: expression for fi} for $f^{(i+1)}$. 
    For this, we start from the definition $f^{(i+1)} = - \{u^{(i)},H_\an\}$.
    Using the rules in \eqref{eq: poisson bracket rule}, 
    we consider each term in the right hand side of \eqref{eq: expression for ui}, 
    indexed by some $(k,\sigma)\in\mathcal R$, and we compute
    \[
    -\{\tilde a_k^\sigma,H_\an\}
    \; = \; 
    -\sum_{k',\sigma'} \hat H_\an(k') \{\tilde a_k^\sigma,a_{k'}^{\sigma'}\}
    \; = \; 
    \imag\sum_{k',\sigma'}\sum_{s=1}^{d_i}\sum_{t=1}^4 \delta(\tilde k_s - k'_t)\delta(\tilde\sigma_s + \sigma'_t)
    \sigma'_t \hat H_\an(k')
    \tilde a_{(k,k')}^{(\sigma,\sigma')}
    \]
    where $d(k',\sigma')= 4$ and where  $(\tilde k, \tilde \sigma)$ are the sub/superscripts of the $a$-operators that have not been omitted in $\tilde a_k^\sigma$. 
    Further, when summing over $(k,\sigma)\in\mathcal R$, we may impose the constraint $((k,k'),(\sigma,\sigma'))\notin \mathcal S$ thanks to the $p$-antisymmetry of $f^{(i+1)}$.
    From the expression for $u^{(i)}$ in \eqref{eq: expression for ui}, 
    from the definition of the set $\mathcal U_{i+1}$ in \eqref{eq: I i set} 
    and from the definition of the set $\mathcal R$ in \eqref{eq: definition of the R set}, 
    we conclude that $f^{(i+1)}$ is given by \eqref{eq: expression for fi}.
\end{proof}

\section{Bound on Denominators}\label{sec: bound denominators}
We remind that $(\omega_x^2)_{x\in\Z}$ is a sequence of i.i.d.\@ random variables, 
with a density that we denote by $\mu$ in this section. 
The density $\mu$ is assumed to be continuously differentiable and compactly supported on $[\omega_-^2,\omega_+^2]$, cf.~Section~\ref{subsec: model}.

Let $I\subset\Z$ be a \emph{finite} interval and let us consider the random operator $\mathcal H_I$ defined in \eqref{eq: Anderson hamiltonian}. 
We remind also that its eigenvalues $\nu^2$ are positive and satisfy $\nu_-^2\leq \nu^2\leq \nu_+^2$ with $\nu_+^2 > \nu_-^2>0$, cf.~Section~\ref{sec: the localized harmonic chain}.
Let us fix an even positive integer $m$, that is a parameter in what follows. 
Given $\sigma \in (\Z\backslash\{0\})^m$ such that $|\sigma_k| \le m$ for all $1\le k \le m$, let us define the random variable 
\begin{equation}\label{eq: Q random variable}
	Q \; = \;  \min\left|\sum_{k=1}^m \sigma_k \nu_k\right|
\end{equation}
where the minimum is taken over all $m$-tuples $(\nu_1^2,\dots,\nu_m^2)$ of different eigenvalues of $\mathcal H_I$.
The estimate we  prove is
\begin{Theorem}\label{thm: denominator}
There exists a constant $\Const_m$ such that, for any $ \epsilon >0$, 
\begin{equation}\label{eq: bound}
	\Proba(Q\leq \epsilon) \;\leq\; \Const_m |I|^{m}  \epsilon^{\frac{1}{2m(m+1)}}.
\end{equation} 
\end{Theorem}
The remainder of this section is devoted to the proof of Theorem~\ref{thm: denominator}.
We now fix some $\epsilon>0$, and the dependence on $\epsilon$ will always be written explicitly. 
Constants are allowed to depend on $m$, but we will not indicate this dependence explicitly.

\subsection{Changing Variables}
Without loss of generality, we may assume that $I = [1,L]\cap\Z$ for some $L\ge 1$.  
Since $\mathcal H_I$ depends only on $\omega_x^2$ for $x\in I$, the full disorder space can be represented as $\Omega = [\omega_-^2,\omega_+^2]^L$.

We introduce the change of variables
\begin{equation}\label{eq: change of variables}
    x \; = \; \frac1L \sum_{i=1}^L \omega^2_i, 
    \qquad 
    \tau_i \;=\;\omega^2_i-\omega^2_1, \quad 1 < i \le L,
\end{equation}
and we set $\tau = (\tau_i)_{1<i\le L} \in {\mathcal T}$ 
with
$$
    \mathcal T \; = \; \left\{a-b : a,b \in[\omega_-^2,\omega_+^2]\right\}^{L-1}. 
$$ 
Let us denote by $\phi: \Omega \to \mathbb{R}\times \mathcal T$ the map $(\omega^2_i)_{1\le i\le L} \mapsto  (x,\tau)$. The map $\phi$ is injective. 
We will use the change of measure
\begin{equation}
	\Proba(\dd\omega_1^2\dots \dd \omega_L^2) 
 \; = \;  \rho(x,\tau)  \dd x  \dd \tau 
\end{equation}
where $\dd \tau= \prod_{j=2}^{L}  \dd\tau_j$  and 
\begin{equation}\label{eq: density}
	\rho(x,\tau) \; = \;  \chi_{\phi(\Omega)}(x,\tau) |\det M| \prod_{j=1}^L \mu(\omega^2_j)
\end{equation}
with $M$ a constant non-singular matrix that is the Jacobian of the change of variables.

It is useful to consider the random variable $Q:\Omega\to\R$ in the coordinates $(x,\tau)$ and we will simply write $Q=Q(x,\tau)$.     
As $x$ varies and $\tau$ remains fixed, the eigenvalues keep their ordering and (non)-degeneracy is conserved, which we exploit now. 
For any $\tau \in \mathcal T$, we can verify that the set 
$$
X_\tau= \{x \in [\omega_-^2,\omega_+^2]:  (x,\tau) \in \phi(\Omega) \}
$$
is a closed interval (possibly empty). 
Let $\tau$ be such that $X_\tau$ is non-empty and such that the spectrum of $\mathcal H_I$ is non-degenerate. 
We order the eigenvalues by increasing order and
we specify an $m$-tuple of distinct eigenvalues by choosing an $m$-tuple $t=(t_1,\ldots,t_m)$ of different elements in $\{1,\ldots,L\}$.  
Given such $\tau$ and $t$, we define 
\[
    Q_t(\cdot,\tau)\; :\; X_\tau \to \mathbb R,\;
    x \mapsto Q_t(x,\tau) \; = \; \sum_{k=1}^m \sigma_k \nu_{t_k}. 
\]
Since the derivative of any eigenvalue $\nu^2$ of $\mathcal H_I$ with respect to $x$ is equal to 1, and since the spectrum of $\mathcal H_I$ is bounded below by $\nu^2_->0$, the function $Q_t(\cdot,\tau)$ is smooth. 
We will make this more quantitative in the next section.

\subsection{Bound on Coefficients in the Taylor Expansion}

Let us now fix for this subsection some \(\tau \in \mathcal{T}\) as well as some $m$-tuple $t=(t_1,\dots,t_m)$ of different elements in $\{1,\dots,L\}$. 
For simplicity, we will not explicitly indicate the dependence on $\tau$ and $t$ in all the notation; however, as in this whole section, no constant will be allowed to depend on them. 
We will also write \(\nu^2_{t_j}\) simply as \(\nu^2_{j}\).

We consider the Taylor expansion of $Q_t(\cdot,\tau)$ around some point $x_0\in X_\tau$. 
A computation yields
\begin{equation}\label{eq: taylor expansion Q}
    Q_{t}(x,\tau) 
    \;=\;   
    Q_t(x_0,\tau) + 
    \sum_{k=1}^{m} v_k(x_0) (x-x_0)^k+ R_m(x,x_0)
\end{equation}
for all $x\in X_\tau$, where
$$
    v_k(x_0) \; = \; s_k \lscal b_k , a \rscal,   
    \qquad 
    s_k \; = \; 
    \frac{1}{2} \left(\frac{1}{2}-1\right) \dots \left(\frac{1}{2}-(k-1)\right)  
    \frac{1}{k!}, 
$$
$\lscal \cdot,\cdot \rscal$ being our notation for the scalar product in $\R^m$,   
and where $a,b_k $ are vectors in $\R^m$ given by 
$$
    a \;=\; (\sigma_1,\ldots,\sigma_m), 
    \qquad  
    b_k=  (\nu_1^{1-2k}, \ldots,  \nu_m^{1-2k})
$$
for $k\ge 1$. 
As $0 < \nu_-^2 \leq \nu^2_j \leq \nu_+^2$, there exists a constant $\Const$ such that the error term is bounded as
\begin{equation}\label{eq: bound remainder taylor}
    |R_m(x,x_0)| \;\leq\; \Const |x-x_0|^{m+1}.
\end{equation}

\begin{Lemma}\label{lem: bound on scalar product}
There is a constant $\const_0$ such that
\begin{equation}\label{eq: bound on scalar product}
    \max_{1 \leq k\leq m}|v_k(x_0)| \;\geq\; \const_0 \min_{i\ne j} |\nu^2_i-\nu^2_j|^m .
\end{equation}
\end{Lemma}
\begin{proof}
It is sufficient to prove the claim with $v_k(x_0)$ replaced by $\lscal b_k , a \rscal$, which we do now.
Let $A$ be the $m\times m$ matrix defined by 
$$
    A_{ij}\;=\; \nu_{j}^{1-2i}
$$
for $1\le i,j\le m$ and let us view $a$ as a column vector, i.e.\@ a $m\times 1$ matrix.
We have 
\begin{equation}\label{eq: relation to det}
    \sum_k \lscal a , b_k\rscal^2 
    \;=\; \lscal a ,  A^T A  a \rscal  
    \;\geq\; \|a\|^2 \lambda_{\min}(A^T A)   
    \;\geq\;  \|a\|^2 \frac{ (\det A)^2 }{\| A\|^{2(m-1)}} 
    \;\geq\; \const (\det A)^2
\end{equation}
where $\lambda_{\min}(A^T A) $ is the smallest eigenvalue of the positive matrix $A^T A$ and 
where the penultimate inequality follows from
$$
	\lambda_{\min}(A^T A) \| A\|^{2(m-1)} \;\geq\; \det (A^T A)  \; = \;  (\det A)^2 .
$$
Since $\nu_j A_{ij}$ is a Vandermonde matrix, the determinant of $A$ is computed explicitly as 
$$
    \det A \;=\;  \left(\prod_{i=1}^m \frac1{\nu_{i}}\right)    
    \prod_{1\leq i<j\leq m-1} \left(\frac1{\nu_i^2}- \frac1{\nu_j^2}\right) .
$$
We remember that $ \nu_-^2 \leq \nu_i^2 \leq \nu_+^2$, which gives us the lower bound
$$
|\det A| \;\geq\; \const \min_{i \ne j} |\nu_i^2-\nu_j^2|^m. 
$$
This yields the lemma upon using \eqref{eq: relation to det}.
\end{proof}

\subsection{Good Behavior of the Resulting Polynomial}\label{subsec: good beavior}

Given $\gamma>0$, let now $ B_\gamma \subset \Omega$ be the event where some level spacing is smaller than $\gamma$:
$$
    B_\gamma
    \;=\; 
    \{ \exists \nu^2 \neq (\nu^2)' \in \mathrm{Spec}(\mathcal H_L):  |\nu^2-(\nu^2)'| \leq \gamma \}.
$$
We notice that if $(x,\tau)\in B_\gamma$, then $(x',\tau)\in B_\gamma$ for all $x'\in X_\tau$. 
The following bound is well-known and due to Minami;  
see \cite{minami1996local,graf2007remark} for proofs and e.g.\@ Section~1 in~\cite{bellissard_2007} for a reformulation corresponding to the claim here. 
\begin{Lemma}\label{lem: bound on b}
There exists a constant $\Const$ such that for all $\delta > 0$, 
\begin{equation}\label{eq: minami}
	\Proba (B_\gamma) \;\leq\; \Const L^2\gamma.
\end{equation}
\end{Lemma}

For the following lemma, which is our main probabilistic estimate, let us once again fix some \(\tau \in \mathcal{T}\) as well as some $m$-tuple $t=(t_1,\dots,t_m)$ of different elements in $\{1,\dots,L\}$.

\begin{Lemma}\label{lem: cartan poly} 
Let $0<\alpha<1$, that doesn't depend on $t,\tau$, and assume that $\tau$ is such that $(x,\tau)\notin B_{\alpha^{1/m}}$ for all $x\in X_\tau$. 
Given $x_0\in X_\tau$, there exists an integer $p$ satisfying $1\le p \le m$, 
that may depend on $x_0,\tau,t$ but not on $\alpha$, 
as well as a constant $\Const$, such that the following holds:
If $\interval(x_0)$ is an interval centered around $x_0$ with length
\begin{equation}\label{eq: lenght interval assumption}
    |\interval (x_0)| \;\leq\; \epsilon^{1/(p+1)},
\end{equation}
then 
\begin{equation} \label{eq: bound on bad part of i}
    \left|\left\{ x \in\interval(x_0)\cap X_\tau 
    : |Q_{t}(x,\tau)| \leq \epsilon \right\} \right| 
    \;\leq\; \Const (\epsilon/\alpha)^{1/p}
\end{equation} 
where $|\cdot|$ denotes the Lebesgue measure on $\R$.
\end{Lemma}

\begin{proof}
Pick a point $x_0 \in X_\tau$.  
We let $1\leq p \leq m$ be the smallest integer such that $|v_p(x_0)| \geq 
\const_0 \alpha$  for $\const_0$ as in Lemma~\ref{lem: bound on scalar product}. 
The existence of $p$ is assured by this lemma. 
Let now $P_p$ be the polynomial of degree $p$ that is obtained by keeping only the first $p$ terms of the Taylor approximation \eqref{eq: taylor expansion Q} to $ Q_{t}(\cdot,\tau)$. 
On the one hand, thanks to the upper bounds $|v_{k}(x_0)|\leq \Const$ for $k=p+1,\ldots,m$ and the bound \eqref{eq: bound remainder taylor} on the remainder of the Taylor expansion, there exists a constant $\Const_1$ such that 
$$
| Q_{t}(x,\tau)-P_p(x)| \;\leq\; \Const_1 \epsilon
$$
for every $x\in\interval (x_0)\cap X_\tau$. 
On the other hand, 
$\frac{1}{v_p(x_0)}P_p$ is a monic polynomial, i.e.\ the coefficient of the leading term is $1$, and by a classic Theorem~\cite{cartan}, we have
\begin{equation}\label{eq: bound measure from cartan}
    \left| \{x\in\R  :   |P_p(x)| \leq 2\Const_1\epsilon \} \right|
    \;\leq\;  
    \Const (2\Const_1\epsilon/|v_p(x_0)|)^{1/p} 
    \;\le\; 
    \Const' (\epsilon/\alpha)^{1/p}.
\end{equation}
Therefore, we find that for every $x \in\interval(x_0)\cap X_\tau $, 
the bound $|Q_t(x,\tau)| \le \Const_1\epsilon$ implies the bound $P_p(x)\le 2\Const_1\epsilon$, 
hence \eqref{eq: bound on bad part of i} follows from \eqref{eq: bound measure from cartan}.  
\end{proof}

In the sequel, given some $0<\alpha<1$ and given some $\tau,t$ as above, 
we will say that an interval $\interval$, centered around some point $x_0\in X_\tau$, is $(\alpha,p)$-\emph{admissible} for some $1\le p\le m$, if $|\interval| = \epsilon^{1/(p+1)}$ and if the bound \eqref{eq: bound on bad part of i} is satisfied.

\subsection{Handling Small Values of $\mu$}
Some special care is needed when the density $\mu$ approaches 0. 
We introduce here some notations and gather some preliminary estimates to deal with this case. 
Let $\delta,\delta'>0$ and consider the sets
$$  
    w_\delta \;=\; \{x \in [\omega_-^2,\omega_+^2]: \mu(x)\leq \delta  \},  
    \qquad  
    w_{\delta,\delta'} \; = \; \{x \in [\omega_-^2,\omega_+^2]: \mathrm{dist}(x,w_\delta)\leq \delta'  \}.
$$
Writing $\mu(U) = \int_U \mu(x)\dd x$ for $U\subset[\omega_-^2,\omega_+^2]$, we state
\begin{Lemma} 
There exists a constant $\Const$ such that
$$
    \mu (w_\delta) \;\leq\; \Const\delta, 
    \qquad  
    \mu (w_{\delta,\delta'} ) \;\leq\; \Const (\delta+\delta').
$$
\end{Lemma}
\begin{proof}
First bound is immediate because the support is contained in a finite interval. 
For the second bound, we use that $\mu$ has a bounded Lipschitz coefficient $\Const_{\mathrm{Lip}}$, so that 
$w_{\delta,\delta'} \subset   w_{\delta + \Const_{\mathrm{Lip}}\delta'}$
and then one uses the first bound.
\end{proof}
We now also define global sets $W_{\delta} \subset W_{\delta,\delta'}\subset \Omega$ as follows:
An element $(\omega_1^2,\dots,\omega_L^2)\in \Omega$ belongs to $W_\delta$ or $W_{\delta,\delta'}$,
if there exists $i$ with $1\le i \le L$ such that $\omega_i^2 \in w_\delta$ or $\omega_i^2\in w_{\delta,\delta'}$, respectively. 
For later use, we note that there exists a constant $\Const$ such that
\begin{equation}\label{eq: bound proba w}
   \Proba(W_{\delta,\delta'}) \;\leq\; \Const L(\delta+\delta') 
\end{equation}
which follows from the above lemma. 

For the following lemma, we fix once again some \(\tau \in \mathcal{T}\) as well as some $m$-tuple $t=(t_1,\dots,t_m)$ of different elements in $\{1,\dots,L\}$: 
\begin{Lemma}\label{lem: functional inequality}
    Let $\interval \subset W_\delta^c \cap X_\tau$ be an interval which length satisfies 
    $|\interval| \le \delta/L$. 
    There exists a constant $\Const$ such that for any non-negative function $f$ on $\R$, 
    \[
    \int_\interval \dd x \, \rho(x,\tau) f(x)
    \; \le \; 
    \frac{\Const}{|\interval|} 
    \int_\interval \dd x \, \rho(x,\tau) \int_\interval \dd x \, f(x).
     \]
\end{Lemma}

\begin{proof}
    Let first $x,y\in\interval$, and let $(\omega_1^2,\dots,\omega_L^2) = \phi^{-1}(x,\tau)$ and $((\omega'_1)^2,\dots,(\omega'_L)^2) = \phi^{-1}(y,\tau)$. 
    Since $|\interval|\le \delta/L$, we have also $|(\omega_i')^2 - \omega_i^2|\le \delta/L$ for every $i=1,\dots,L$. 
    We compute 
    \[
    |\log \rho(y,\tau) - \log \rho (x,\tau)|
    \; \le \; 
    \sum_{i=1}^L \int_{\omega_i^2}^{(\omega_i')^2} \dd u \left|\frac{\mu'(u)}{\mu(u)}\right|
    \; \le \; \Const L \frac{\delta/L}{\delta} \; \le \; \Const'.  
    \]
    Therefore, there exists a constant $\Const$ such that $\max_{x\in\interval}\rho(x,\tau) \le \Const \min_{x\in\interval}\rho(x,\tau)$, and therefore 
    \begin{align*}
    \int_\interval \dd x \, \rho(x,\tau) f(x)
    \; &\le \; 
    \max_{x\in\interval} \rho(x,\tau) \int_\interval\dd x f(x)
    \; \le \; 
    \Const \min_{x\in\interval} \rho(x,\tau) \int_\interval\dd x f(x)\\
    \; &\le \; 
     \frac{\Const}{|\interval|} 
    \int_\interval \dd x \, \rho(x,\tau) \int_\interval \dd x \, f(x).
    \end{align*}
\end{proof}

\subsection{Concluding the Proof of Theorem~\ref{thm: denominator}}\label{sec: finishing the proof denominators}

\begin{proof}[Proof of Theorem~\ref{thm: denominator}]
Let now 
\begin{equation}\label{eq: def of parameters denominator}
    \delta' \; = \; \epsilon^{1/(m+1)}, 
    \qquad
    \delta \; = \; L\delta', 
    \qquad 
    \alpha \; = \; \epsilon^{1/2(m+1)}.
\end{equation}
We estimate
\begin{equation}\label{eq: proof denominaor 1st bound}
    \Proba(Q<\epsilon)
    \; \le \; 
    \Proba\left(\{Q<\epsilon\}\cap W^c_{\delta,\delta'}\cap B^c_{\alpha^{1/m}}\right)
    + \Proba(W_{\delta,\delta'})
    + \Proba(B_{\alpha^{1/m}}).
\end{equation}

Let us focus on the first term in the right hand side of this bound. 
Since $Q = \min_t|Q_t|$ where $t$ ranges over all $m$-tuples $t=(t_1,\dots,t_m)$ of different elements in $\{1,\dots,L\}$, and since the number of such $m$-tuples is bounded by $L^m$, a union bound yields
\begin{equation}\label{eq: union bound denominator}
    \Proba\left(\{Q<\epsilon\}\cap W^c_{\delta,\delta'}\cap B^c_{\alpha^{1/m}}\right)
    \; \le \;
    L^m \max_t \Proba \left(\{|Q_t|<\epsilon\}\cap W^c_{\delta,\delta'}\cap B^c_{\alpha^{1/m}}\right).
\end{equation}
In the sequel, we fix some admissible $m$-tuple $t$ of different elements in $\{1,\dots,L\}$ and we denote the event $\{|Q_t|<\epsilon\}$ by $G_\epsilon$ for brevity. We will not write explicitly the dependence on $t$.

Given $\tau\in\mathcal T$ be such that $(x,\tau)\notin B_{\alpha^{1/m}}$ for all $x\in X_\tau$, we define a finite family of intervals $(\interval_i)_i$  with the following properties: 
\begin{enumerate}
    \item 
    The interval $\interval_j$ is $(\alpha,p)$-admissible for some $1 \le p\le m$, cf.~the end of Section~\ref{subsec: good beavior}.

    \item
    $$
    \{ x\in X_\tau : (x,\tau) \notin W_{\delta,\delta'}\} 
    \;\subset\; \bigcup_{j} \interval_j \;\subset\; 
    \{ x \in \R : (x,\tau) \notin W_{\delta}\}.
    $$
    
    \item 
    Any $x\in X_\tau$ lies in the intersection of at most two intervals of the family. 
\end{enumerate}
Let us prove that such a family of intervals can indeed be constructed. 
First, by Lemma~\ref{lem: cartan poly}, it is possible to construct an $(\alpha,p)$-admissible interval centered around each point $x_0\in X_\tau \cap W_{\delta,\delta'}^c$, 
for some $1\le p \le m$ (that may depend on $x_0$). 
The resulting (infinite) family of intervals forms a cover of $X_\tau \cap W_{\delta,\delta'}^c$. 
Since $X_\tau \cap W_{\delta,\delta'}^c$ is a finite union of intervals and since $(\alpha,p)$-admissible intervals have length at least $\epsilon^{1/2}$, it is possible to extract a finite family that remains a cover of $X_\tau \cap W_{\delta,\delta'}^c$.
Second, since all intervals in this family have length at most $\epsilon^{1/(m+1)}$, and since $\delta' = \epsilon^{1/(m+1)}$, we conclude that every interval in the family has an empty intersection with $W_\delta$, hence the second condition above is satisfied.
Finally, it follows from elementary geometric considerations that if a point $x$ lies in the intersection of three intervals in this family, then one of these three intervals can be removed without affecting the family's coverage of $X_\tau \cap W_{\delta,\delta'}^c$. Using this recursively implies the last point above. 

With the above definition, we compute
\begin{align}
    \Proba( G_{\epsilon} \cap W^c_{\delta,\delta'} \cap B_{\alpha^{1/m}}^c )
    &\;=\;  
    \int_{\mathcal T} \dd\tau   \int_{X_\tau} \dd x\, 
    \rho \,\chi_{G_{\epsilon}} \,\chi_{W^c_{\delta,\delta'}} \,\chi_{B_{\alpha^{1/m}}^c} \nonumber\\
    &\;\leq \; 2 \int_{\mathcal T} \dd\tau \, \chi_{B_{\alpha^{1/m}}^c}  
    \sum_{j} \int_{\interval_j } \dd x\, \rho  \,\chi_{G_{\epsilon}} \nonumber\\  
    & \; \le \; \Const \int_{\mathcal T} \dd\tau \, \chi_{B_{\alpha^{1/m}}^c} \,  \sum_{j}
    \frac{1}{|\interval_j|} \int_{\interval_j} \dd x\, \rho
    \int_{\interval_j} \dd x\, \chi_{G_\epsilon} \nonumber
\end{align} 
where the last bound follows from Lemma~\ref{lem: functional inequality}, that applies since $\interval_j \subset W_\delta^c$ and $|\interval_j| \le \epsilon^{1/(m+1)}=\delta' = \delta/L$. 
Using now that the intervals $\interval_j$ are $(\alpha,p)$-admissible, we conclude from Lemma~\ref{lem: cartan poly} that, on $B_{\alpha^{1/m}}^c$, 
\[
    \frac1{|\interval_j|}\int_{\interval_j} \dd x  \, \chi_{G_\epsilon}
    \;\le\;
    \frac{(\epsilon/\alpha)^{1/p}}{\epsilon^{1/(p+1)}}
    \; \le \; 
    \max_{1\le p \le m}
    \frac{(\epsilon/\alpha)^{1/p}}{\epsilon^{1/(p+1)}}
    \; =: \; 
    \max_{1\le p \le m} Z_p(\alpha,\epsilon)
    \; =: \; 
    Z(\alpha,\epsilon)
\]
and we conclude thus that 
\[
    \Proba( G_{\epsilon} \cap W^c_{\delta,\delta'} \cap B_{\alpha^{1/m}}^c )
    \; \le \; 
    \Const Z(\alpha,\epsilon). 
\]

Inserting this last bound in \eqref{eq: union bound denominator} and \eqref{eq: proof denominaor 1st bound} yields 
\[
    \Proba(Q<\epsilon)
    \; \le \; 
    \Const L^m Z(\alpha,\epsilon) + \Const L (\delta + \delta') + \Const L^2 \alpha^{1/m}
\]
where the second term in the right hand side is obtained from \eqref{eq: bound proba w}, and the last one from \eqref{eq: minami}. 
We now check that $Z(\alpha,\epsilon)\le Z_m(\alpha,\epsilon)$, since 
$p\mapsto Z_p(\alpha,\epsilon)$ is increasing as long as $p\le m$ and $\alpha \ge \epsilon^{2/(m+1)}$, which holds true by \eqref{eq: def of parameters denominator}. 
Inserting the values for $\delta',\delta,\alpha$ from \eqref{eq: def of parameters denominator} yields the claim of the theorem. 
\end{proof}

\section{A Class of Random Variables} \label{sec: control of local observables}
In this section, we introduce a class of random variables and we prove their two main properties.

\subsection{Preliminaries on Anderson Localization}\label{sec: preliminary Anderson}
We state here some useful and well-known results on Anderson localization. 
Let $I\subset\Z$ be an interval.  
One says there is strong exponential dynamical localization of $\mathcal H_I$, if there are constants $\Const$ and $\xi$ (that do not depend on $I$) such that 
\begin{equation} \label{eq: bound on eigenfuntion correlator}
    \sup_{x\in I} \sum_{{ y: |x-y|\geq R}}  \E [(Q_I(x,y))^2] 
    \;\leq\; 
    \Const e^{-R/\xi}
\end{equation}
for any $R>0$, 
where $Q_I$ is the eigenfunction correlator
$$
    Q_I(x,y)\; =\; \sum_{k=1}^{|I|}  |\psi_k(x)\psi_k(y)|.
$$
The bound \eqref{eq: bound on eigenfuntion correlator} is proven in \cite{Kunz}.

Next, given an interval $I\subset\Z$, we let $W:I\to\R^{+}$ be a normalized \emph{positive weight function} if it satisfies
$$
    W(x)\; > \; 0 \quad \forall x \in I, 
    \qquad \sum_{x\in I} \frac1{W(x)} \; = \; 1.
$$
Given $x_0\in I$, we will consider the weight
\begin{equation}\label{eq: weight function}
    W(x) \; = \;  W_I(x,x_0) \; = \; \frac1{N_I} (1+(x-x_0)^2)
\end{equation}
where $N_I$ ensures normalization. 
We notice that there exist universal constants $\Const,c$ such that $c\le N_I\le \Const$ for all $I$ and all $x_0 \in I$.
Given such a weight, we say that a normalized function $\psi\in L^2(I)$ has localization center $x\in I$ if and only if
$$
	|\psi(x)|^2 \;\geq\; \frac1{W_I(x,x_0)}.
$$
It immediately follows that
\begin{enumerate}
    \item 
    Every normalized function in $L^2(I)$ has at least one localization center. 
    \item 
    The number of eigenfunctions who have $x$ as a localization center, is bounded by $W_I(x,x_0)$. 
\end{enumerate}
In the sequel, given an eigenfunction $\psi_k$ of $\mathcal H_I$, we will denote by $\loccen(k)$ (or sometimes by $\loccen(\psi_k)$) a localization center of $\psi_k$. 
By Theorem~7.4 in \cite{aizenman2015random}, the strong localization result in \eqref{eq: bound on eigenfuntion correlator} implies

\begin{Theorem} \label{thm: complete loc}
Given an interval $I\subset\Z$, given $x\in I$, and given the weight $W_I(\cdot,x)$ as in \eqref{eq: weight function} there exists a random variable $A_{I,x}\ge 0$ such that for any eigenfunction $\psi_k$ of $\mathcal H_I$
and any $y\in I$, 
\begin{equation}\label{eq: definition big a}
    |\psi_k(y)|^2 \; \leq \; A_{I,x} \left(W_I(\loccen(k),x)\right)^2 
    \ed^{-|y-\loccen(k)|/\xi}.
\end{equation}
\end{Theorem}
\noindent
In the remainder, the factor $\left(W_I(\loccen(k),x)\right)^2$ will feature as $\Const (1 + (x- \loccen(k))^4)$.

\subsection{$Z$-Collections}\label{subsec: Z collections}

Below, given a (possibly infinite) interval $I\subset \Z$, we will consider elements of $L^2(I)$ as elements of $L^2(\Z)$ by extending them by 0 outside $I$.  
Given $d\ge 1$ and eigenfunctions $\psi_{k_1},\dots,\psi_{k_d}$ of $\mathcal H_I$, 
we denote by $\loccen(k)$ the set $\{\loccen(k_1),\dots,\loccen(k_d)\}$, with $k=(k_1,\dots,k_d)$.
For $x\in\Z$, we will also use the notation 
\begin{equation}\label{eq: the definition of R}
    R(x,k) \; = \; \max \{ |x-y|, y\in\mathbf x (k)\}.
\end{equation}

We will now define \emph{$Z$-collections} of random variables. 
For this we need some preliminary definitions. 
Let $m\in\N^*$ be given, 
and let $\kappa_0,\dots,\kappa_7$ as well as $\Const_1,\dots,\Const_5$ be given positive numbers.

Let us assume that the two following families of non-negative random variables are given. 
Let an interval $I\subset\Z$ and $1 \le r \le m$ be given. 
First, let 
\begin{equation}\label{eq: def D variables}
    (D_{I,r}(k))_k\, , 
    \qquad 
    k \in\{1, \dots, |I|\}^{d(r)},
\end{equation}
for some $d(r)\in\N^*$,
be such that $D_{I,r}(k)$ are measurable functions of the eigenvalues $(\nu_{k_1}^2,\dots,\nu_{k_{d(r)}}^2)$ of $\mathcal H_I$.
Second, given also $x\in I$, let 
\begin{equation}\label{eq: def N variables}
    (N_{I,x,r}(k))_k\, , 
    \qquad 
    k \in\{1, \dots, |I|\}^{d(r)},
\end{equation}
be such that $N_{I,x,r}(k)$ are measurable functions of the eigenpairs $(\nu_{k_1}^2,\psi_{k_1}),\dots,(\nu_{k_{d(r)}}^2,\psi_{k_{d(r)}})$ of $\mathcal H_I$.

To elaborate on the properties of these random variables, we need to define a class of subsets of $(\N^*)^{d}$ for $d\in \N^*$.  
For a pair $(i,j)$ with $1\le i<j\le d$, we define the sets
\begin{equation*}
    z_{i,j}(d) \; = \; \big\{ k \in (\N^*)^{d} : k_i = k_j \big\}.
\end{equation*}
Let $\mathcal A(d)$ be the $\sigma$-algebra of subsets of $(\N^*)^{d}$ generated by the collection of sets $ (z_{i,j}(d))_{1\le i<j\le d}$, i.e.\@ by taking complements, unions and intersections of such sets.  
For an element $\mathcal Z$ of $\mathcal A(d)$ and an interval $I\subset\Z$, we also write $\mathcal Z_I=\mathcal Z \cap \{1,\ldots,|I|\}^d$.

Now, for each $1\leq r \leq m$, we fix two elements 
\begin{equation}\label{eq: Z 1 Z 2}
    \mathcal Z^{(1)}(r),\mathcal Z^{(2)}(r) \;\in\; \mathcal A(d(r)).
\end{equation}
We also define 
\begin{equation}\label{eq: Z union}
\mathcal Z(r) \; = \; \mathcal Z^{(1)}(r) \cap \mathcal Z^{(2)}(r).
\end{equation}
With these definitions, we can state the properties of the families of random variables $(D_{I,r}(k))_k$ and $(N_{I,x,r}(k))_k$: 

\begin{enumerate}
\item 
Given intervals $I, J\subset\Z$ and
given $1\le r \le m$, 
\begin{equation}\label{eq: bounds on Dr}
    D_{I,r}(k) \; \le \; \Const_1, 
    \qquad 
    |D_{I,r}(k) - D_{J,r}(k')| 
    \; \le \; \Const_2 
    \left(\sum_{j=1}^{d(r)} |\nu_{k_j}^2 - \nu_{k_{j'}}^2|^2\right)^{\kappa_0/2}
\end{equation}
for all $k\in\{1,\dots,|I|\}^{d(r)}$ and for all $k'\in\{1,\dots,|J|\}^{d(r)}$.

\item 
Given a finite interval $I\subset\Z$ and given 
$1 \le r \le m$,
\begin{equation}\label{eq: condition 1 Z}
    \Proba\left(\min_{k\in \mathcal Z^{(2)}_I(r)} D_{I,r}(k) \le \varepsilon \right) 
    \; \le \; 
    \Const_3 |I|^{\kappa_1} \varepsilon^{\kappa_2}
\end{equation}
for all $\varepsilon > 0$.

\item
Given an interval $I\subset\Z$, $x \in I$ and $1\le r \le m$, 
\begin{equation}\label{eq: condition 2 Z}
    N_{I,x,r}(k) \; \le \; \Const_4 A^{\kappa_3}_{I,x}\ed^{-\kappa_4 R(x,k)}
\end{equation}
for all $k\in\mathcal Z_I^{(1)}(r)$.

\item 
Given intervals $I\subset J\subset\Z$, $x \in I$ and $1\le r \le m$, 
\begin{multline}\label{eq: condition 3 Z}
    |N_{I,x,r}(k) - N_{J,x,r}(k')| 
    \; \le \; \\
    \Const_5 \left(A_{I,x} A_{J,x}\right)^{\kappa_5}
    \left( 1 + R(x,(k,k')) \right)^{\kappa_6}
    \left( \sum_{j=1}^{d(r)} |\nu^2_{k_j} - \nu^2_{k'_j}|^2 
    + \|\psi_{k_j} - \psi_{k_j'} \|^2 \right)^{\kappa_7/2} 
\end{multline}
for all $k\in\mathcal Z_I^{(1)}(r)$,
all $k'\in\mathcal Z_J^{(1)}(r)$,
and all choices of gauge for the eigenvectors.
\end{enumerate}

Assuming that Definitions~\eqref{eq: def D variables} to \eqref{eq: Z union} are given, and that Properties~\eqref{eq: bounds on Dr} to \eqref{eq: condition 3 Z} hold, we can state

\begin{Definition}[$Z$-collection]\label{def: Z collection}
We say that a collection $\{Z_I(x) :  I\subset \Z \text{ an interval},\, x\in I\}$ of random variables is a \emph{$Z-$collection} if
\begin{equation}\label{eq: form of Z(x)}
    Z_I(x) 
    \; = \;
    \sum_{r=1}^{m}\sum_{k\in\mathcal Z_I(r)}
    \frac{N_{I,x,r}(k)}{D_{I,r}(k)}.
\end{equation}
\end{Definition}

Several examples of $Z$-collections will be presented in Section~\ref{sec: proof of the three propositions}, with a detailed treatment provided in Section~\ref{sec: explicit Z collection}.
When $I=\Z$, we will use the notation 
\[
Z(x) := Z_\Z(x)
\]
for all $x\in\Z$.
For the remainder of this section and the next, we will assume that some $Z$-collection $\{Z_I(x)\}$ is given, and we will allow constants to depend on the numbers $m,\Const_{1-5}$, $\kappa_{0-7}$ introduced above.
The following basic properties will be central: 
\begin{Proposition}\label{thm: integrability}
Let 
$$
	\mu \; = \; \frac{\kappa_2}{\kappa_2 \kappa_3 + 2}.
$$
There exists a constant $\Const$ such that for any $M > 0$, for any interval $I\subset\Z$ and any $x\in I$, 
$$
	\Proba(Z_I(x) \geq  M) \;\le\; \frac{\Const}{M^{\mu}}.
$$
In particular, $\E((Z_I(x))^p)$ stays finite for any $0<p<\mu$. 
\end{Proposition}

\begin{Proposition}\label{thm: local approx}
There exist constants $\Const,c$ such that, for any interval $I\subset\Z$, for any $x\in I$ and for any $\ell\ge 0$, we have 
$$
    |Z_I(x)- Z_{I(x,\ell)}(x)| \;\leq\;   \Const \ed^{-c\ell}
$$
with probability at least $1- \Const \ed^{-c\ell}$, and where $I(x,\ell) = [x-\ell,x+\ell]\cap I$.
\end{Proposition}

\subsection{Proof of Proposition \ref{thm: integrability} and Proposition \ref{thm: local approx}}

As displayed in \eqref{eq: form of Z(x)}, $Z_I(x)$ is a sum over $r$ terms, with $r\le m$. 
Since we treat $m$ as a constant, 
it is enough to prove Propositions \ref{thm: integrability} and \ref{thm: local approx} for each given value of $r$ separately. 
In addition, we will restrict ourselves to the case $I=\Z$ and $x=0$. 
Proving the results for finite intervals $I\subset\Z$ and arbitrary values of $x\in I$ would be cumbersome due to boundaries but follows precisely the same lines.
For simplicity, we will omit $I,x=0,r$ from most of our notation and work with 
$$
    Z \; = \; \sum_{k\in \mathcal Z_\Z} \frac{N(k)}{D(k)}
    \qquad \text{and} \qquad
    Z_{\Lambda_\ell} \; = \; \sum_{k \in \mathcal Z_{\Lambda_\ell}} \frac{N(k)}{D(k)}
$$
with $\ell\in\N^*$.
The second expression enters only in the proof of Proposition~\ref{thm: local approx}.

Given $\ell \in \N^*$, let us define the smallest level spacing $\Delta_\ell$ of $\mathcal H_{\Lambda_\ell}$: 
\begin{equation}\label{eq: minimal level spacing}
    \Delta_\ell 
    \; = \; 
    \min\left\{|\nu_{k_1}^2 - \nu_{k_2}^2| : 1 \le k_1 \ne k_2 \le |\Lambda_\ell|, \nu_{k_1}^2,\nu_{k_2}^2 \in \mathrm{Spec}(\mathcal H_{\Lambda_\ell})\right\}.
\end{equation}
Given $\gamma_0,\gamma_1,\gamma_2 > 0$, we define the event 
\begin{equation}\label{eq: def of T ell}
    T_\ell 
    \; = \; 
    \left\{ \min_{k\in \mathcal Z^{(2)}_{\Lambda_\ell}}
    D(k) \ge \ed^{-\gamma_0 \ell}, \quad \Delta_\ell \ge \ed^{-\gamma_1 \ell}, 
    \quad A_{\ell} \le \ed^{\gamma_2 \ell}, \quad A_{\Z} \le \ed^{\gamma_2 \ell}\right\}
\end{equation}
where we have written $A_\ell$ for $A_{\Lambda_\ell, 0}$ and $A_\Z$ for $A_{\Z,0}$. 
By Property~\eqref{eq: condition 1 Z}, by Minami's estimate \eqref{eq: minami},
and by the fact that $\E(A_\ell),\E(A_\Z) \le \Const$ for some constant $\Const$, we deduce using Markov inequality that
\begin{equation}\label{eq: bound proba T l c}
    \Proba(T_\ell^c) \; \le \; \Const \ell^{2+\kappa_1} \ed^{-\min\{\kappa_2\gamma_0,\gamma_1,\gamma_2\}\ell}. 
\end{equation}

\begin{proof}[Proof of Proposition~\ref{thm: integrability}]
Our aim is to get a bound on the expression~\eqref{eq: form of Z(x)} with $I=\Z$. 
For this, let us introduce the random variable $\ell_0$ defined as 
\begin{equation}\label{eq: definition of l not}
	\ell_0\; = \; \max \{ \ell+1 \,:\,  T_{\ell} \text{ is not true} \} \vee \Const_0,
\end{equation}
where $\Const_0$ is a constant to be fixed later.
First, using the bound~\eqref{eq: condition 2 Z} together with the definition of $T_\ell$, we find that
\begin{equation}\label{eq: bound on N E}
    N(k) 
    \; \leq \; 
    \Const A_\Z^{\kappa_3} \ed^{-\kappa_4 R(0,k)} 
    \; \le \; 
    \Const \ed^{\gamma_2 \kappa_3 \ell_0 - \kappa_4 R(0,k)}
\end{equation}
for any $k\in(\N^*)^d$.

Second, let $\ell\ge \ell_0$ and let $k = (k_1,\dots,k_d)\in \mathcal Z^c_\Z$ be such that $R(0,k)\le \ell/2$ for some $\ell\ge \ell_0$.
We may invoke Lemma~\ref{lem: from global to local eigenpairs} in Appendix~\ref{sec: localization}: 
For $\ell$ large enough, there exist eigenpairs $(\psi_{k'_1},\nu^2_{k'_1})$, $\dots$, $(\psi_{k'_d},\nu^2_{k'_d})$ of $\mathcal H_{\Lambda_\ell}$ such that  
$$
    \left( \sum_{j=1}^d |\nu^2_{k_j} - \nu^2_{k'_j}|^2 
    + \|\psi_{k_j} - \psi_{k_j'} \|^2 \right)^{1/2} 
    \;\le\; 
    \Const \ed^{-\beta\ell} 
$$
for 
\begin{equation}\label{eq: constraint 1}
	4\beta = 1/3\xi - 2\gamma_1 - \gamma_2,
\end{equation} 
where $\xi$ is the inverse localization length featuring in Theorem~\ref{thm: complete loc}.
We assume that $\gamma_1,\gamma_2$ are small enough so that $\beta > 0$. 
Moreover, $k'=(k_1',\dots,k_d')\in\mathcal Z_{\Lambda_\ell}$.
Indeed, first, $k\in\mathcal Z_\Z^{(1)}$ implies $k'\in\mathcal Z_{\Lambda_\ell}^{(1)}$ by the definition of $\mathcal A(d)$ above~\eqref{eq: Z 1 Z 2}. 
Second, if $k\in\mathcal Z_\Z^{(2)}$ then $k'\in\mathcal Z_{\Lambda_\ell}^{(2)}$,
since otherwise two distinct eigenpairs in the $d$-tuple indexed by $(k_1,\dots,k_d)$ would have been mapped by Lemma~\ref{lem: from global to local eigenpairs} to the same eigenpair, which is impossible.

Using the condition \eqref{eq: bounds on Dr}, we obtain
\begin{equation}\label{eq: bound on ener diff}
	|D(k') - D(k)| \; \le \; \Const \ed^{-\kappa_0\beta \ell}
\end{equation}
with $k' = (k_1',\dots,k_d')$.
Hence, using the first condition in the definition \eqref{eq: def of T ell} of $T_\ell$ together with this bound, we deduce that 
\begin{equation}\label{eq: obtained bound on 1/D}
	\frac1{D(k)} \; \le \; \Const \ed^{\gamma_0 \ell}
\end{equation}
provided that $\gamma_0$ is chosen small enough such that 
\begin{equation}\label{eq: constraint 2}
	\gamma_0 \; < \; \kappa_0\beta
\end{equation}
and provided $\ell$ is large enough, which is guaranteed by taking $\Const_0$ large enough in the definition \eqref{eq: definition of l not} of $\ell_0$. 

Inserting the bounds \eqref{eq: bound on N E} and \eqref{eq: obtained bound on 1/D} into \eqref{eq: form of Z(x)} with $I=\Z$ yields
\begin{align*}
	Z
	\; &\le \; 
	\sum_{k: R(0,k)\leq \ell_0/2}  \frac{N(k)}{D(k)}
	+ \sum_{\ell=\ell_0+1}^{\infty}  \sum_{k: R(0,k)=\lfloor\ell/2\rfloor}  \frac{N(k)}{D(k)} \\
	& \; \le \; 
	\Const \sum_{k: R(0,k)\leq \ell_0/2} \ed^{(\gamma_2 \kappa_3 + \gamma_0) \ell_0}
	+ 
	\Const \sum_{\ell=\ell_0+1}^{\infty}  \sum_{k: R(0,k)=\lfloor\ell/2\rfloor} \ed^{\gamma_2 \kappa_3 \ell_0 + \gamma_0 \ell - \kappa_4 R(0,k)}\\
	& \; \le \; 
	\Const \ell_0^{3m}\ed^{(\gamma_2 \kappa_3 + \gamma_0) \ell_0} + \Const \sum_{\ell \ge 1} \ell^{3m} \ed^{-(\kappa_4/2 - \gamma_2 \kappa_3 - \gamma_0)\ell}
\end{align*}
where the last estimate follows from the fact that the number of eigenstates with localization center at distance $x$ from the origin grows at most quadratically with $x$, 
as stated at the beginning of this section. 
The second term in the right hand side of the above expression will be bounded by a constant provided that
\begin{equation}\label{eq: constraint 3}
	\kappa_4/2 - \gamma_2 \kappa_3 - \gamma_0 > 0. 
\end{equation}
We assume that this condition holds. 
Moreover, from the definition \eqref{eq: definition of l not} of $\ell_0$, the bound \eqref{eq: bound proba T l c} translates into 
\begin{equation}\label{eq: bound proba l not}
	\Proba(\ell_0 > \ell) \; \le \;  \Const \ell^{2+\kappa_1} \ed^{-\min\{\kappa_2\gamma_0,\gamma_1,\gamma_2\}\ell}
\end{equation}
for any $\ell \in \N^*$. 
Hence, for any $\delta >0$, we find a constant $\Const_\delta$ such that
\begin{equation}\label{eq: constraint 4}
	\Proba(Z\ge M) 
	\; \le \;
	\frac{\Const_\delta}{M^{\min\{\kappa_2\gamma_0,\gamma_1,\gamma_2\}/(\gamma_2\kappa_3 + \gamma_0) - \delta}}. 
\end{equation}

To conclude, we have to make sure that we can select $\gamma_0$, $\gamma_1$ and $\gamma_2$ in such a way that $\beta$ defined in \eqref{eq: constraint 1} is positive, 
that the constraints \eqref{eq: constraint 2} and \eqref{eq: constraint 3} are satisfied, 
and that the exponent in \eqref{eq: constraint 4} is lager than $\kappa_2/(\kappa_2\kappa_3 + 2)$ for $\delta>0$ small enough. 
Let us choose 
$$
	\gamma_2 = \kappa_2 \gamma_0, \qquad \gamma_1 = 2 \gamma_0.
$$
It then suffices to take $\gamma_0$ small enough so that 
$$
	\gamma_0 (\kappa_2+4) \; < \; 1/3\xi, 
	\qquad 
	\gamma_0 \; < \; \kappa_0\beta, 
	\qquad 
	\gamma_0 (\kappa_2\kappa_3 + 1) \; < \; \kappa_4/2.
$$
Doing so will guarantee that the exponent in \eqref{eq: constraint 4} is equal to $\kappa_2/(\kappa_2\kappa_3 + 1) - \delta$, 
that can be made smaller than $\kappa_2/(\kappa_2\kappa_3 + 2)$ by taking $\delta$ small enough. 
\end{proof}

\begin{proof}[Proof of Proposition~\ref{thm: local approx}]
Let us remind that the event $T_\ell$ is defined in \eqref{eq: def of T ell} for $\ell\in\N^*$. 
By \eqref{eq: bound proba T l c}, $\Proba(T_\ell^c)$ decays exponentially with $\ell$. 
Hence, it suffices thus to prove that $|Z(0) - Z_{\Lambda_\ell}(0)|$ decays exponentially with $\ell$ on $T_\ell$. For the remainder of this proof, we assume that $T_\ell$ holds.

The bijection between local and global eigenpairs established in Appendix~\ref{sec: localization} right before Lemma~\ref{lem: bijection local global}, can be readily generalized to the case of $d$-tuple of eigenpairs as we need here.
Let us assume that $\gamma_1$ and $\gamma_2$ are small enough so that condition \eqref{eq: condition on beta and gamma} is satisfied, i.e.\@ $\beta > 0$, with $\beta$ defined in \eqref{eq: constraint 1}, and $12\gamma_2\leq 1/\xi$. 
Let us also define the sets
$$
    \mathbf X \; = \; \{ k \in\mathcal Z_{\Lambda_\ell} : 
    |R(0,k)| \leq \ell/2 \}, 
    \qquad 
    \mathbf Y \; = \; \{k'\in \mathcal Z_\Z : |R(0,k')| \leq 2\ell/3 \}.
$$
There exists a bijection $\boldsymbol\Phi : \mathbf X \to \Phi(\mathbf X)\subset \mathbf Y$, with $\boldsymbol\Phi = \Phi\times \dots \times \Phi$, where $\Phi$ is defined in Appendix~\ref{sec: localization} right before Lemma~\ref{lem: bijection local global}. The restrictions to $\mathcal Z_{\Lambda_\ell}$ in the definition of $\mathbf X$ and to $\mathcal Z_{\Z}$ in the definition of $\mathbf Y$ is shown to be possible using the same argument as the one just below \eqref{eq: constraint 1}. 
Moreover, by Lemma~\ref{lem: bijection local global}, if $k'\in \mathcal Z_\Z \backslash \boldsymbol\Phi(\mathbf X)$, 
\begin{equation}\label{eq: large loc center}
    |R(0,k')| \;\ge\; \ell/4.  
\end{equation}

We decompose
$$
    Z - Z_{\Lambda_\ell}
    \; = \; 
    \sum_{k\in \mathbf X} 
    \left(\frac{N(k')}{D(k')} - \frac{N(k)}{D(k)}\right)
    + \sum_{k' \notin \boldsymbol\Phi(\mathbf X)} \frac{N(k')}{D(k')}
    - \sum_{k\notin \mathbf X} \frac{N(k)}{D(k)}
$$
where the notation $k' = \boldsymbol\Phi(k)$ is used in the first sum.
Thanks to \eqref{eq: large loc center} and the definition of $\mathbf X$, 
and assuming that $\gamma_0$ and $\gamma_2$ are small enough so that \eqref{eq: constraint 2} and \eqref{eq: constraint 3} hold, we proceed as in the proof of Proposition~\ref{thm: local approx}, and conclude that the last two sums in the formula above decay exponentially with $\ell$.
To get a bound on the terms in the first sum, we rewrite
$$
    \frac{N(k')}{D(k')} - \frac{N(k)}{D(k)}
    \; = \; 
    \frac{N(k')-N(k)}{D(k')} - \frac{N(k)}{D(k')D(k)}(D(k') - D(k)).
$$
From condition~\eqref{eq: condition 3 Z} and Lemma~\ref{lem: bijection local global}, we find 
\begin{multline*}
    |N(k')-N(k)| 
    \; \le \; 
    \Const (A_{\ell} A)^{\kappa_5} 
    (1 + R(0,(k,k')))^{\kappa_6}
    \left( \sum_{j=1}^d |\nu^2_{k_j} - \nu^2_{k'_j}|^2 
    + \|\psi_{k_j} - \psi_{k_j'} \|^2 \right)^{\kappa_7/2}\\
    \; \le \; 
    \Const \ell^{\kappa_6}
    \ed^{-(\beta \kappa_7 - (\kappa_5 + \kappa_6) \gamma_2)\ell}. 
\end{multline*}
In addition, the bound \eqref{eq: bound on ener diff} on $|D(k') - D(k)|$ established in the proof of Proposition~\ref{thm: local approx} remains in force, as well as the bound \eqref{eq: obtained bound on 1/D} provided $\ell$ is large enough, where $k$ can also be replaced by $k'$. 
This yields the bound 
$$
    \left|\frac{N(k')}{D(k')} - \frac{N(k)}{D(k)}\right|
    \; \le \; 
    \Const \left(\ell^{\kappa_6}\ed^{-(\beta\kappa_7 - (\kappa_5+\kappa_6)\gamma_2 - \gamma_0)\ell} + \ed^{-(\beta - \kappa_3\gamma_2 - 2\gamma_0)\ell}\right).
$$
Since the number of terms in the sum over $k\in \mathbf X$ is polynomial in $\ell$, 
we will obtain exponential decay in $\ell$ by taking $\gamma_0$, $\gamma_1$ and $\gamma_2$ small enough. 
\end{proof}

\section{Properties of $Z$-Collections}\label{sec: properties of the Z random variables}

In all this section, we assume that a $Z$-collection $\{Z_I(x)\}$ is given, and we derive properties that will be of direct use to prove our results. 
In Subsection~\ref{subsec: strong LLN}, we derive properties that will serve mainly for the proof of Proposition~\ref{pro: assumptions theorem 1}, 
while the properties shown in Subsection~\ref{subsec: bad events} will be useful for Propositions~\ref{pro: assumptions theorem 2} and~\ref{pro: assumptions theorem 3}.

\subsection{Strong Law of Large Numbers}\label{subsec: strong LLN}
  
We remind that we write $Z(x)$ for $Z_\Z(x)$ for $x\in \Z$. 
We start by establishing decay of correlations:  

\begin{Lemma}\label{lem: exp decay} 
Let us assume that the power $\mu$ in Proposition \ref{thm: integrability} satisfies $\mu>4$. 
There exists constants $\Const,c$ such that for any $x,y\in\Z$, 
$$ 
	\left|\E(Z(x)Z(y)) - \E(Z(x)) \E(Z(y))\right| \;\le\; \Const e^{-c|x-y|}   
$$ 
\end{Lemma}
\begin{proof} 
Let $\ell = |x-y|/3$, and  let $T$ be the event such the conclusion of Proposition~\ref{thm: local approx} holds for both $Z(x)$ and  $Z(y)$. 
We have $\Proba(T^c)\leq \Const e^{-c\ell}$.  
Let us denote by $Z_x(x)$ and $Z_y(y)$ the variables $Z_{I(x,\ell)}(x)$ and $Z_{I(y,\ell)}(y)$ respectively. 
We notice that $Z_x(x)$ and $Z_y(y)$ are independent. 
Then 
\begin{align}
\E(Z(x)Z(y)) \; = &\;
\E(Z(x)Z(y)(1-\chi(T)) )+\E(Z(x)Z(y)\chi(T) ))  \nonumber \\
\; = & \; 
\E(Z(x)Z(y)(1-\chi(T)) )+\E(Z_{x})(x)Z_{y}(y)\chi(T) )) + O(e^{-c \ell}) \nonumber \\
\; = & \; 
\E(Z(x)Z(y)(1-\chi(T)) )-
\E(Z_{x})(x)Z_{y}(y)(1-\chi(T) )) \nonumber \\
& \; + \E(Z_{x})(x)Z_{y}(y))+  O(e^{-c \ell})  \label{eq: splitting t}  
\end{align}  
The first term in the right hand side of \eqref{eq: splitting t} is estimated using twice Cauchy-Schwartz: 
$$
 	\E(Z(x)Z(y)(1-\chi(T)) ) \leq \E((1-\chi(T)))^{1/2} 
  	\E((Z(x)^4)^{1/4} \E((Z(y)^4)^{1/4} 
$$
Since we assumed that $\mu>4$ in Proposition~\ref{thm: integrability}, 
$\E(Z(x)^4),\E(Z(y)^4)<+\infty$, and so this expression is bounded by $\Const e^{-c\ell}$.  
The second term is estimated analogously. 
In the third term, the two variables are independent and hence this term equals 
\begin{align*}
 	\E(Z_{x}(x))) \E(Z_{y}(y) )) & =
	( \E(Z(x)))+  O(e^{-c\ell})) \times ( \E(Z(y)  +  O(e^{-c\ell})) \\
	& =
 	\E(Z(x)) \E(Z(y)  +  O(e^{-c\ell}))
\end{align*}  
This proves the Lemma. 
\end{proof} 

We deduce the strong law of large numbers from the decay of correlations: 
\begin{Proposition}\label{pro: lln}
Let us assume that the power $\mu$ in Proposition \ref{thm: integrability} satisfies $\mu>4$. 
There exists a constant $\mathrm C$ such that 
$$
\limsup_{L}\frac{1}{|\Lambda_L|}\sum_{x\in\Lambda_L} Z_{\Lambda_L}(x) \; \le \; \Const \qquad \text{a.s.}
$$
\end{Proposition} 

\begin{proof}
First, let us show that
$$
	\frac1{|\Lambda_L|}\sum_{x\in\Lambda_L} \left(Z_{\Lambda_L}(x) - Z(x)\right) \; \to \; 0 \qquad \text{a.s.}
$$
We can first get rid of boundary terms by noting that the sum over $x\in\Lambda_L$ may be replaced by a sum over $x$ such that $|x|\le L-L^{1/2}$, 
up to an error that vanishes almost surely in the limit $L\to\infty$, thanks to the Borel-Cantelli lemma and the condition $\mu > 4$. 
By the Borel-Cantelli lemma, it is then enough to show that for any $\varepsilon > 0$
\begin{equation}\label{eq: second limit lln}
	\sum_{L\ge 1}\Proba \left( \frac1{|\Lambda_L|}\left|\sum_{x:|x|\le L-L^{1/2}}Z_{\Lambda_L}(x) - Z(x)\right| \ge \varepsilon \right) \; < \; + \infty.
\end{equation}
Given $\varepsilon > 0$, Proposition~\ref{thm: local approx} implies that for sufficiently large $L$ relative to $\varepsilon$,  
\begin{align*}
	\Proba \left( \frac1{|\Lambda_L|}\left|\sum_{x:|x|\le L-L^{1/2}}Z_{\Lambda_L}(x) - Z(x)\right| \ge \varepsilon \right)
	\; &\le\; 
	\sum_{x:|x|\le L-L^{1/2}} \Proba(|Z_{\Lambda_L}(x)-Z(x)|\ge \varepsilon)\\
	\; &\le\; \Const\sum_{x\in\Lambda_L} \ed^{-cL^{1/2}} \; \le \; \Const L \ed^{-cL^{1/2}} .
\end{align*}
Hence \eqref{eq: second limit lln} holds. 

Second, let
$$
	z\; = \; \sup_{x\in\Z} \E[Z(x)], \qquad   \sigma \; = \; \sup_{x\in\Z} (\mathrm{Var}(Z(x)))^{1/2}
$$
with the notation $\mathrm{Var}(Y) = \E((Y - \E(Y))^2)$, 
and note that $z$ and $\sigma^2$ are bounded uniformly in $L$, by Theorem \ref{thm: integrability}. 
We write then 
$$
	\frac1{|\Lambda_L|}\sum_x Z(x) \; \leq \; z + \frac{\sigma}{|\Lambda_L|} \sum_{x\in\Lambda_L} \frac{Z(x) - \E(Z(x))}{(\mathrm{Var}(Z(x)))^{1/2}} 
$$
and the strong law of large numbers follows from the decay of correlations established in Lemma~\ref{lem: exp decay} and Theorem~6 in \cite{rlyons}. 
\end{proof}

\subsection{Control on Consecutive Bad Events}\label{subsec: bad events}
We introduce random weights $w_L^\pm$ 
that will allow us to control the random variable $\ell_0$ introduced in Proposition~\ref{pro: assumptions theorem 2} 
and the weight $w$ in Proposition~\ref{pro: assumptions theorem 3}.
Given $M> 0$ and $x\in\Lambda_L$, let 
\begin{align}
	&w_L^+(x,M) \; = \; \min\{|x-y|^2 : y > x, |Z_{\Lambda_L}(y)| \le M\}, \label{eq: def w+L}\\
	&w_L^-(x,M) \; = \; \min\{|x-y|^2 : y < x, |Z_{\Lambda_L}(y)| \le M\}, \label{eq: def w-L}
\end{align}
with the convention $Z_{\Lambda_L}(-L-1) = Z_{\Lambda_L}(L+1)=0$. 
We define similarly the weights $w^\pm$ in infinite volume, i.e.\@ with $\Lambda_L$ replaced by $\Z$ in the definitions above. 

We first prove
\begin{Lemma}\label{pro: consecutive bad events}
There exist constants $\Const,c$ and $M_0$ such that, for $M>M_0$,
	$$
		\Proba(w_L^\pm(x,M) > \ell^2) \;\le\; \Const \ed^{-c \ell^{1/2}} 
	$$
    for all $\ell \ge 0$ and for all $x\in\Lambda_L$. 
    The same bound holds for $w^\pm_L$ replaced by $w^\pm$.
\end{Lemma}
\begin{proof}
Let $M>0$.
We present the proof for $w^+(0,M)$, as the other cases are analogous. 
It suffices to prove the claim for integers $\ell\ge \ell_0$ for some fixed $\ell_0$.
We first bound 
$$
	\Proba(Z(1)>M,\dots Z(\ell)>M) \; \le \; \Proba(Z(x)>M,x\in \mathcal J(\ell))
$$
with $\mathcal J(\ell)$ the largest set of points in between $1$ and $\ell$ such that the distance between any two points is at least $\sqrt{\ell}$. 
Let us also consider the event 
$$
	T \; = \; \bigcap_{x\in\mathcal J(\ell)} T_x
$$
with $T_x$ the event such that the bound of Proposition~\ref{thm: local approx} holds {with $\ell$ replaced by $(1/3)\sqrt{\ell}$}.
By Proposition~\ref{thm: local approx}, it holds that 
\begin{equation}\label{eq: proof sub exp 1}
	\Proba(T^c) \; \le \; \Const \ell^{1/2} \ed^{-c\ell^{1/2}}.
\end{equation}
We now estimate  
$$
	\Proba(w^+(0,M) > \ell^2) \;\le\; \Proba(T^c) + \Proba(  Z(x)>M, x\in \mathcal J(\ell) \, | \, T) \Proba(T).
$$
To get a bound on the second term, we notice that, on $T$, 
$$
	Z(x) > M \quad \Rightarrow \quad Z_{I(x,{\frac{1}{3}\sqrt{\ell}})}(x) > M-1
$$
for all $x\in \mathcal J(\ell)$ provided $\ell_0$ is large enough, and that the variables $(Z_{I(x,\frac13\sqrt\ell)}(x))_{x\in\mathcal J(\ell)}$ are independent. Hence, for $M$ large enough, Proposition \ref{thm: integrability} implies 
\begin{equation}\label{eq: proof sub exp 2}
	\Proba(Z(x)>M,x\in \mathcal J(\ell) \, | \, T) \Proba(T) 
	\; \le \; 
	\Proba\left(Z_{I(x,{\frac{1}{3}\sqrt{\ell}})}(x) > M-1,x\in\mathcal J(\ell)\right)
	\; \le \; 
	\Const \ed^{-c\ell^{1/2}}.
\end{equation}
Combining \eqref{eq: proof sub exp 1} and \eqref{eq: proof sub exp 2} yields the claim. 
\end{proof}

It will be useful to consider weights $\overline w_L^\pm$ and $\overline w^\pm$ for which the constraint of being smaller than $M$ is mollified.
Let us define $\overline w^+$ explicitly.
By definition, for $M>0$ and $x\in\Z$, we have
$$
    w^+(x,M) 
    \; = \; 
    \sum_{y>x} |y-x|^2  
    \left(\prod_{x<y'<y} \chi_{Z(y') > M}\right) \chi_{Z(y) \le M}
$$ 
with the convention that the empty product is 1.
Let us introduce  smoothened indicator functions $\varphi_<$ and $\varphi_>$ satisfying $0 \leq \varphi_{<} \leq 1$ and $ 0 \leq \varphi_{>} \leq 1$, as well as
$$
\varphi_< (u,M)\; = \; 
\begin{cases}  1 &  \text{if} \,\,  u\le M \\   0 & \text{if} \,\,  u \geq M+1  \end{cases}, \qquad 
\varphi_> (u,M) \; = \; \begin{cases}  0 &  \text{if} \,\, u\le M-1 \\   1 & \text{if} \,\, u \geq M  \end{cases} 
$$
for all $u\in\R$.
We define 
\begin{equation}\label{eq: smooth w}
    \overline w^+(x,M) 
    \; = \; 
    \sum_{y>x} |y-x|^2 \left(\prod_{x<y'<y} \varphi_>(Z(y'),M)\right) \varphi_< (Z(y),M).
\end{equation}
We observe that 
\begin{equation}\label{eq: easy bound w w bar}
	w^+(x,M) \;\le\; \overline w^+(x,M).
\end{equation}
We can analogously define $\overline w^-$ and $\overline w_L^\pm$, and the property analogous to \eqref{eq: easy bound w w bar} will hold as well.

A statement analogous to Lemma~\ref{pro: consecutive bad events} holds for the mollified versions of our weights: 
\begin{Lemma}\label{pro: consecutive bad events overline}
There exist constants $\Const,c$ and $M_0$ such that, for $M>M_0$,
	$$
		\Proba(\overline w_L^\pm(x,M) > \ell^2) \;\le\; \Const \ed^{-c \ell^{1/3}} 
	$$
    for all $\ell \ge 0$ and for all $x\in\Lambda_L$. 
    The same bound holds for $\overline w^\pm_L$ replaced by $\overline w^\pm$.
\end{Lemma}
\begin{proof}
Let us prove the claim for $w^+$; the other cases are analogous. 
Let $x\in\Z$ and $M>0$.
We observe that, if $w^+ (x,M) \le \ell^2$ for some integer $\ell\ge 1$, 
then at most the $\ell$ first terms in the definition \eqref{eq: smooth w} of $\overline w^+(x,M+1)$ are non-zero, hence $\overline w^+(x,M+1)\le \ell^3$. 
Therefore
\begin{equation}\label{eq: bound probability w bar not bar}
    \Proba(\overline w^+(x,M+1)>\ell^3) \; \le \; \Proba(w^+(x,M)> \ell^{2}).
\end{equation}
The proof then follows from Lemma~\ref{pro: consecutive bad events}. 
\end{proof}

We can now prove decay of correlations for the variables $\overline w^\pm$:  
\begin{Lemma}\label{pro: decay of correlations w}
    There exist constants $\Const,c$ and $M_0$ such that, for all $M\ge M_0$ and for all $x,y\in\Z$, 
    $$
    |\E(\overline w^\pm(x,M) \overline w^\pm (y,M)) - \E(\overline w^\pm(x,M))\E(\overline w^\pm (y,M))|
    \; \le \; 
    \Const \ed^{- c |x-y|^{1/2}}.
    $$
\end{Lemma}
\begin{proof}
Let us fix some $M>0$ as well as some $x,y\in\Z$.
We consider the case of $\overline w^+$, as the case of $\overline  w^-$ is fully analogous.
Let $\ell = |x-y|/5$. 
We introduce a ``bad'' set $B_\ell$ defined as
\begin{align*}
	B_\ell 
	\; = \;& 
	\{ w^+(x,M-1)>\ell^2\}\cup\{ w^+(y,M-1)>\ell^2\} \\
	&\cup 
	\{\exists u\in\Z : (|u-x|\le \ell \text{ or } |u-y|\le \ell) \text{ and }|Z(u)-Z_\ell(u)|>\Const_0\ed^{-c_0\ell}\}
\end{align*}
where we have abbreviated $Z_{I(u,\ell)}(u)$ by $Z_\ell$, 
and where $\Const_0$ and $c_0$ are the constants provided by Proposition~\ref{thm: local approx}.

Provided $M$ is large enough, 
Lemma~\ref{pro: consecutive bad events} and Proposition~\ref{thm: local approx} imply that $\Proba(B_\ell) \le \Const \ed^{- c \ell^{1/2}}$ and, since $\overline w^+(x)$ has bounded moments of all orders by Lemma~\ref{pro: consecutive bad events overline}, we find
\begin{equation}\label{eq: approx product on set}
	\E(\overline  w^+(x,M) \overline  w^+(y,M))
	\; = \; 
	\E(\overline  w^+(x,M)\overline  w^+(y,M) \chi_{B_\ell^c}) + \mathcal O(\ed^{-c \ell^{1/2}}).   
\end{equation}
The bound  $w^+(x,M-1) \le \ell^2$ holds on $B_\ell^c$.
Hence, reasoning as in the proof of Lemma~\ref{pro: consecutive bad events overline}, we find that the sum over $y$ in the definition \eqref{eq: smooth w} of $\overline w^+(x)$ may be restricted to a sum over $y$ such that $x<y\le x+\ell$. 
We deduce from this that, on $B_\ell^c$, the variable $\overline w^+(x)$ is a Lipschitz continuous function of  $Z(y')$, $x<y'<x+\ell$, with Lipschitz constant bounded by $\Const \ell^4$. 
Indeed,  $\varphi_<$  and $\varphi_>$ have Lipschitz constant $\Const$, the number of factors in each product is at most $\ell$, and the number of terms in the sum is at most $\ell$.

This motivates the following definition: 
$$
	\overline w_{(\ell)}^+(x,M)
	\; = \; 
	\sum_{x<y\le x+\ell} |y-x|^2 \left(\prod_{x<y'<y} \varphi_>(Z_\ell(y'),M)\right) \varphi_< (Z_\ell(y),M).
$$
We define $\overline w_{(\ell)}^+(y,M)$ similarly. 
On $B_\ell^c$, by Lipschitz continuity of $\overline w^+$  and $\overline w^+$, and by Proposition~\ref{thm: local approx}, we then find
\begin{equation}\label{eq: lipschitz}
 	|\overline w^+(x,M) - \overline w_{(\ell)}^+(x,M)|, |\overline w^+(y,M) - \overline w_{(\ell)}^+(y,M)| \; \le \; \Const \ed^{-c \ell} .   
\end{equation}
We now use \eqref{eq: approx product on set}, \eqref{eq: lipschitz} and the fact that the variables $\overline w_{(\ell)}^+(x,M)$ and $\overline w_{(\ell)}^+(y,M)$ are independent, to conclude:
\begin{align*}
	\E(\overline w^+(x,M)\overline w^+(y,M)) 
	\; &= \;
	\E(\overline w_{(\ell)}^+(x,M) \overline w_{(\ell)}^+(y,M) \chi_{B_\ell^c})  + \mathcal O(\ed^{-c \ell^{1/2}})\\
	\; &= \; 
	\E(\overline w_{(\ell)}^+(x,M))\E(\overline w_{(\ell)}^+(y,M)) + \mathcal O(\ed^{-c \ell^{1/2}})\\
	\; &= \; 
	\E(\overline w_{(\ell)}^+(x,M)\chi_{B_\ell^c})\E(\overline w_{(\ell)}^+(y,M)\chi_{B_\ell^c}) + \mathcal O(\ed^{-c \ell^{1/2}})\\
	\; &= \; 
	\E(\overline w^+(x,M))\E(\overline w^+(y,M)) + \mathcal O(\ed^{-c \ell^{1/2}}).
\end{align*}
\end{proof}

Again, we deduce a strong law of large numbers from the decay of correlations: 
\begin{Proposition}\label{pro: lln for weight}
    There exist constants $\Const$ and $M_0$ such that, if $M\ge M_0$, 
    $$
    \limsup_{L\to\infty} \frac1{|\Lambda_L|}\sum_{x \in \Lambda_L} w_L^\pm(x,M) 
    \; \le \; \Const \qquad \text{a.s.}
    $$
\end{Proposition}

\begin{proof}
    Let $M$ be large enough for the claims below to hold.
    We first notice that, since $w_L^\pm \le \overline w_L^\pm $, we may prove the claim for $w_L^\pm$ replaced by $\overline w_L^\pm $.
    The proof parallels that of Proposition~\ref{pro: lln}.
    First, 
    \[
    \frac{1}{|\Lambda_L|} \sum_{x\in\Lambda_L} \overline w_L^+(x,M) - \overline w^+ (x,M)
    \; \to \; 0
    \qquad \text{a.s.}
    \]
    Indeed, just as in the proof of Proposition~\ref{pro: lln} and since our weights have bounded moments of all orders, we can focus on the terms with $|x|\le L - L^{1/2}$, and it then suffices to prove that, for any $\varepsilon > 0$, we can take $L$ large enough so that $\Proba(|\overline w_L^+(x,M) - \overline w^+(x,M)|>\varepsilon) \le \Const \ed^{- c L^{1/2}}$. 
    This is shown precisely as we showed that \eqref{eq: lipschitz} holds on $B_\ell^c$ in the proof of Lemma~\ref{pro: decay of correlations w}.
    The remainder of the proof is then a consequence of the decay of correlations established in Lemma~\ref{pro: decay of correlations w}. 
\end{proof}

\section{Completing the Proof}\label{sec: proof of the three propositions}

In this Section, we prove Propositions~\ref{pro: assumptions theorem 1} to \ref{pro: assumptions theorem 3}, 
thus completing the proof of Theorems~\ref{th: decorrelation} to \ref{th: Green Kubo}, 
and we indicate the few adaptations needed to prove Theorem~\ref{th: out of equilibrium}.
To begin, we provide an explicit construction of a $Z$-collection, as defined in Section~\ref{subsec: Z collections}, that is directly relevant for the proof of Proposition~\ref{pro: assumptions theorem 1}. 
This construction serves as a detailed blueprint for all $Z$-collections used in our proofs.  
Here and below, it is understood that the parameters $m, \kappa_0, \dots, \kappa_7, \Const_0, \dots, \Const_5$ appearing in the definition of $Z$-collections in Section~\ref{subsec: Z collections} can be assumed to be constants.

\subsection{An Instance of a $Z$-Collection}\label{sec: explicit Z collection}

Let us fix some $k_0\in\{0,\dots,|\Lambda_L|\}$, 
and let $u_{k_0}$ and $g_{k_0}$ be the functions that solve \eqref{eq: hyp 1 theo 1}. 
These functions are described in Section~\ref{sec: perturbative expansion}, and explicit expressions are provided in Proposition~\ref{pro: main result perturbation theory}.
In particular, $g_{k_0}$ is given by \eqref{eq: expression for fi} with $i=n+1$, i.e.\@ $g_{k_0} = f^{(n+1)}$: 
\begin{align*}
    g_{k_0}
    \; = &\; 
    \sum_{(s,t)\in \mathcal U_{n+1}} \sum_{(k,\sigma)\in \mathcal R} 
    \sum_{j=1}^4 \delta(k_0 - k_j)
    \delta(k_{s_1}-k_{t_1})\dots \delta(k_{s_{n}} - k_{t_{n}}) \delta(\sigma_{s_1}+\sigma_{t_1})\dots\delta(\sigma_{s_{n}}+\sigma_{t_{n}})\nonumber\\
    &\frac{\imag \sigma_j\sigma_{t_1}\dots \sigma_{t_{n}}
    \hat H_{\mathrm{an}}(k'_1)\dots \hat H_{\mathrm{an}}(k'_{n+1})}
    {\Delta(k'_1,\sigma'_1)\Delta((k'_1,k'_2),(\sigma'_1,\sigma'_2)) 
    \dots \Delta((k'_1,\dots,k'_{n}),(\sigma'_1,\dots,\sigma'_{n}))}
    \; (\tilde a_k^\sigma(s,t) - \langle\tilde a_k^\sigma (s,t)\rangle).
\end{align*}
To derive this expression from \eqref{eq: expression for fi}, we used the following facts. 
First, we used the expression \eqref{eq: coef f1 k0} for $\hat f(k_1',\sigma_1')$. 
Second, we made it explicit that $\tilde a_k^\sigma$ depends on $(s,t)$, i.e.\@ on the choice of contractions. 
Finally, we used that $\langle g_k \rangle = 0$ to subtract the averages $\langle\tilde a_k^\sigma (s,t)\rangle$. 
It is not necessary to do this here, but will become so in the proof of Proposition~\ref{pro: assumptions theorem 3}, and we chose to showcase how this property can be used in advance. 

We further modify the above expression in two ways. 
First, instead of considering $(k,\sigma)\in\mathcal R$ in the second summation, 
we fix some $\sigma$ and consider $k\in\mathcal R(\sigma)$ with $k\in\mathcal R(\sigma)$ if and only if $(k,\sigma)\in\mathcal R$. 
Second, we replace the product over delta factors by imposing $k$ to belong to a certain set: 
Given $s,t,j,\sigma$, we say that $k \in \mathcal C(k_0,s,t,j,\sigma)$ if and only if 
\begin{equation}\label{eq: product of deltas}
    k_0 = k_j, \quad 
    k_{s_1} = k_{t_1}, \quad \dots, \quad k_{s_n} = k_{t_n}, \quad 
    \sigma_{s_1} = - \sigma_{t_1}, \quad \dots , \quad \sigma_{s_n} = \sigma_{t_n}. 
\end{equation}
This set is thus empty if the constraints over $\sigma$ are not satisfied. 
With these definitions, we rewrite
\begin{multline}
    g_{k_0}
    \; = \; 
    \sum_{\substack{(s,t)\in \mathcal U_{n+1},\\ 1 \le j \le 4}} \sum_{\sigma}
    \sum_{k\in \mathcal R(\sigma)\cap\mathcal C(k_0,s,t,j,\sigma)} \\
    \frac{\imag\sigma_j\sigma_{t_1}\dots \sigma_{t_{n}}
    \hat H_{\mathrm{an}}(k'_1)\dots \hat H_{\mathrm{an}}(k'_{n+1})}
    {\Delta(k'_1,\sigma'_1)
    \dots \Delta((k'_1,\dots,k'_{n}),(\sigma'_1,\dots,\sigma'_{n}))}
    \; (\tilde a_k^\sigma(s,t) - \langle\tilde a_k^\sigma (s,t)\rangle).
    \label{eq: 2nd form g k0 instance}
\end{multline}

Let now $x\in\Lambda_L$ and let $0<q<1$ be a fractional exponent. 
We will define $Z_{\Lambda_L}(x)$ to be an upper bound on $\sum_{k_0=1}^{|\Lambda_L|} |\psi_{k_0}(x)|^2 \langle g_{k_0}^2\rangle^q$. 
To estimate $\langle g_{k_0}^2 \rangle$, 
we use the second item in Corollary~\ref{cor: decay of correlations} in Appendix~\ref{sec: gibbs},
\begin{equation}\label{eq: crude bound on the a}
     |\langle \tilde a_k^\sigma(s,t) ; 
     \tilde a_{\underline k}^{\underline\sigma}(\underline s,\underline t)\rangle| 
    \; \le \; \Const.
\end{equation}
Using the subadditivity of the map $v\mapsto v^q$ ($v\ge 0$), we find 
\begin{multline}\label{eq: 1st concrete Z collection}
    \sum_{k_0=1}^{|\Lambda_L|} |\psi_{k_0}(x)|^2 \langle g_{k_0}^2\rangle^q
    \; \le \; 
    \Const
    \sum_{\substack{(s,t),(\underline s,\underline t)\in \mathcal U_{n+1},\\ j,\underline j}} 
    \sum_{\sigma,\underline\sigma}
    \sum_{k_0}
    \sum_{\substack{k\in \mathcal R(\sigma)\cap\mathcal C(k_0,s,t,j,\sigma)\\
    \underline k\in \mathcal R(\underline\sigma)\cap\mathcal C(k_0,\underline s,\underline t,\underline j,\underline\sigma)}} \\
    \frac{|\psi_{k_0}(x)|^{2q}|\hat H_{\mathrm{an}}(k'_1)\dots \hat H_{\mathrm{an}}(k'_{n+1}) \hat H_{\mathrm{an}}(\underline k'_1)\dots \hat H_{\mathrm{an}}(\underline k'_{n+1})|^q}
    {|\Delta(k'_1,\sigma'_1)
    \dots \Delta((k'_1,\dots,k'_{n}),(\sigma'_1,\dots,\sigma'_{n}))
    \Delta(\underline k'_1,\underline\sigma'_1)
    \dots \Delta((\underline k'_1,\dots,\underline k'_{n}),(\underline\sigma'_1,\dots,\underline\sigma'_{n}))|^q}.
\end{multline}

The right hand side of this expression is taken to define $Z_{\Lambda_L}(x)$.
We observe that the expression \eqref{eq: 1st concrete Z collection} still makes sense if $\Lambda_L$ is replaced by any interval $I\subset\Z$, provided $x\in I$ and the eigenpairs are interpreted as eigenpairs of $\mathcal H_I$. 
Doing this replacement yields thus a definition of a collection of variables $Z_I(x)$, with $x\in I$ and $I\subset\Z$ an interval. 
We now proceed to show that this forms a $Z$-collection, as defined in Section~\ref{subsec: Z collections}, and that $\mu = \kappa_2/(\kappa_2\kappa_3 +2)$ can be made larger than $4$ by taking $q$ small enough. 

Let us first identify the various elements of the definition: 
\begin{enumerate}
    \item 
    The sum over $r$ in \eqref{eq: form of Z(x)} represents the two first summations in \eqref{eq: 1st concrete Z collection}, i.e. the sums over $r = (s,t,\underline s,\underline t,j,\underline j,\sigma,\underline\sigma)$. This is a finite sum.

    \item 
    To avoid confusion, let us denote by $\mathbf k$ what is denoted by $k$ in Section~\ref{subsec: Z collections}, and we rewrite the two last sums in \eqref{eq: 1st concrete Z collection} as a sum over $\mathbf k = (k_0,k,\underline k_0,\underline k)$. So, $\mathbf k$ are $d$-tuples with 
    $$
        d \; = \; 2 (1 + d_1 + 4n) \; = \; 2 (1 + 4(n+1)).
    $$
    For every $r$, we set $d(r) = d$ as above, with $d(r)$ as defined in Section~\ref{subsec: Z collections}. 
    The fact that $d(r)$ is independent of $r$ will not hold true for all $Z$-collections that we will have to consider. 

    \item 
    We define the random variables $D_{I,r}(\mathbf k)$ to be the denominator in \eqref{eq: 1st concrete Z collection}, and the variables $N_{I,x,r}(\mathbf k)$ to be the numerator in \eqref{eq: 1st concrete Z collection}. It is a direct consequence of their definitions that these variables are measurable functions of the eigenvalues or eigenpairs of $\mathcal H_I$ respectively. 

    \item 
    The set $\mathcal Z_I^{(1)}(r)$ is such that $\mathbf k = (k_0,k,\underline k) \in \mathcal Z_I^{(1)}(r)$ if and only if $k\in \mathcal C(k_0,s,t,j,\sigma)$ and $\underline k \in C(k_0,\underline s,\underline t,\underline j,\underline\sigma)$. 
    Our definition of $\mathcal C$ guarantees that $\mathcal Z_I^{(1)}(r)$ is an element of $\mathcal A(d(r))$, cf.~\eqref{eq: Z 1 Z 2}. 

    \item 
    Similarly, the set $\mathcal Z_I^{(2)}(r)$ is such that $\mathbf k = (k_0,k,\underline k) \in \mathcal Z_I^{(1)}(r)$ if and only if $k\in\mathcal R(\sigma)$ and $\underline k\in\mathcal R(\underline\sigma)$. The definition of $\mathcal R$ in \eqref{eq: definition of the R set} implies that $\mathcal Z^{(2)}_I(r)$ is an element of $\mathcal A(d(r))$, cf.~\eqref{eq: Z 1 Z 2}. 
\end{enumerate}

With all these definitions, we prove Properties \eqref{eq: bounds on Dr} to \eqref{eq: condition 3 Z} and that $\mu\to\infty$ as $q\to 0$: 
\begin{enumerate}
    \item
    Property \eqref{eq: bounds on Dr}. 
    Since, the spectrum of $\mathcal H_I$ is bounded as in \eqref{eq: nu - nu +}, the factors $\Delta$, as defined in \eqref{eq: definition of denominator}, are smooth and bounded functions of the eigenvalues of $\mathcal H_I$. Therefore, \eqref{eq: bounds on Dr} holds with $\kappa_0 = q$. 

    \item\label{it: 2nd property}
    Property~\eqref{eq: condition 1 Z}. 
    We first use the bound 
    \begin{equation}\label{eq: proba to be bounded for cond 1 Z}
        \Proba\left(\min_{\mathbf k \in \mathcal Z_I^{(2)}(r)}D_r(\mathbf k) \le \varepsilon \right)
        \; \le \; 
        \sum_{i=1}^{2n} 
        \Proba\left(\min_{\mathbf k \in \mathcal Z_I^{(2)}(r)}|\Delta_i|\le \varepsilon^{\frac{1}{2n q}}\right)
    \end{equation}
    where we have used the notation $\Delta_i$ with $1\le i \le 2n$ to denote each of the $\Delta$ factors in the denominator of  \eqref{eq: 1st concrete Z collection}. 
    Each of the random variables $\min_{\mathbf k \in \mathcal Z_I^{(2)}(r)}|\Delta_i|$ is of the form $Q$, as defined in \eqref{eq: Q random variable}, with $m$ there being at most $4n$.  Hence, Theorem~\ref{thm: denominator} yields 
    \[
    (\ref{eq: proba to be bounded for cond 1 Z})
    \; \le \;
    \Const n |I|^{4n} \varepsilon^{\frac{1}{2nq 2 (4n)(4n+1)}}.
    \]
    This establishes~\eqref{eq: condition 1 Z}, with $\kappa_2\to\infty$ as $q\to 0$. 
  
    \item\label{it: 3rd property} Property \eqref{eq: condition 2 Z}. 
    The expression for $\hat H_{\an}$ is provided in \eqref{eq: H an k}. 
    From the bound \eqref{eq: definition big a} on the components of eigenvectors, 
    we deduce the generic bound 
    \[
    |\hat H_{\an}(l)|
    \; \le \;
    \Const A_{I,x}^2
    \left(1 + (R(x,(l)))^8\right) 
    \ed^{-c \max_{1 \le i,i' \le 4}|\mathbf x(l_i) - \mathbf x(l_{i'})|}
    \]
    for all $l = (l_1,\dots,l_4)\in\{1,\dots,|I|\}^4$, 
    and with $R$ as defined in \eqref{eq: the definition of R}. 
    Using this bound for all factors $\hat H_{\an}$ in \eqref{eq: 1st concrete Z collection} and using again \eqref{eq: definition big a} for the factor $|\psi_{k_0}(x)|^2$, we find 
    \begin{multline*}
    (N_{I,x,r}(\mathbf k))^{1/q}
    \; \le \; 
    \Const A_{I,x}^{1 + 4(n+1)}
    \left(1 + (R(x,\mathbf k))^{4 + 16(n+1)} \right)\\
    \ed^{-c\left( |x-\loccen(k_0)| + \max_{1 \le i,i' \le 4}|
    \mathbf x((k'_1)_i) - \mathbf x((k'_{1})_{i'})| + \dots 
    + \max_{1 \le i,i' \le 4}|
    \mathbf x((\underline k'_{n+1})_i) - \mathbf x((\underline k'_{n+11})_{i'})|\right)}.
    \end{multline*}
    taking into account the connectivity implied by the constraints 
    $k\in\mathcal C(k_0,s,t,j,\sigma)$ and $k'\in\mathcal C(k_0,\underline s,\underline t,\underline j,\underline\sigma)$, we find that the last exponent is lower-bounded by $R(x,\mathbf k)$. 
    Hence, \eqref{eq: condition 2 Z} is established with 
    \[
        \kappa_3 \; = \; q(1 + 4(n+1)), 
        \qquad 
        \kappa_4 \; = \; cq.
    \]

    \item 
    Property \eqref{eq: condition 3 Z}. 
    To avoid clutter, we assume that the product of eigenfrequencies in the denominator of the expression \eqref{eq: H an k} that defines $\hat H_{\an}$ has been absorbed into the constant $\Const$ in \eqref{eq: 1st concrete Z collection}, which is possible thanks to the bound \eqref{eq: nu - nu +} on the spectrum. 
    Thanks to the generic bound $||a|^q - |b|^q|\le |a-b|^q$ for any $a,b\in\R$, we see that we need to get a bound on the absolute value of the difference 
    \begin{multline*}
    \sum_{y_1,\dots,y_{2(n+1)}}
    (\psi_{k_0}(x))^2 \psi_{k_1}(y_1)\dots \psi_{k_4}(y_1) \psi_{k_5}(y_2) \dots
    \psi_{\underline k_{4(n+1)}}(y_{2(n+1)}) - \\
    (\psi_{k_0}'(x))^2 \psi_{k_1}'(y_1)\dots \psi_{k_4}'(y_1) \psi_{k_5}'(y_2) \dots
    \psi_{\underline k_{4(n+1)}}'(y_{2(n+1)}).
    \end{multline*}
    Using the discrete analog of Leibniz product rule, this difference is rewritten as a sum over $2(n+2)$ differences, each of them featuring a single difference of the type $\psi_{k_i}(y_j) - \psi_{k_i}'(y_j)$. 
    The absolute value of each of these differences can be bounded by $\|\psi_{k_i} - \psi_{k_i}'\|$, which can then be factorized out of the sum over $y_1,\dots,y_{2(n+1)}$.
    We are left with a sum over $2(n+2)$ terms, each of which is a sum over $y_1,\dots,y_{2(n+1)}$. They can all be bounded by 1. 
    To see this, note that in the sum over $y$ in the expression \eqref{eq: H an k} defining $\hat{H}_{\an}$, we can bound any two out of the four factors by $1$ and still conclude that the sum is bounded by $1$.
    This shows \eqref{eq: condition 3 Z} with $\kappa_5 = \kappa_6 = 0$ and $\kappa_7 = q$. 
    Note that later, it will not always be possible to obtain a deterministic bound, and we will have to take $\kappa_5,\kappa_6$ non-zero.

    \item 
    As shown in items \ref{it: 2nd property} and $\ref{it: 3rd property}$ above, $1/\kappa_2,\kappa_3 \to 0$ as $q\to\infty$, which guarantees $\mu>4$ for $q$ small enough. 
\end{enumerate}
We have thus established that the collection $(Z_{I}(x))_{I,x}$ defined bellow \eqref{eq: 1st concrete Z collection} is a $Z$-collection and that $\mu> 4$ for $q>0$ small enough.

\subsection{Proof of Proposition~\ref{pro: assumptions theorem 1}}

As in Section~\ref{sec: explicit Z collection}, let us fix some $k_0\in\{0,\dots,|\Lambda_L|\}$, and let $u_{k_0}$ and $g_{k_0}$ be the functions that solve \eqref{eq: hyp 1 theo 1}. 
In the same way as we built there a $Z$-collection based on the function $g_{k_0}$, we can construct a $Z$-collection based on $u_{k_0}$. 
The only extra difficulty is that $u_{k_0}$ is not directly given by \eqref{eq: expression for ui} for some value of $i$, but instead $u_{k_0} = u^{(1)} + \dots + \lambda^{n-1}u^{(n)}$. 
This implies the presence of an extra finite summation over $1\le i,\underline i\le n$ in \eqref{eq: 1st concrete Z collection}. These new indices are incorporated into the index $r$. A consequence of this change is that the dimension $d(r)$ of the $\mathbf k$-tuples does now indeed depend on $r$. With these modifications, the proof in Section~\ref{sec: explicit Z collection} remains valid and carries over directly.

Hence, we conclude that there exists a $Z$-collection $(Z_I(x))_{I,x}$ such that
\begin{equation}\label{eq: bound ergodic average}
    \frac{1}{|\Lambda_L|} \sum_{k_0 = 1}^{|\Lambda_L|} 
    \left(\langle g_{k_0}^2 \rangle^q + \langle u_{k_0}^2 \rangle^q\right)
    \; = \; 
    \frac{1}{|\Lambda_L|} \sum_{x\in\Lambda_L} \sum_{k_0 = 1}^{|\Lambda_L|} 
    |\psi_{k_0}(x)|^2
    \left(\langle g_{k_0}^2 \rangle^q + \langle u_{k_0}^2 \rangle^q\right)
    \; \le \;
    \frac{1}{|\Lambda_L|} \sum_{x\in\Lambda_L} Z_{\Lambda_L}(x)
\end{equation}
with $\mu > 4$ provided $q>0$ is taken small enough. 
Let $M> 0$. 
On the one hand, the limit superior of the left hand side of \eqref{eq: bound ergodic average} is lower bounded by $M^q \nu(M)$. 
On the other hand, by the ergodic theorem in Proposition~\ref{pro: lln}, 
there exists a constant $\Const$ so that the limit superior of the right hand side of  \eqref{eq: bound ergodic average} is upper bounded by $\Const$. 
Hence 
$$
	\nu(M) \;\le \; \frac{\Const}{M^q}. 
$$

\subsection{Proof of Proposition~\ref{pro: assumptions theorem 2}}\label{subsec: proof prop 2}
This time, we fix some $x\in\Lambda_L$, 
we let $g_x$ and $u_x$ be the functions that solve \eqref{eq: hyp 1 theo 2}, 
and we construct a $Z$-collection $(Z_I(x))_{I,x}$ such that 
\[
    \langle g_x^2 \rangle + \langle u_x^2 \rangle 
    \; \le \;
    Z_{\Lambda_L}(x).
\]
This will suffice to establish Proposition~\ref{pro: assumptions theorem 2}. 
Indeed, for all $M>0$, $\ell_{0,L}(M) \le( w_L^+(0,M))^{1/2}$, with $\ell_0 = \ell_{0,L}(M)$ defined in \eqref{eq: def ell 0 proposition 2}, and with $w^+_L(0,M)$ defined in \eqref{eq: def w+L}.
Therefore the claim \eqref{eq: tail of x not distribution} follows from Lemma~\ref{pro: consecutive bad events}.
In addition, the fact that the bound is still satisfied for \(\limsup_{L \to \infty} \ell_{0,L}\) follows from \(\limsup_{L \to \infty} w_L^+(0, M) \leq w^+(0, M-1)\). This inequality, in turn, follows from the observation that \(w_L^+(0, M) \leq w^+(0, M-1)\) holds on the event where \(w^+(0, M-1) \leq (L/2)^2\) and \(|Z_L(x) - Z(x)| < 1/2\) for all \(1 \leq x \leq L/2\), and the probability of this event decays stretched exponentially with \(L\) by Lemma~\ref{pro: consecutive bad events} and Proposition~\ref{thm: local approx}.

To construct the above $Z$-collection, we follow the same steps as in Section~\ref{sec: explicit Z collection} and arrive at a bound analogous to \eqref{eq: 1st concrete Z collection}:
\begin{multline}\label{eq: 2nd concrete Z collection}
    \langle g_{x}^2\rangle
    \; \le \; 
    \Const
    \sum_{(s,t),(\underline s,\underline t)\in \mathcal U_{n+1}} 
    \sum_{\sigma,\underline\sigma}
    \sum_{\substack{k\in \mathcal R(\sigma)\cap\mathcal C(s,t,\sigma)\\
    \underline k\in \mathcal R(\underline\sigma)\cap\mathcal C(\underline s,\underline t,\underline\sigma)}} \\
    \frac{|\hat f_{x}(k'_1)\hat H_{\an}(k_2')\dots \hat H_{\mathrm{an}}(k'_{n+1}) 
    \hat f_{x}(\underline k'_1)\hat H_{\an}(\underline k_2')\dots \hat H_{\mathrm{an}}(\underline k'_{n+1})|}
    {|\Delta(k'_1,\sigma'_1)
    \dots \Delta((k'_1,\dots,k'_{n}),(\sigma'_1,\dots,\sigma'_{n}))
    \Delta(\underline k'_1,\underline\sigma'_1)
    \dots \Delta((\underline k'_1,\dots,\underline k'_{n}),(\underline\sigma'_1,\dots,\underline\sigma'_{n}))|}
\end{multline}
with 
\begin{equation}\label{eq: def f x k 1 prime}
    \hat f_x (k'_1) \; = \; (\psi_{k_1}(x-1) - \psi_{k_1}(x)) \psi_{k_2}(x),  
    \qquad
    k_1' \; = \;  (k_1,k_2). 
\end{equation}
In \eqref{eq: 2nd concrete Z collection}, the set $\mathcal C(s,t,\sigma)$ contains the values of $k$ that satisfy the constraint \eqref{eq: product of deltas}, excluding the equality $k_0 = k_j$, 
and the expression for $\hat f_x(k_1')$ stems from \eqref{eq: coef f1 j}, which is now used as the definition for $f$ in \eqref{eq: expression for fi}. 
We note that the part involving $\delta$-factors in \eqref{eq: coef f1 j} can be omitted as this is already implied by the constraint $(s,t),(\underline s,\underline t) \in \mathcal U_{n+1}$.

At this point, we can take over the arguments in Section~\ref{sec: explicit Z collection} with minor adaptations: 
The right-hand side of \eqref{eq: 2nd concrete Z collection} still makes sense for eigenpairs of $\mathcal H_I$ on any interval $I\subset\Z$ and any $x\in\Z$ and forms a $Z$-collections. Finally, as explained, in the proof of Proposition~\ref{pro: assumptions theorem 1}, a similar bound is derived for $\langle u_x^2 \rangle$, yielding to the definition of a second $Z$-collection. As the sum of two $Z$-collections is still a $Z$-collection, this concludes the proof.

\subsection{Proof of Proposition~\ref{pro: assumptions theorem 3}}

Let again $x\in\Lambda_L$, 
and $g_x$ and $u_x$ be the functions that solve \eqref{eq: hyp 1 theo 2}. 
Here, we construct a $Z$-collection $(Z_I(x))_{and I,x}$ such that 
\[
    \sum_{y\in \Lambda_L} |\langle g_x g_y\rangle| +  |\langle u_x u_y\rangle|
    \; \le \; 
    Z_{\Lambda_L}(x), 
\]
which will allow to establish Proposition~\ref{pro: assumptions theorem 3}. 
Indeed, $(w_L(x, M))^2 \le w_L^-(x, M)$ for any $M > 0$, where the left-hand side is defined in \eqref{eq: weight green kubo} with the dependence on $L$ and $M$ made explicit, and the right-hand side is defined in \eqref{eq: def w-L}. The result then follows from the law of large numbers in Proposition~\ref{pro: lln for weight}.

Let us construct this $Z$-collection, following the blueprint in Section~\ref{sec: explicit Z collection} with the additional difficulty that we need to control the sum over $y\in\Lambda_L$. 
Let us start with the analog of \eqref{eq: 2nd form g k0 instance}: 
\begin{multline*}
    g_{x}
    \; = \; 
    \sum_{(s,t)\in \mathcal U_{n+1}} \sum_{\sigma}
    \sum_{k\in \mathcal R(\sigma)\cap\mathcal C(s,t,\sigma)} \\
    \frac{\imag\sigma_{t_1}\dots \sigma_{t_{n}}
    f_{x}(k'_1)\hat H_{\mathrm{an}}(k'_2)\dots \hat H_{\mathrm{an}}(k'_{n+1})}
    {\Delta(k'_1,\sigma'_1)
    \dots \Delta((k'_1,\dots,k'_{n}),(\sigma'_1,\dots,\sigma'_{n}))}
    \; (\tilde a_k^\sigma(s,t) - \langle\tilde a_k^\sigma (s,t)\rangle)
\end{multline*}
with $\mathcal C(s,t,\sigma)$ as defined in the proof of Proposition~\ref{pro: assumptions theorem 2} and $\hat f_x(k_1')$ defined in \eqref{eq: def f x k 1 prime}.
It is only here that subtracting the average of $\tilde a_k^\sigma$ becomes crucial, as the crude bound \eqref{eq: crude bound on the a} no longer suffices to control the sum over $y$, requiring us to rely on the decay of correlations in the Gibbs state.
The third item in Corollary~\ref{cor: decay of correlations} in Appendix~\ref{sec: gibbs} yields
\[
    |\langle \tilde a_k^\sigma(s,t) ; 
    \tilde a_{\underline k}^{\underline\sigma}(\underline s,\underline t)\rangle| 
    \; \le \; \Const
    \sum_{z,\underline z}
    \left|\psi_l(z) \psi_{\underline l} (\underline z) \right| \ed^{-c d(z,\underline z)}
\]
with $l=(l_1,\dots,l_p)$ and $\underline l = (\underline l_1,\dots, \underline l_p)$ the non-contracted indices in $k$ and $\underline k$ respectively, 
$z = (z_1,\dots,z_p)$, $\underline z = (\underline z_1,\dots,\underline z_p)$, 
\begin{equation}\label{eq: the non contracted psi}
\psi_l(z) \; = \; \psi_{l_1}(z_1)\dots \psi_{l_1}(z_p),
\qquad
\psi_{\underline l} (\underline z)
\; = \; \psi_{\underline l_1}(\underline z_1) \dots \psi_{\underline l_1}(\underline z_p)
\end{equation}
and $p=d_1 + 2n = 2(n+1)$. 
Hence the bound 
\begin{multline}\label{eq: 3d concrete Z collection}
    \sum_{y\in \Lambda_L} |\langle g_x g_y\rangle|
    \; \le \; 
    \Const
    \sum_{(s,t),(\underline s,\underline t)\in \mathcal U_{n+1}} 
    \sum_{\sigma,\underline\sigma}
    \sum_{\substack{k\in \mathcal R(\sigma)\cap\mathcal C(s,t,\sigma)\\
    \underline k\in \mathcal R(\underline\sigma)\cap\mathcal C(\underline s,\underline t,\underline\sigma)}} \\
    \frac{\sum_y|\hat f_{x}(k'_1)\hat H_{\an}(k_2')\dots \hat H_{\mathrm{an}}(k'_{n+1}) 
    \hat f_{y}(\underline k'_1)\hat H_{\an}(\underline k_2')\dots \hat H_{\mathrm{an}}(\underline k'_{n+1})|
    \sum_{z,\underline z}
    \left|\psi_l(z) \psi_{\underline l} (\underline z) \right| \ed^{-c d(z,\underline z)}}
    {|\Delta(k'_1,\sigma'_1)
    \dots \Delta((k'_1,\dots,k'_{n}),(\sigma'_1,\dots,\sigma'_{n}))
    \Delta(\underline k'_1,\underline\sigma'_1)
    \dots \Delta((\underline k'_1,\dots,\underline k'_{n}),(\underline\sigma'_1,\dots,\underline\sigma'_{n}))|}.
\end{multline}

As in Section~\ref{sec: explicit Z collection}, we note that the right-hand side of \eqref{eq: 3d concrete Z collection} keeps making sense for eigenpairs of $\mathcal H_I$ on any interval $I\subset\Z$ and any $x\in\Z$ and we show that it forms a $Z$-collections.
With straightforward adjustments, the identification of the various elements of the definition of a $Z$-collection proceeds as in Section~\ref{sec: explicit Z collection}. 
While the proofs of Properties~\eqref{eq: bounds on Dr} and \eqref{eq: condition 1 Z} carry over, proving Properties~\eqref{eq: condition 2 Z} and \eqref{eq: condition 3 Z} requires some attention:

\bigskip 

\textbf{Property~\eqref{eq: condition 2 Z}}. 
We rewrite the numerator as 
\begin{multline}\label{eq: expression for N k in in proof prop 3}
    N_{I,x,r}(\mathbf k)
    \; = \; 
    \sum_{z,\underline z} | \psi_l(z) \psi_{\underline l} (\underline z)|^{1/2}\\
    \sum_y
    \left| \hat f_{x}(k'_1)\hat H_{\an}(k_2')\dots \hat H_{\mathrm{an}}(k'_{n+1})\right|
    \left| \hat f_{y}(\underline k'_1)\hat H_{\an}(\underline k_2')\dots \hat H_{\mathrm{an}}(\underline k'_{n+1})\right|
    \left( |\psi_l(z) \psi_{\underline l} (\underline z)|^{1/2} \ed^{-c d(z,\underline z)}\right). 
\end{multline}
Using \eqref{eq: definition big a}, the first factor above is bounded as
\begin{equation}\label{eq: bound psi for sum over z}
     | \psi_l(z) \psi_{\underline l} (\underline z)|^{1/2}
     \; \le \; 
     \Const A_{I,x}^{2(n+1)/2} \left( 1 + (R(x,(l,\underline l)))^{8(n+1)/2} \right)
     \ed^{-c |z_1 - \loccen(l_1)| + \dots + |\underline z_{2n+1} - \loccen(\underline l_{2n+1})|}
\end{equation}
while the last one is bounded as
\begin{align}
     \left| \psi_l(z) \psi_{\underline l} (\underline z) \right|^{1/2} \ed^{-c d(z,\underline z)}
     \; &\le \; 
     A_{I,x}^{2(n+1)/2} \left( 1 + (R(x,(l,\underline l)))^{8(n+1)/2} \right)
     \ed^{-c \left(d(z,\loccen(l)) + d(z,\underline z) + d(\underline z,\loccen(\underline l))\right)} \nonumber\\
     \; &\le \; 
     A_{I,x}^{2(n+1)/2} \left( 1 + (R(x,(l,\underline l)))^{8(n+1)/2} \right) 
     \ed^{-c d(\loccen(l),\loccen(\underline l))}. 
     \label{eq: bound last factor in prop 6 12 after 8 10}
\end{align}
Hence, using the bound \eqref{eq: bound psi for sum over z} to perform the sum over $z,\underline z$ in \eqref{eq: expression for N k in in proof prop 3}, proceeding as in Section~\ref{sec: explicit Z collection} for the two intermediate factors in \eqref{eq: expression for N k in in proof prop 3}, and noting that $d(\loccen(l),\loccen(\underline l)) \ge d(\loccen(k),\loccen(\underline k))$, yields
\begin{equation}\label{eq: bound on N with many decay factors}
    N_{I,x,r}(\mathbf k) \; \le \;
    \Const 
    A_{I,x}^{p(n)} (1 + R(x,\mathbf k)^{4p(n)})
    \sum_y \ed^{-\frac{c}2 R(y,\underline k)}
    \left(\ed^{-c R(x,k)} \ed^{-\frac{c}2 R(y,\underline k)} 
    \ed^{-cd(\loccen(k),\loccen(\underline k))}\right)
\end{equation}
with $p(n) = 2(1+2n) + 2(n+1) = 4 + 6n$. Now, since
\[
    R(x,k) + R(y,\underline k) + d(\loccen(k),\loccen(\underline k))
    \;\ge\;
    \frac12 \left( R(x,k) + R(y,\underline k) + |x-y| \right)
    \; \ge \; \frac12 R(x,\mathbf k), 
\]
we conclude that 
\[
    N_{I,x,r}(\mathbf k) \; \le \;
    \Const 
    A_{I,x}^{p(n)} (1 + R(x,\mathbf k)^{4p(n)})
    \ed^{- c R(x,\mathbf k)},
\]
and this establishes Property~\eqref{eq: condition 2 Z}. 

\bigskip
\textbf{Property~\eqref{eq: condition 3 Z}}. 
An additional complication arises compared to the proof of Property~\eqref{eq: condition 3 Z} for the $Z$-collection discussed in Section~\ref{sec: explicit Z collection}: we must handle the summations over \( z, \underline{z}, y \) that appear in \eqref{eq: expression for N k in in proof prop 3}. 
To address this, instead of the generic bound \( |\psi(u) - \psi'(u)| \leq \|\psi - \psi'\| \) used there, we apply 
\begin{equation}\label{eq: advanced bound psi diference}
    |\psi(u) - \psi'(u)| 
    \leq 
    \|\psi - \psi'\|^{1/2} \big(|\psi(u)|^{1/2} + |\psi'(u)|^{1/2}\big),
\end{equation}
which results in an exponent \( \kappa_7 \) in Property~\eqref{eq: condition 3 Z} that will be two times smaller. 
By factoring out the norms of the differences of eigenvectors, we are left with a finite sum of expressions like \eqref{eq: expression for N k in in proof prop 3}, where each eigenvector now gains an additional power of \( 1/2 \), and part of the $\psi_j$ are replaced by $\psi'_j$. 
Since exponential decay as a function of \( R(x, \mathbf{k}) \) is no longer required, estimating the equivalent of \eqref{eq: expression for N k in in proof prop 3} becomes simpler: the last factor in \eqref{eq: expression for N k in in proof prop 3} can now be bounded by \( 1 \), instead of the bound in \eqref{eq: bound last factor in prop 6 12 after 8 10}. 
We note that this simplification proves necessary: we wouldn't have been able to establish the equivalent of \eqref{eq: bound last factor in prop 6 12 after 8 10}, because we have no obvious control on the distance between the localization center of a mode $\psi_j$ and the one of $\psi'_j$.
We end thus up with an expression of the type \eqref{eq: bound on N with many decay factors} where the last factor is absent, which suffices to prove Property~\eqref{eq: condition 3 Z}.

\subsection{Proof of Theorem~\ref{th: out of equilibrium}}
We detail the adaptations needed to complete the proof of Theorem~\ref{th: out of equilibrium}. 
The proof parallels that of Theorem~\ref{th: current}, that rests on Proposition~\ref{pro: assumptions theorem 2}.
Below we denote the time by $\theta$ to avoid confusion with the indices $s,t$ in use above. 

First, we modify the statement of Proposition~\ref{pro: assumptions theorem 2} to make it useful here. 
The set $B'(M)$ defined in \eqref{eq: 2d bad set}, in the assumptions of Proposition~\ref{pro: assumptions theorem 2}, 
needs to be replaced by the set 
\begin{equation}\label{eq: def tilde B prime M}
    \tilde B'(M) \; = \; 
    \{ x\in\Lambda_L : \exists \, {\theta} \ge 0 : \langle u_x^2 (\theta) \rangle_{\mathrm{ne}} \ge M \text{ or } \langle g_x^2 (\theta) \rangle_{\mathrm{ne}} \ge M\}. 
\end{equation}
Theorem~\ref{th: out of equilibrium} is then established using this modified version of Proposition~\ref{pro: assumptions theorem 2}, in the same way as Theorem~\ref{th: current} is derived from Proposition~\ref{pro: assumptions theorem 2} in Section~\ref{subsec: proof of Theorem 2}. 
Indeed, the only bound that is modified is \eqref{eq: j 0 t to be modified for last theorem}: $\langle H_x^2\rangle$, $\langle u_{\ell_0}^2\rangle$ and $\langle g_{\ell_0}^2 \rangle$ are changed into $\sup_\theta\langle H_x^2(\theta)\rangle$, $\sup_\theta\langle u_{\ell_0}^2(\theta)\rangle$ and $\sup_\theta\langle g_{\ell_0}^2 (\theta)\rangle$ respectively. 
The first of these suprema is bounded by a constant using Assumption~\ref{as: time evolved ensemble}, and the two other ones are bounded by $M$ using the definition of $\tilde B'(M)$ in \eqref{eq: def tilde B prime M}. 
This remainder of the proof carries over as such. 

Second, we need to prove the modified version of Proposition~\ref{pro: assumptions theorem 2}, where $B'(M)$ is replaced by $\tilde B'(M)$ defined in \eqref{eq: def tilde B prime M}. 
The proof parallels that in Section~\ref{subsec: proof prop 2}. 
We fix some $x\in\Lambda_L$ and we let $g_x$ and $u_x$ be the functions that solve \eqref{eq: hyp 1 theo 2}, 
and we construct a $Z$-collection $(Z_I(x))_{I,x}$ such that 
\[
    \sup_\theta\langle g_x^2 (\theta) \rangle + 
    \sup_\theta\langle u_x^2 (\theta) \rangle 
    \; \le \;
    Z_{\Lambda_L}(x),
\]
which, by the same argument as in Section~\ref{subsec: proof prop 2}, implies the modified version of Proposition~\ref{pro: assumptions theorem 2}. 

To construct this collection, we need to establish a bound on $\sup_\theta \langle g_x^2(\theta)\rangle$ analogous to \eqref{eq: 2nd concrete Z collection}. 
The derivation of this bound involved \eqref{eq: crude bound on the a}, which can no longer be used here. Instead, using \eqref{eq: new expansion a k sigma} together with Assumption~\ref{as: time evolved ensemble} yields
\begin{equation}\label{eq: modified crude bound on the a}
     |\langle \tilde a_k^\sigma(s,t;\theta)
     \tilde a_{\underline k}^{\underline\sigma}(\underline s,\underline t;\theta)\rangle| 
    \; \le \; \Const \sum_{z,\underline z} |\psi_l(z)\psi_{\underline l}(\underline z)|
\end{equation}
with $l=(l_1,\dots,l_p)$ and $\underline l = (\underline l_1,\dots, \underline l_p)$ the non-contracted indices in $k$ and $\underline k$ respectively, 
$z = (z_1,\dots,z_p)$, $\underline z = (\underline z_1,\dots,\underline z_p)$, and $\psi_l(z),\psi_{\underline l}(\underline z)$ as in \eqref{eq: the non contracted psi}.
First, we note that the left-hand side of \eqref{eq: modified crude bound on the a} now contains a simple product instead of a correlator, as in \eqref{eq: crude bound on the a}, because we cannot subtract the average here. Second, all the time dependence is in the left-hand side of \eqref{eq: modified crude bound on the a}, while the bound itself is time-independent.
Using \eqref{eq: modified crude bound on the a}, we arrive at a bound analogous to \eqref{eq: 2nd concrete Z collection}: 
\begin{multline}\label{eq: 4th concrete Z collection}
    \langle g_{x}^2\rangle
    \; \le \; 
    \Const
    \sum_{(s,t),(\underline s,\underline t)\in \mathcal U_{n+1}} 
    \sum_{\sigma,\underline\sigma}
    \sum_{\substack{k\in \mathcal R(\sigma)\cap\mathcal C(s,t,\sigma)\\
    \underline k\in \mathcal R(\underline\sigma)\cap\mathcal C(\underline s,\underline t,\underline\sigma)}} \\
    \frac{|\hat f_{x}(k'_1)\hat H_{\an}(k_2')\dots \hat H_{\mathrm{an}}(k'_{n+1}) 
    \hat f_{x}(\underline k'_1)\hat H_{\an}(\underline k_2')\dots \hat H_{\mathrm{an}}(\underline k'_{n+1})|
    \sum_{z,\underline z} |\psi_l(z)\psi_{\underline l}(\underline z)|}
    {|\Delta(k'_1,\sigma'_1)
    \dots \Delta((k'_1,\dots,k'_{n}),(\sigma'_1,\dots,\sigma'_{n}))
    \Delta(\underline k'_1,\underline\sigma'_1)
    \dots \Delta((\underline k'_1,\dots,\underline k'_{n}),(\underline\sigma'_1,\dots,\underline\sigma'_{n}))|}.
\end{multline}
From there on, the $Z$-collection can be constructed as in Section~\ref{subsec: proof prop 2}, using in addition the more advanced bound~\ref{eq: advanced bound psi diference} to control the sum over $z,\underline z$ when proving Property~\eqref{eq: condition 3 Z}.

\appendix
\section{Localization Estimates} \label{sec: localization}

\subsection{Between Global and Local Eigenvectors}
Here we show how eigenpairs of $\mathcal H_\Z$ can be approximated by eigenpairs of $\mathcal H_{\Lambda_\ell}$ and vice versa, for given $\ell \in \N^*$. 
To simplify notations, we will denote $\mathcal H_\Z$ by $\mathcal H$ and $\mathcal H_{\Lambda_\ell}$ by $\mathcal H_\ell$.
This can be generalized to the comparison of eigenpairs on two intervals $I\subset I'\subset\Z$ at the cost of more cumbersome expressions. 
We will consider functions in $L^2(\Lambda_\ell)$ as functions in $L^2(\Z)$ by extending them by 0 outside $\Lambda_\ell$.  
The main tool is the following lemma that we learned from \cite{elgart2016eigensystem}: 
\begin{Lemma}\label{lem: globaltolocal}\phantom{a}
\begin{enumerate}
\item
Let $(\psi,\nu^2)$ be an eigenpair of $\mathcal H$
and let $P_{\nu^2,a}=P_{\nu^2,a}(\mathcal H_{\ell})$ be the spectral projection on the interval $[\nu^2-a,\nu^2+a]$ of $\mathcal H_\ell$. 
There exists a constant $\Const$ such that 
$$
    \|(1-P_{\nu^2,a}) \psi \| \; \leq \; \frac{\Const}{a} \| \chi_{\Lambda_{\ell-1}^c}\psi \|.
$$
\item
Let $(\psi,\nu^2)$ be an eigenpair of $\mathcal H_{\ell}$ and
let now $P_{\nu^2,a}=P_{\nu^2,a}(\mathcal H)$ be the spectral projection on the interval $[\nu^2-a,\nu^2+a]$ of $\mathcal H$.
There exists a constant $\Const$ such that 
$$
    \|(1-P_{\nu^2,a}) \psi \| 
    \; \leq \; \frac{\Const}{a} \|\chi_{\partial\Lambda_\ell}\psi \|
$$
where $\partial \Lambda_\ell=  \{ x\in \Lambda_\ell:   \dist(x,\Z \setminus \Lambda_\ell) = 1)\}$.
\end{enumerate}
\end{Lemma}

\begin{proof}
We note that 
$$  
    \|(\mathcal H_{\ell}-\nu^2)\psi\| \;\leq\; \| (\mathcal H_{\ell}-\mathcal H) \psi \|  \; \le \; \Const \|\chi_{\Lambda_{\ell-1}^c}\psi\|.
$$
Furthermore, by spectral calculus,
$$
	\|(1-P_{\nu^2,a}) \psi \| 
	\;\leq\;  
	\left\|(1-P_{\nu^2,a}) \frac{1}{\mathcal H_{\ell}-\nu^2} \right\| \, \| (\mathcal H_{\ell}-\nu^2) \psi \|  
	\;\leq\; 
	\frac{\Const}{a} \| \chi_{\Lambda_{\ell-1}^c}\psi\|.
$$
The other item is proven similarly.
\end{proof}

Now we derive some more practical estimates.
We recall that the minimal level spacing $\Delta_\ell$ is defined in \eqref{eq: minimal level spacing}. 
Given $\gamma_1,\gamma_2>0$, let us define the event 
\begin{equation}\label{eq: definition s}
	S_\ell(\gamma_1,\gamma_2) \; = \; 
	\{\Delta_\ell \geq e^{-\gamma_1\ell}, \quad   A_\ell \leq e^{\gamma_2\ell}, \quad   A_\Z \leq e^{\gamma_2\ell}\}.
\end{equation}
We also introduce the exponent $\beta$ defined by
$$
	4\beta \; = \; \frac{1}{3\xi} -2\gamma_1-\gamma_2
$$
where $\xi$ is the localization length introduced in Section~\ref{sec: preliminary Anderson}.
We will assume that $\gamma_1$ and $\gamma_2$ are sufficiently small so that 
\begin{equation}\label{eq: condition on beta and gamma}
\beta \; > \; 0, \qquad \gamma_2  \; < \;  \frac{1}{12 \xi}.
\end{equation}
In the following lemmas, the expression ``for $\ell$ large enough" means that there exists a constant such that the statements hold true for all $\ell$ greater than this constant.
We will repeatedly use the bound \eqref{eq: definition big a} to estimate the tails of eigenvectors. 
We remind that $\lscal \cdot,\cdot\rscal$ denotes the scalar product on $L^2(\Z)$.

\begin{Lemma}\label{lem: from global to local eigenpairs}
Let $(\psi,\nu^2)$ be an eigenpair of $\mathcal H$ such that $|\loccen(\psi)| \leq (2/3)\ell$.
Assume that the event $S_\ell(\gamma_1,\gamma_2)$ holds.
For $\ell$ large enough, there exists a unique eigenpair $(\psi',(\nu^2)')$ of $\mathcal H_{\ell}$ such that 
$$
|\lscal\psi ,\psi'\rscal|^2 \;\geq\; 1- e^{-2\beta\ell}, \qquad |\nu^2-(\nu^2)'| \;\leq\; e^{-\beta \ell}.
$$ 
Moreover, if $(\psi_1,\nu_1^2)$ and $(\psi_2,\nu_2^2)$ are two distinct such eigenpairs of $\mathcal H$, then the corresponding eigenpairs  $(\psi_1',(\nu_1^2)')$ and $(\psi_2',(\nu_2^2)')$ of $\mathcal H_\ell$ are also distinct. 
\end{Lemma}
\begin{proof}
Assuming that  $S_\ell(\gamma_1,\gamma_2)$ holds,
we apply the first item of Lemma \ref{lem: globaltolocal} with $a=(1/4)e^{-\gamma_1\ell}$:
$$
	\|(1-P_{\nu^2,a}(\mathcal H_\ell))\psi \|^2 
	\;\leq\; 
	\frac{\Const}{a^2} A_\Z \ell^4 e^{-\ell/3\xi} 
	\; \leq \; 
	\Const \ell^4 e^{-4\beta\ell}.
$$
By the level spacing condition in $S_\ell(\gamma_1,\gamma_2)$, the spectral projection $P_{\nu^2,a}$ is one-dimensional 
and hence we get a unique eigenpair $(\psi',(\nu^2)')$ satisfying all requirements for $\ell$ large enough.
Indeed
$$
    1 - |\lscal\psi ,\psi'\rscal|^2
    \; = \; 
    \|(1-P_{\nu^2,a}(\mathcal H_\ell))\psi \|^2, 
$$
hence the first claim follows from the bound obtained above. 
To get the second bound, we observe that 
\[
    (\nu^2 - (\nu^2)' )\psi
    \; = \; 
    (\mathcal H - \mathcal H_\ell) \psi
     + (\mathcal H_\ell - (\nu^2)') (\psi - \psi').
\]
Hence, with a right choice of gauge for $\psi'$, we find 
\begin{align*}
    |\nu^2 - (\nu^2)'|
    \; &\le \; 
    \|(\mathcal H - \mathcal H_\ell)\psi\| + \Const \|\psi - \psi'\|
    \; \le \;
    \Const \| \chi_{\Lambda_{\ell - 1}^c} \psi \| + \Const  (1 - |\lscal\psi,\psi'\rscal|^2)^{1/2}\\
    \; &\le \; 
    \Const \left(\ell^4 A_\Z e^{-\ell/3\xi}\right)^{1/2} + \Const \ell^2 e^{-2 \beta \ell}
    \; \le \; e^{-\beta \ell}
\end{align*}
for $\ell$ large enough, which proves the first part of the lemma.  
The second part follows from the fact that, if $\psi_1'$ and $\psi_2'$ were not distinct, then the overlap $|\lscal\psi_1,\psi_2\rscal|$ would be close to one, which is impossible since $\psi_1$ and $\psi_2$ ar orthogonal. 
\end{proof}

\begin{Lemma}\label{lem: from local to global eigenpairs}
Let $(\psi,\nu^2)$ be an eigenpair of $\mathcal H_\ell$ such that $|\loccen(\psi)| \leq \ell/2$.  
Assume that the event $S_\ell(\gamma_1,\gamma_2)$ holds. 
For $\ell$ large enough, there exists a unique eigenpair $(\psi',(\nu^2)')$ of $\mathcal H$ such that 
$$
    |\lscal\psi,\psi'\rscal|^2 
    \;\geq\; 
    1-\ed^{-2\beta\ell}, \qquad |\nu^2-(\nu^2)'| 
    \;\leq\; 
    \ed^{-\beta \ell}, 
    \qquad 
    |\loccen(\psi')| \leq 2\ell/3.
$$
Moreover, if $(\psi_1,\nu_1^2)$ and $(\psi_2,\nu_2^2)$ are two distinct such eigenpairs of $\mathcal H_\ell$, then the corresponding eigenpairs  $(\psi_1',(\nu_1^2)')$ and $(\psi_2',(\nu_2^2)')$ of $\mathcal H$ are also distinct. 
\end{Lemma}

\begin{proof}
Assume that  $S_\ell(\gamma_1,\gamma_2)$ holds, and let $a=(1/4)e^{-\gamma_1\ell}$. 
Let also $((\psi'_j, (\nu_j^2)'))_j$ be the eigenpairs of $\mathcal H$ such that  $(\nu_j^2)' \in [\nu^2-a,\nu^2+a]$. 

Let us first show that, for $\ell$ large enough, there is at least one such eigenpair with $|\loccen(\psi'_j)| \leq (2/3)\ell$. 
Indeed, from Lemma \ref{lem: globaltolocal} with $a$ as above, we find 
$$
	\| (1-P_{\nu^2,a}(\mathcal H))\psi \|^2 
	\;\leq\; 
	\frac{\Const}{a^2} \ell^4 A_\ell e^{-\ell/2\xi} 
	\; \leq \;  
	\ed^{-2\beta\ell}
$$
for $\ell$ large enough. 
Let us write 
\begin{equation}\label{eq: normalization}
	\| P_{\nu^2,a}(\mathcal H)\psi \|^2 \; = \; \sum_{j}| \lscal\psi'_j , \psi\rscal|^2.
\end{equation}
The contribution to this sum from eigenpairs $(\psi'_j,(\nu_j^2)')$ with $|\loccen(\psi'_j)| > (2/3)\ell$ is bounded as
$$
    \sum_{j: |\loccen(\psi'_j)| > (2/3)\ell} |\lscal\psi'_j,\psi\rscal|^2 
    \; \le \;
    \sum_{j: |\loccen(\psi'_j)| > (2/3)\ell}
    \Const \ell^{10} A_\Z A_\ell e^{-\ell/6\xi} 
    \; \le \; 
    \Const \ell^{13} e^{-(1/6\xi - 2 \gamma_2) \ell}, 
$$
where we have used that the number of eigenfunctions with given localization center grows at most quadratically with the distance of the localization center to origin.
Thanks to the assumption $\gamma_2 < 1/12\xi$, this tends to $0$ as $\ell\to\infty$.
Since the left hand side of \eqref{eq: normalization} is close to $1$ for $\ell$ large enough, we conclude that there exists  one eigenpair $(\psi'_j, (\nu_j^2)')$ such that $|\loccen(\psi'_j)| \leq (2/3)\ell$.

Lemma~\ref{lem: from global to local eigenpairs} applies to this eigenpair: 
For $\ell$ large enough, we find a corresponding eigenpair $((\nu_j^2)'',\psi''_j)$ for $\mathcal H_\ell$. 
Looking into the proof of Lemma~\ref{lem: from global to local eigenpairs} reveals that $|(\nu^2_j)'-(\nu^2_j)''|\leq a$. 
Hence, 
thanks to the level spacing condition in the event $S_\ell(\gamma_1,\gamma_2)$, the eignpair $((\nu_j^2)'',\psi''_j)$ is  equal to $(\psi,\nu^2)$. 
Moreover, since for large enough $\ell$, the overlap $|\lscal \psi,\psi'_j\rscal|$ is close to 1, there can be at most one such eigenpair $(\psi'_j,(\nu_j^2)')$.
We conclude using Lemma~\ref{lem: from global to local eigenpairs}.

The second part of the lemma is shown as the second part of Lemma~\ref{lem: from global to local eigenpairs}.
\end{proof}

Finally, we combine the two lemmas above to establish a bijection between local and global eigenvectors on a suitably chosen set. 
Assume that 
the event $S_\ell(\gamma_1,\gamma_2)$ holds and define the sets
    \begin{align*}
        X \; &= \; \{ (\psi,\nu^2) \text{ eigenpair of }\mathcal H_{\ell}: |\loccen(\psi)| \leq \ell/2 \}, \\
	\qquad 
	Y \; &= \; \{(\psi',(\nu^2)') \text{ eigenpair of }\mathcal H: |\loccen(\psi')| \leq 2\ell/3 \}.
    \end{align*}
    For $\ell$ large enough, there exists a bijection $\Phi : X \to \Phi(X)\subset Y$ such that for any $(\psi,\nu^2)\in X$, $\Phi(\psi,\nu^2)$ is the eigenpair provided by Lemma~\ref{lem: from local to global eigenpairs}, and for any $(\psi',(\nu^2)')\in \Phi(X)$, $\Phi^{-1}(\psi',(\nu^2)')$ is the eigenpair furnished by Lemma~\ref{lem: from global to local eigenpairs}. 
    Moreover
\begin{Lemma}\label{lem: bijection local global}
    If $(\psi',(\nu^2)') \in Y\backslash \Phi(X)$, 
    $$
	|\loccen(\psi')| \;\ge\; \ell/4.  
    $$
\end{Lemma}
\begin{proof}
On the one hand, by Lemma~\ref{lem: from global to local eigenpairs}, there exists an eigenpair $(\psi'', (\nu^2)'')$ of $\mathcal H_\ell$ so that 
\begin{equation}\label{eq: one hand}
	|\lscal\psi'' ,\psi'\rscal|^2 \;\geq\; 1- \ed^{-2\beta\ell}, \qquad |(\nu^2)''-(\nu^2)'| \;\leq\; \ed^{-\beta \ell}.
\end{equation}
On the other hand, we get the bound 
$$
    |\lscal \psi' , \psi''\rscal|^2
    \; \le \; 
    \Const A_\Z A_\ell \ell^{10} e^{-|\loccen(\psi'')-\loccen(\psi')|/\xi}
    \; \le \; 
    \Const \ell^{10} e^{2\gamma_2 \ell - |\loccen(\psi'')-\loccen(\psi')|/\xi}.
$$
We know that $|\loccen(\psi'')|>\ell/2$, as otherwise $(\psi',(\nu^2)')$ would belong to $\Phi(X)$. Hence, if $|\loccen(\psi')|<\ell/4$, we would have 
$$
    |\lscal \psi' , \psi''\rscal|^2
    \; \le \; 
    \Const \ell^{10} e^{-\ell(1/4\xi - 2\gamma_2)}
$$
which converges to $0$ as $\ell\to \infty$ since we assume $\gamma_2\le 1/12\xi$, 
and this would imply a contradiction with \eqref{eq: one hand}.
\end{proof}

\section{Decay of Correlation in the Gibbs State} \label{sec: gibbs}
We state the needed results on the decay of correlations in the Gibbs state. 
Given a function $u$ on $\R^{2|\Lambda_L|}$, we denote by $\mathrm{supp}(u) \subset \Lambda_L$ the smallest set of points such that $u$ is constant as a function of 
$(q_x, p_x)$ if $x \notin \mathrm{supp}(u)$ (it is thus not the usual support of a function on $\R^{2|\Lambda_L|}$). 
Given two functions $u,v$ on $\R^{2|\Lambda_L|}$, we denote by $d(u,v)$ the distance between $\mathrm{supp}(u)$ and $\mathrm{supp}(v)$.
\begin{Proposition}\label{pro: decay of correlations}
	There exist constants $\Const,c$ such that, 
	for any polynomials on $\R^{2|\Lambda_L|}$ satisfying $\langle u \rangle = \langle v \rangle = 0$, it holds that 
	$$
		|\langle uv \rangle| \; \le \; \Const \ed^{- c d(u,v)} \langle (\nabla u)^2\rangle^{1/2}\langle (\nabla v)^2\rangle^{1/2} \qquad a.s.
	$$
\end{Proposition} 
\begin{proof}
See Lemma~2 in \cite{WDR_Huveneers_Olla}. 
See also the original references \cite{HELFFER_1998}\cite{HELFFER_1999}\cite{bodineau_helffer_2000}.
\end{proof}

We will exploit this result through the following corollary. 
We will make use of the following notations: 
Given $m\ge 1$, given $k \in \{1,\dots, |\Lambda_L|\}^m$ and given $\sigma \in \{\pm\}^m$, 
let $a_k^\sigma = a_{k_1}^{\sigma_1} \dots a_{k_m}^{\sigma_m}$, with $a_{k_j}^{\sigma_j}$ defined in \eqref{eq: pqa transform} for $1\le j \le m$.
Given $x\in \Lambda_L^m$, let also $\psi_k(x) = \psi_{k_1}(x_1)\dots\psi_{k_m}(x_m)$, where $\psi_{k_j}$ is an eigenvector of $\mathcal H_{\Lambda_L}$. 
\begin{Corollary}\label{cor: decay of correlations}
Let $m\ge 1$. There exist constants $\Const_m,c$ such that  
\begin{enumerate}
	\item 
	if $f$ is a monomial in $(q_x,p_x)_{x\in \Lambda_L}$ of degree at most $m$, then 
	$$
		|\langle f \rangle| \; \le \; \Const_m, 
	$$
	\item 
	for all $k \in \{1,\dots, |\Lambda_L|\}^m$ and $\sigma \in \{\pm\}^m$, 
	$$
		|\langle a_k^\sigma \rangle| \; \le \; \Const_m, 
	$$
	and in particular $\langle E_k;E_k\rangle \le \Const_2$, 
	\item 
	for all $k = (k_1,\dots,k_{m})\in \N$ and $k'=(k_1',\dots,k_{m'}')$, with $m'\le m$, 
	$$
		|\langle a_k^\sigma;a_{k'}^{\sigma'}\rangle| \; \le \; \Const_m \sum_{x,x'}|\psi_k(x)\psi_{k'}(x')| \ed^{-c d(x,x')}.
	$$
\end{enumerate}
\end{Corollary}
\begin{proof}
To simplify the notation, we will not explicitly indicate the dependence of the constants on $m$.
Let us start with item 1. 
Since the Gibbs measure is a product measure in the $p$ variables, 
and since $\langle q_x^{m}\rangle = 0$ for $m$ odd, it is enough to get a bound on $\langle q_x^{m}\rangle$ with $m$ even ad $x\in \Lambda_L$. 
For $m=2$, we deduce from Proposition~\ref{pro: decay of correlations} that 
$$
		\langle q_x^2 \rangle \; = \;  \langle q_x q_x \rangle \; \le \; \Const.
$$
Let us next assume that the result holds for $m\ge 2$, and let us prove it for $m+2$: 
$$
	\langle q_x^{m+2}\rangle 
	\; = \; 
	\langle q_x^{(m+2)/2};q_x^{(m+2)/2}\rangle + \langle q_x^{(m+2)/2}\rangle^2 
	\; \le \; 
	\Const \langle (\nabla q_x^{(m+2)/2})^2 \rangle + \langle q_x^{(m+2)/2}\rangle^2 
	\; \le \; \Const
$$
where the first bound follows from Proposition~\ref{pro: decay of correlations}, and the second from our inductive hypothesis. 

We prove item 2 in a similar way. 
The definition \eqref{eq: pqa transform} may be rewritten as 
\begin{equation}\label{eq: new expansion a k sigma}
	a_{k_j}^{\sigma_j}
	\; = \; 
	\sum_{x\in \Lambda_L} \psi_{k_j}(x) a_{k_j}^{\sigma_j}(x)
	\qquad \text{with} \qquad 
	a_{k_j}^{\sigma_j}(x) \; := \; \frac1{\sqrt 2}\left(\nu_{k_j}^{1/2} q_x - \frac{\mathrm{i}\sigma_{k_j}}{\nu_k^{1/2}}p_x\right).
\end{equation}
Again, it is enough to prove the result for $m$ even. 
For $m=2$, we deduce from Proposition~\ref{pro: decay of correlations} that
$$
	\langle a_{k_1}^{\sigma_1} a_{k_2}^{\sigma_2} \rangle 
	\; \le \; 
	\Const \langle (\nabla a_{k_1}^{\sigma_1})^2 \rangle^{1/2} \langle (\nabla a_{k_2}^{\sigma_2})^2 \rangle^{1/2}
	\; \le \; 
	\Const 
$$
where the last bound follows by applying the first item of this corrolary, and from the normalization of the eigenvectors $\psi_{k_1},\psi_{k_2}$.
Let us next assume that the result holds for $m\ge 2$ and let us prove it for $m+2$. 
We decompose
$$
	\langle a_{k_1}^{\sigma_1} \dots a_{k_{m+2}}^{\sigma_{m+2}}\rangle 
	\; = \; 
	\langle a_{k_1}^{\sigma_1} \dots a_{k_{m/2}+1}^{\sigma_{m/2+1}} ; a_{k_{m/2}+2}^{\sigma_{m/2+2}} \dots a_{k_{m+2}}^{\sigma_{m+2}}\rangle 
	+ 
	\langle a_{k_1}^{\sigma_1} \dots a_{k_{m/2}+1}^{\sigma_{m/2+1}}\rangle\langle a_{k_{m/2}+2}^{\sigma_{m/2+2}} \dots a_{k_{m+2}}^{\sigma_{m+2}}\rangle .
$$
The second term in the right hand side is bounded by the inductive hypothesis. For the first one, we apply Proposition~\ref{pro: decay of correlations} and get the bound
$$
	\Const \langle (\nabla a_{k_1}^{\sigma_1} \dots a_{k_{m/2}+1}^{\sigma_{m/2+1}})^2\rangle^{1/2}  
	\langle (\nabla a_{k_{m/2}+2}^{\sigma_{m/2+2}} \dots a_{k_{m+2}}^{\sigma_{m+2}})^2\rangle^{1/2}.
$$
Expanding the gradients according to Leibniz rule and using again the normalization of the wave functions, 
we deduce that this is bounded by a constant thanks to our inductive hypothesis.    

Item 3 is shown by expanding $a_k^\sigma$  and $a_{k'}^{\sigma'}$ as in \eqref{eq: new expansion a k sigma},
and then using Proposition~\ref{pro: decay of correlations} and item 1 of this corrolary. 
\end{proof}

\bibliographystyle{plain}
\bibliography{bibilography}
\end{document}